%% file: twenty_question_r5.tex
\documentclass[onecolumn,10pt]{IEEEtran}
\input{preamble}

\usepackage[utf8]{inputenc} 
\usepackage[T1]{fontenc}    
\usepackage{url,bbm}            
\usepackage{booktabs}       
\usepackage{amsfonts,algorithm,algorithmic}       
\usepackage{nicefrac}       
\usepackage{microtype}      

\usepackage{cite}
\usepackage{subfigure,transparent,color,graphicx} 

\graphicspath{{./figures2/}}
\allowdisplaybreaks[2]
\newcommand{\blue}[1]{\textcolor{blue}{#1}} 
\newcommand{\bbo}{\mathbbm{1}}

\begin{document}

\title{Resolution Limits for the Noisy Non-Adaptive 20 Questions Problem}
\author{Lin Zhou and Alfred Hero \\
\thanks{The authors are with the Department of Electrical Engineering and Computer Science, University of Michigan, Ann Arbor, MI, USA, 48109-2122 (Emails: linzhou@umich.edu and hero@eecs.umich.edu.). This work was partially supported by ARO grant W911NF-15-1-0479.}
\thanks{A preliminary version of this paper will be presented at ISIT 2020.}
}
\maketitle

\begin{abstract}
We establish fundamental limits on estimation accuracy for the noisy 20 questions problem with measurement-dependent noise and introduce optimal non-adaptive procedures that achieve these limits. The minimal achievable resolution is defined as the absolute difference between the estimated and the true locations of a target over a unit cube, given a finite number of queries constrained by the excess-resolution probability. Inspired by the relationship between the 20 questions problem and the channel coding problem, we derive non-asymptotic bounds on the minimal achievable resolution to estimate the target location. Furthermore, applying the Berry--Esseen theorem to our non-asymptotic bounds, we obtain a second-order asymptotic approximation to the achievable resolution of optimal non-adaptive query procedures with a finite number of queries subject to the excess-resolution probability constraint. We specialize our second-order results to measurement-dependent versions of several channel models including the binary symmetric, the binary erasure and the binary Z- channels. The theory is extended to simultaneous searching for multiple targets. As a complement, we establish a second-order asymptotic achievability bound for adaptive querying and use this to bound the benefit of adaptive querying.
\end{abstract}


\section{Introduction}
The noisy 20 questions problem (cf. \cite{renyi1961problem,burnashev1974interval,ulam1991adventures,pelc2002searching,jedynak2012twenty,chung2018unequal,lalitha2018improved}) arises when one aims to accurately estimate an arbitrarily distributed random variable $S$ by successively querying an oracle and using noisy responses to form an estimate $\hatS$. A central goal in this problem is to find optimal query strategies that yield a good estimate $\hatS$ of the unknown target $S$.

Depending on the framework adopted to design queries, the 20 questions problem is either adaptive or non-adaptive. In adaptive query procedures, the design of a subsequent query depends on all previous queries and noisy responses to these queries from the oracle. In non-adaptive query procedures, all the queries are designed independently in advance. For example, the bisection policy~\cite[Section 4.1]{jedynak2012twenty} is an adaptive query procedure and the dyadic policy~\cite[Section 4.2]{jedynak2012twenty} is a non-adaptive query procedure. Compared with adaptive query procedures, non-adaptive query procedures have the advantage of lower computation cost, parallelizability and no need for feedback. Depending on whether or not \blue{the noisy channel used to corrupt the noiseless responses} depends on the queries, the noisy 20 questions problem is classified into two categories: querying with measurement-independent noise (e.g.,~\cite{jedynak2012twenty,chung2018unequal}); and querying with measurement-dependent noise (e.g.,~\cite{kaspi2018searching,lalitha2018improved}). As argued in \cite{kaspi2018searching}, measurement-dependent noise can be a better model in many practical applications. For example, for target localization with a sensor network, the noisy response to each query can depend on the size of the query region. Another example is in human query systems where personal biases abut the state may affect the response.

In earlier works on the noisy 20 questions problem, e.g.,~\cite{jedynak2012twenty,tsiligkaridis2014collaborative,tsiligkaridis2015decentralized}, the queries were designed to minimize the entropy of the posterior distribution of the target variable $S$. As pointed out in later works, e.g., \cite{chung2018unequal,chiu2016sequential,kaspi2018searching,lalitha2018improved,chung2017bounds}, other accuracy measures, such as the resolution and the quadratic loss are often better criteria, where the resolution is defined as the absolute difference between $S$ and its estimate $\hatS$, $|\hatS-S|$, and the quadratic loss is $(\hatS-S)^2$. \blue{In particular, in estimation problems, if one aims to minimize the differential entropy of the posterior uncertainty of the target variable, then any two queries which can reduce the entropy by the same amount are deemed equally important, even if one query achieves higher estimation accuracy. For example, one query might ask about the most significant bit of the binary expansion of the target variable while the other query might ask about a much less significant bit. These two queries induce equal reductions in the entropy of the posterior distribution. By using the resolution or the quadratic loss, which are directly related with the estimation error, to drive the query design, such ambiguity is avoided. Relations between resolution and entropy were quantified by the bounds in \cite[Theorem 1]{chung2017bounds}.}

\subsection{Our Contributions}
Motivated by the scenario of limited resources, computation and response time, we obtain new results on the non-asymptotic tradeoff among the number of queries $n$, the achievable resolution $\delta$ and the excess-resolution probability $\varepsilon$ of optimal adaptive and non-adaptive query procedures for the following noisy 20 questions problem: 
\blue{estimation of the location of a target random vector $\bS=(S_1,\ldots,S_d)$ with arbitrary distribution on the unit cube of dimension $d$, i.e., $[0,1]^d$.} For the case of adaptive query procedures, we derive achievable second-order asymptotic bounds on the optimal resolution. We define the benefit of adaptivity, called adaptivity gain, as the logarithm of the ratio between achievable resolutions of optimal non-adaptive and adaptive query procedures. This benefit of adaptivity can be attributed to the more informative number of bits extracted by optimal adaptive querying in the binary expansion of each dimension of the target variable. We numerically evaluate a lower bound on the adaptivity gain for measurement-dependent versions of binary symmetric, binary erasure and binary Z- channels.

\blue{Our main focus of this paper is on non-adaptive query procedures.} Our contributions for the case of non-adaptive querying are as follows. Firstly, we derive non-asymptotic resolution bounds of optimal non-adaptive query procedures for arbitrary number of queries $n$ and any excess-resolution probability $\varepsilon$. To do so, \blue{similarly to} \cite{kaspi2018searching}, we exploit the connection between the 20 questions problem and the channel coding problem. This allows us to borrow ideas from finite blocklength analyses for channel coding~\cite{polyanskiy2010finite} (see also \cite{TanBook}). In particular, we adopt the change-of-measure technique of \cite{csiszar2011information} in the achievability proof to handle the case of measurement-dependent noise.

Secondly, applying the Berry-Esseen theorem, under mild conditions on the measurement-dependent noise, we obtain a second-order asymptotic approximation to the achievable resolution of optimal non-adaptive query procedures with finite number of queries. \blue{A key implication of our result states that searching over each dimension of a multidimensional target separately, although asymptotically optimal, is \emph{not} optimal for a finite number of queries.} As a corollary, we establish a phase transition for optimal non-adaptive query procedures. This implies that, if one is allowed to make an infinite number of optimal non-adaptive queries, regardless of the excess-resolution probability, the \blue{asymptotic} average number of bits (in the binary expansion of each dimension of the target variable) extracted per query remains the same. \blue{Furthermore, we show how to extend our theory to simultaneous searching for multiple targets~\cite{kaspi2015searching} over the unit cube.}

We specialize our second-order analyses to three measurement-dependent channel models: the binary symmetric, the binary erasure and the binary Z- channels. Similarly to our proofs for measurement-dependent channels, the second-order asymptotic approximation to the achievable resolution of optimal non-adaptive query procedures for measurement-independent channels is obtained. To compare the performances of optimal non-adaptive query procedures for measurement-dependent and measurement-independent channels, we contrast the minimal achievable resolution of both scenarios. \blue{In our definitions (Definitions \ref{def:mdBSC} to \ref{def:mdZ}) of measurement-dependent channels, given a fixed channel noise parameter, the noise level is always smaller for a measurement-dependent channel.} Intuitively, the optimal non-adaptive query procedure for a measurement-dependent channel should have a higher resolution compared to its measurement-independent counterpart. We verify this intuition for the asymmetric binary erasure and binary Z channels. However, for the binary symmetric channel, we find that the achievable resolution of optimal non-adaptive query procedures for the measurement-independent channel can in fact achieve a higher resolution when the crossover probability is large, \blue{i.e., $>0.5$.} We provide plausible explanations for this counter-intuitive phenomenon.

{\color{blue}
\subsection{Comparison to Previous Work}
\label{sec:comp}
Here we compare the contributions of our paper to related work in the literature~\cite{kaspi2018searching,chiu2016sequential}. First of all, our results hold for arbitrary discrete channels, while the results in \cite{kaspi2018searching,chiu2016sequential} were only established for a measurement-dependent BSC. Furthermore, we consider a multidimensional target while the the results in \cite{kaspi2018searching,chiu2016sequential} were only established for a one-dimensional target. In the following, we compare our results, specialized to a single one-dimensional target, with the results in \cite{kaspi2018searching,chiu2016sequential}.

In terms of our results on resolution of non-adaptive query schemes, the most closely related work is \cite{kaspi2018searching}. The authors in \cite{kaspi2018searching} derived first-order asymptotic characterizations of the resolution decay rate when the excess-resolution probability vanishes for a measurement-dependent BSC. Our results in Theorem \ref{result:second} extends \cite[Theorem 1]{kaspi2018searching} in several directions. First, Theorem \ref{result:second} is a  second-order asymptotic result which provides an approximation to the performance of optimal query procedures employing a finite number of queries, while \cite[Theorem 1]{kaspi2018searching} only characterizes the asymptotic performance when the number of queries tends to infinity. Second, our results hold for any measurement-dependent channel satisfying a mild condition while \cite[Theorem 1]{kaspi2018searching} only considers the measurement-dependent binary symmetric channel (cf. Definition \ref{def:mdBSC}). Furthermore, our results apply methods recently developed for finite blocklength information theory. This results in the first non-asymptotic bounds (cf. Theorems \ref{ach:fbl} and \ref{fbl:converse}) for non-adaptive query schemes for 20 questions search. These bounds extend the analysis of \cite{kaspi2018searching}, in which the derived lower bound on the decay rate of the excess-resolution probability is only tight in the asymptotic limit of large $n$ (e.g., infinite number of queries). Other works concerning non-adaptive query procedures~\cite{jedynak2012twenty,pelc2002searching,chung2018unequal} consider either different performance criteria or different models and thus not comparable to our work.

For resolution limits of adaptive querying, the most closely related publications are \cite{chiu2016sequential,kaspi2018searching,lalitha2018improved}. The authors in \cite{lalitha2018improved} considered the case where the noise is measurement-dependent Gaussian noise. However, the setting in \cite{lalitha2018improved} is different from ours. In \cite{kaspi2018searching}, the authors considered two adaptive query procedures using ideas due to Forney \cite{forney1968exponential} and Yamamoto-Itoh~\cite{yamamoto1979asymptotic}. The authors of \cite{kaspi2018searching} derived a lower bound on the exponent of the excess-resolution probability for both procedures and showed that the performance of the three-stage adaptive query procedure based on Yamamoto-Itoh~\cite{yamamoto1979asymptotic} has better performance. Furthermore, in \cite{kaspi2018searching}, an asymptotic upper bound on the average number of queries is derived given a particular target resolution and excess-resolution probability~\cite[Theorem 2]{kaspi2018searching}. In \cite{chiu2016sequential}, the authors proposed a single-stage adaptive query procedure using sorted posterior matching and derived a non-asymptotic upper bound on the average number of queries subject to a constraint on the excess-resolution probability with respect to a given resolution. In contrast, we present a single-stage adaptive query procedure using the ideas in \cite{polyanskiy2011feedback} on finite blocklength analysis for channel coding with feedback and we derive a non-asymptotic upper bound on the excess-resolution probability with respect to a given resolution subject to a constraint on the average number of queries (cf. Theorem \ref{fbl:ach:adaptive}).

It would be interesting to compare our non-asymptotic achievability bound in Theorem \ref{fbl:ach:adaptive} to the results in \cite[Theorem 2]{kaspi2018searching}, \cite{chiu2016sequential}. However, since the results are derived under different theoretical assumptions, an analytical comparison is challenging. Note that Theorem \ref{fbl:ach:adaptive} addresses the decay rate of the achievable resolution subject to constraints on the average number of queries and an excess-resolution probability. In contrast, the results in \cite{kaspi2018searching,chiu2016sequential} address an upper bound on the average number of queries subject to a given resolution and an excess-resolution probability constraint. It is difficult to transform our results to an upper bound on the average number of queries or to transform their results to a lower bound on the decay rate of achievable resolution. Instead, we compare the asymptotic resolution decay rate of our adaptive query algorithm and the algorithms in \cite[Theorem 2]{kaspi2018searching} and \cite{chiu2016sequential}. Asymptotically, the three-stage algorithm in \cite[Theorem 2]{kaspi2018searching} and the single-stage posterior matching algorithm in \cite{chiu2016sequential} achieve the same asymptotic resolution rate. Compared with our proposed adaptive query procedure, unless the values of excess-resolution probability are small or the channel noise is high, the asymptotic performances of \cite{kaspi2018searching,chiu2016sequential} can be worse than our proposed adaptive query procedure ((See the remarks associated with Theorem \ref{second:fbl:adaptive} for details)). Numerical simulation results (Figure \ref{sim_adap}) are presented to compare the non-asymptotic performance of our proposed algorithm and the sorted posterior matching algorithm in \cite{chiu2016sequential} for estimation of a uniformly distributed one-dimensional target variable.}

\section{Problem Formulation}
\subsection*{Notation}
Random variables and their realizations are denoted by upper case variables (e.g.,  $X$) and lower case variables (e.g.,  $x$), respectively. All sets are denoted in calligraphic font (e.g.,  $\mathcal{X}$). Let $X^n:=(X_1,\ldots,X_n)$ be a random vector of length $n$. We use $\Phi^{-1}(\cdot)$ to denote the inverse of the cumulative distribution function (cdf) of the standard Gaussian. We use $\bbR$, $\bbR_+$ and $\bbN$ to denote the sets of real numbers, positive real numbers and integers respectively. Given any two integers $(m,n)\in\bbN^2$, we use $[m:n]$ to denote the set of integers $\{m,m+1,\ldots,n\}$ and use $[m]$ to denote $[1:m]$. \blue{Given any $(m,n)\in\bbN^2$, for any $m$ by $n$ matrix $\ba=\{a_{i,j}\}_{i\in[m],j\in[n]}$, the infinity norm is defined as $\|\ba\|_{\infty}:=\max_{i\in[m],j\in[n]}|a_{i,j}|$.} The set of all probability distributions on a finite set $\calX$ is denoted as $\calP(\calX)$ and the set of all conditional probability distributions from $\calX$ to $\calY$ is denoted as $\calP(\calY|\calX)$. Furthermore, we use $\calF(\calS)$ to denote the set of all probability density functions on a set $\calS$. All logarithms are base $e$ unless otherwise noted. Finally, we use $\bbo()$ to denote the indicator function.

\subsection{Noisy 20 Questions Problem On the Unit Cube}
Consider an arbitrary integer $d\in\bbN$. Let $\bS=(S_1,\ldots,S_d)$ be a continuous random vector defined on the unit cube of dimensional $d$ (i.e., $[0,1]^d$) with arbitrary probability density function (pdf) $f_{\bS}$. Note that any searching problem over a bounded $d$-dimensional region is equivalent to a searching problem over the unit cube of dimension $d$ with normalization in each dimension. 

In the estimation problem formulated under the framework of noisy 20 questions, a player aims to accurately estimate the target random variable $\bS$ by posing a sequence of queries $\calA^n=(\calA_1,\ldots,\calA_n)\subseteq[0,1]^{nd}$ to an oracle knowing $\bS$. After receiving the queries, the oracle finds binary answers $\{X_i=\bbo(\bS\in\calA_i)\}_{i\in[n]}$ and passes these answers through a measurement-dependent channel with transition matrix $P_{Y^n|X^n}^{\calA^n}\in\calP(\calY^n|\{0,1\}^n)$ yielding noisy responses $Y^n=(Y_1,\ldots,Y_n)$. Given the noisy responses $Y^n$, the player uses a decoding function $g:\calY^n\to[0,1]^d$ to obtain an estimate $\hat{\bS}=(\hatS_1,\ldots,\hatS)$ of the target variable $\bS=(S_1,\ldots,S_d)$. Throughout the paper, we assume that the alphabet $\calY$ for the noisy response is finite.

A query procedure for the noisy 20 questions problem consists of the Lebesgue measurable query sets $\calA^n\subseteq[0,1]^{nd}$ and a decoder $g:\calY^n\to[0,1]^d$. In general, these procedures can be classified into two categories: non-adaptive and adaptive querying. In a non-adaptive query procedure, the player needs to first determine the number of queries $n$ and then design all the queries $\calA^n$ simultaneously. In contrast, in an adaptive query procedure, the design of queries is done sequentially and the number of queries is a variable. In particular, when designing the $i$-th query, the player can use the previous queries and the noisy responses from the oracle to these queries, i.e., \blue{$\{\calA_j,Y_j\}_{j\in[i-1]}$}, to formulate the next query $\calA_i$. Furthermore, the player needs to choose a stopping criterion, which may be random, determining the number of queries to make. 

We illustrate the difference between non-adaptive and adaptive query procedures in Figure \ref{illustrate:procedures}. In subsequent sections, we clarify the notion of the measurement-dependent channel with concrete examples and present specific definitions of non-adaptive and adaptive query procedures.
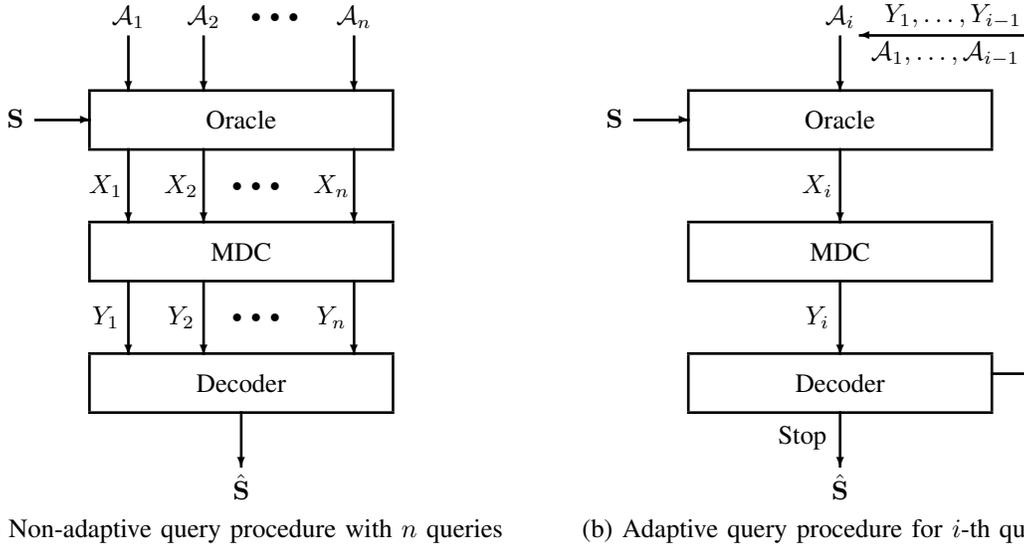
\begin{figure}[tb]
\centering
\setlength{\unitlength}{0.5cm}
\begin{tabular}{cc}
\scalebox{1}{
\begin{picture}(8,12.5)
\linethickness{1pt}
\put(1,12.5){\makebox(0,0){$\calA_1$}}
\put(3,12.5){\makebox(0,0){$\calA_2$}}
\put(4.5,12.5){\makebox(0,0){\circle*{0.2}}}
\put(5,12.5){\makebox(0,0){\circle*{0.2}}}
\put(5.5,12.5){\makebox(0,0){\circle*{0.2}}}
\put(7,12.5){\makebox(0,0){$\calA_n$}}
\put(1,12){\vector(0,-1){1.5}}
\put(3,12){\vector(0,-1){1.5}}
\put(7,12){\vector(0,-1){1.5}}
\put(-2,9.8){\makebox(0,0){$\bS$}}
\put(-1.5,9.75){\vector(1,0){1.5}}
\put(0,9){\framebox(8,1.5){Oracle}}
\put(1,9){\vector(0,-1){2}}
\put(3,9){\vector(0,-1){2}}
\put(7,9){\vector(0,-1){2}}
\put(0.4,8){\makebox(0,0){$X_1$}}
\put(2.4,8){\makebox(0,0){$X_2$}}
\put(4,8){\makebox(0,0){\circle*{0.2}}}
\put(4.5,8){\makebox(0,0){\circle*{0.2}}}
\put(5,8){\makebox(0,0){\circle*{0.2}}}
\put(6.4,8){\makebox(0,0){$X_n$}}
\put(0,5.5){\framebox(8,1.5){MDC}}
\put(1,5.5){\vector(0,-1){2}}
\put(3,5.5){\vector(0,-1){2}}
\put(7,5.5){\vector(0,-1){2}}
\put(0.4,4.5){\makebox(0,0){$Y_1$}}
\put(2.4,4.5){\makebox(0,0){$Y_2$}}
\put(4,4.5){\makebox(0,0){\circle*{0.2}}}
\put(4.5,4.5){\makebox(0,0){\circle*{0.2}}}
\put(5,4.5){\makebox(0,0){\circle*{0.2}}}
\put(6.4,4.5){\makebox(0,0){$Y_n$}}
\put(0,2){\framebox(8,1.5){Decoder}}
\put(4,2){\vector(0,-1){1.5}}
\put(4,0){\makebox(0,0){$\hat{\bS}$}}
\end{picture}}
&
\hspace{.4in}
\scalebox{1}{
\begin{picture}(8,12.5)
\linethickness{1pt}
\put(4,12.5){\makebox(0,0){$\calA_i$}}
\put(4,12){\vector(0,-1){1.5}}
\put(-2,9.8){\makebox(0,0){$\bS$}}
\put(-1.5,9.75){\vector(1,0){1.5}}
\put(0,9){\framebox(8,1.5){Oracle}}
\put(4,9){\vector(0,-1){2}}
\put(3.4,8){\makebox(0,0){$X_i$}}
\put(0,5.5){\framebox(8,1.5){MDC}}
\put(4,5.5){\vector(0,-1){2}}
\put(3.4,4.5){\makebox(0,0){$Y_i$}}
\put(9,3){\line(0,1){9}}
\put(9,12){\vector(-1,0){4.5}}
\put(7,12.5){\makebox(0,0){$Y_1,\ldots,Y_{i-1}$}}
\put(6.8,11.5){\makebox(0,0){$\calA_1,\ldots,\calA_{i-1}$}}
\put(0,2){\framebox(8,1.5){Decoder}}
\put(8,3){\line(1,0){1}}
\put(4,2){\vector(0,-1){1.5}}
\put(3,1.3){\makebox(0,0){Stop}}
\put(4,0){\makebox(0,0){$\hat{\bS}$}}
\end{picture}}
\vspace{.1in}
\\
(a) Non-adaptive query procedure with $n$ queries&\hspace{.2in} (b) Adaptive query procedure for $i$-th query
\end{tabular}
\caption{Illustration of query procedures for the noisy 20 questions problem with measurement-dependent channel (MDC). In the non-adaptive case (a), a target slate $\bS$ is known to the oracle who responds to a block of queries $\calA_1,\ldots,\calA_n$ and provides binary responses $X_1,\ldots,X_n$, respectively. These responses are corrupted by a measurement-dependent channel (MDC) that outputs symbols $Y_1,\ldots,Y_n$, which are used by the decoder to produce estimate $\hat{\bS}$. In the adaptive case (b), the queries are posed sequentially and decoder needs to determine when to stop the query procedure.
}
\label{illustrate:procedures}
\end{figure}

\subsection{The Measurement-Dependent Channel}
In this subsection, we describe succinctly the measurement-dependent channel scenario~\cite{kaspi2018searching}, also known as a channel with state~\cite[Chapter 7]{el2011network}. Given a sequence of queries $\calA^n\subseteq[0,1]^{nd}$, the channel from the oracle to the player is a memoryless channel whose transition probabilities are functions of the queries. Specifically, for any $(x^n,y^n)\in\{0,1\}^n\times\calY^n$,
\begin{align}
P_{Y^n|X^n}^{\calA^n}(y^n|x^n)
&=\prod_{i\in[n]}P_{Y|X}^{\calA_i}(y_i|x_i),
\end{align}
where $P_{Y|X}^{\calA_i}$ denotes the transition probability of the channel which depends on the $i$-th query $\calA_i$. Given any query $\calA\subseteq[0,1]^d$, define the volume $|\calA|$ of $\calA$ as its Lebesgue measure, i.e., $|\calA|=\int_{t\in\calA}\rmd t$. Throughout the paper, \blue{we consider only Lebesgue measurable query sets} and 
assume that the measurement-dependent channel $P_{Y|X}^{\calA}$ depends on the query $\calA$ \blue{only through} its size. Thus, $P_{Y|X}^{\calA}$ is equivalent to a channel with state $P_{Y|X}^q$ where the state $q=|\calA|\in[0,1]$.

For any $q\in[0,1]$, any $\xi\in(0,\min(q,1-q))$ and any subsets $\calA$, $\calA^+$ and $\calA^-$ of $[0,1]$ with sizes $|\calA|=q$, $|\calA^+|=q+\xi$ and $|\calA^-|=q-\xi$, we assume the measurement-dependent channel is continuous in the sense that there exists a constant $c(q)$ depending on $q$ only such that
\begin{align}
\max\left\{\left\|\log\frac{P_{Y|X}^{\calA}}{P_{Y|X}^{{\calA^+}}}\right\|_{\infty},\left\|\log\frac{P_{Y|X}^\calA}{P_{Y|X}^{\calA^-}}\right\|_{\infty}\right\}\leq c(q)\xi\label{assump:continuouschannel},
\end{align}
\blue{where the infinity norm is defined as $\|\ba\|_{\infty}:=\max_{i\in[m],j\in[n]}|a_{i,j}|$ for any matrix $\ba=\{a_{i,j}\}_{i\in[m],j\in[n]}$.}

Some examples of measurement-dependent channels satisfying the continuous constraint in \eqref{assump:continuouschannel} are as follows.
\begin{definition}
\label{def:mdBSC}
Given any $\calA\subseteq[0,1]$, a channel $P_{Y|X}^{\calA}$ is said to be a measurement-dependent Binary Symmetric Channel (BSC) with parameter $\nu\in[0,1]$ if $\calX=\calY=\{0,1\}$ and 
\begin{align}
P_{Y|X}^{\calA}(y|x)=(\nu|\calA|)^{\bbo(y\neq x)}(1-\nu|\calA|)^{\bbo(y=x)},~\forall~(x,y)\in\{0,1\}^2.
\end{align}
\end{definition}
This definition generalizes \cite[Theorem 1]{kaspi2018searching}, where the authors considered a measurement-dependent BSC with parameter $\nu=1$. \blue{Note that the binary output bit of a measurement-dependent BSC with parameter $\nu$ is flipped with probability $\nu|\calA|$.}

\begin{definition}
Given any $\calA\subseteq[0,1]$, a measurement-dependent channel $P_{Y|X}^{\calA}$ is said to be a measurement-dependent Binary Erasure Channel (BEC) with parameter $\tau\in[0,1]$ if $\calX=\{0,1\}$, $\calY=\{0,1,\rme\}$ and
\begin{align}
\blue{
P_{Y|X}^{\calA}(y|x)=(1-\tau|\calA|)^{\bbo(y=x)}(\tau|\calA|)^{\bbo(y=\rme)}
}
\end{align}
\end{definition}
\blue{Note that the binary output bit of a measurement-dependent BEC with parameter $\tau$ is erased with probability $\nu|\calA|$.}

\begin{definition}
\label{def:mdZ}
Given any $\calA\subseteq[0,1]$, a measurement-dependent channel $P_{Y|X}^{\calA}$ is said to be a measurement-dependent Z-channel with parameter $\zeta\in[0,1]$ if $\calX=\{0,1\}$, $\calY=\{0,1\}$ and
\begin{align}
\blue{
P_{Y|X}^{\calA}(y|x)=(1-\zeta|\calA|)^{\bbo(y=x=1)}(\zeta|\calA|)^{\bbo(y=0,x=1)}(0)^{\bbo(y=1,x=0)}.
}
\end{align}
\end{definition}
\blue{Note that the binary output bit of a measurement-dependent Z-channel is flipped with probability $\zeta|\calA|$ if the input is $x=1$.}

Each of these measurement-dependent channels will be considered in the sequel.

\subsection{Non-Adaptive Query Procedures}
A non-adaptive query procedure with resolution $\delta$ and excess-resolution constraint $\varepsilon$ is defined as follows.

\begin{definition}
\label{def:procedure}
Given any $(n,d)\in\bbN^2$, $\delta\in\bbR_+$ and $\varepsilon\in[0,1]$, an $(n,d,\delta,\varepsilon)$-non-adaptive query procedure for the noisy 20 questions consists of 
\begin{itemize}
\item $n$ queries $(\calA_1,\ldots,\calA_n)$ where each $\calA_i\subseteq[0,1]^d$,
\item and a decoder $g:\calY^n\to[0,1]^d$
\end{itemize}
such that the excess-resolution probability satisfies
\begin{align}
\rmP_\rme(n,d,\delta)&:=\sup_{f_{\bS}\in\calF([0,1]^d)}\Pr\{\exists~i\in[d]:~|\hatS_i-S_i|>\delta\}\leq \varepsilon\label{def:excessresolution2},
\end{align}
where $\hatS_i$ is the estimate of $i$-th element of the $d$-dimensional target $\bS$ using the decoder $g$, i.e., $g(Y^n)=(\hatS_1,\ldots,\hatS_d)$.
\end{definition}

\blue{In Algorithm \ref{procedure:nonadapt}, we provide a non-adaptive query procedure which is used in our achievability proof. The procedure is parametrized by two parameters $M$ and $p$, where $\frac{1}{M}$ is the target resolution and $p$ is the design parameter.} The definition of the \emph{excess-resolution probability} with respect to $\delta$ is inspired by rate-distortion theory~\cite{berger1971rate,kostina2013lossy}. Our formulation generalizes that of \cite{kaspi2018searching} where the authors constrained the target-dependent maximum excess-resolution probability for the case of $d=1$, i.e., \blue{i.e., $\sup_{s_1\in[0,1]}\Pr\{|\hatS_1-s_1|>\delta\}$.}

In practical applications, the number of queries are often limited to minimize total cost of queries and maintain low latency. We are interested in the establishing a non-asymptotic fundamental limit to achievable resolution $\delta$:
\begin{align}
\delta^*(n,d,\varepsilon)
&:=\inf\big\{\delta\in[0,1]:\exists\mathrm{~an~}(n,d,\delta,\varepsilon)\mathrm{-non}\mathrm{-adaptive}\mathrm{~query}\mathrm{~procedure}\big\}\label{def:delta*}.
\end{align}
Note that $\delta^*(n,d,\varepsilon)$ denotes the minimal resolution one can achieve with probability at least $1-\varepsilon$ using a non-adaptive query procedure with $n$ queries. In other words, $\delta^*(n,d,\varepsilon)$ is the achievable resolution of optimal non-adaptive query procedures tolerating an excess-resolution probability of $\varepsilon\in[0,1]$. Dual to \eqref{def:delta*} is the sample complexity, determined by the minimal number of queries required to achieve a resolution $\delta$ with probability at least $1-\varepsilon$, i.e.,
\begin{align}
n^*(d,\delta,\varepsilon):=\inf\big\{n\in\bbN:\exists\mathrm{~an~}(n,d,\delta,\varepsilon)\mathrm{-non}\mathrm{-adaptive}\mathrm{-query}\mathrm{-procedure}\big\}\label{def:n*}.
\end{align}
One can easily verify that for any $(\delta,\varepsilon)\in\bbR_+\times[0,1]$,
\begin{align}
n^*(d,\delta,\varepsilon)
&=\inf\{n:\delta^*(n,d,\varepsilon)\leq \delta\}\label{def:sc_non}.
\end{align}
Thus, it suffices to derive the fundamental limit $\delta^*(n,d,\varepsilon)$.

\subsection{Adaptive Query Procedures}

An adaptive query procedure with resolution $\delta$ and excess-resolution constraint $\varepsilon$ is defined as follows.
\begin{definition}
\label{def:adaptive:procedure}
Given any $(l,d,\delta,\varepsilon)\in\bbR_+\times\bbN\times\bbR_+\times[0,1]$, an $(l,d,\delta,\varepsilon)$-adaptive query procedure for the noisy 20 questions problem consists of
\begin{itemize}
\item a sequence of adaptive queries where for each $i\in\bbN$, the design of query $\calA_i\subseteq[0,1]^d$ is based all previous queries $\{\calA_j\}_{j\in[i-1]}$ and the noisy responses $Y^{i-1}$ from the oracle
\item a sequence of decoding functions $g_i:\calY^i\to[0,1]^d$ for $i\in\bbN$ 
\item a random stopping time $\tau$ depending on noisy responses $\{Y_i\}_{i\in\bbN}$ such that under any pdf $f_{\bS}$ of the target random variable $\bS$, the average number of queries satisfies
\begin{align}
\mathbb{E}[\tau]\leq l,
\end{align}
\end{itemize}
such that the excess-resolution probability satisfies
\begin{align}
\rmP_{\rme,\rma}(l,d,\delta):=\sup_{f_{\bS}\in\calF([0,1]^d)}\Pr\{\exists~i\in[d]:~|\hatS_i-S_i|>\delta\}\leq \varepsilon\label{def:excessresolution},
\end{align}
where $\hatS_i$ is the estimate of $i$-th element of the target $\bS$ using the decoder $g$ at time $\tau$, i.e., $g(Y^\tau)=(\hatS_1,\ldots,\hatS_d)$.
\end{definition}
\blue{An adaptive query procedure is provided in Algorithm \ref{procedure:adapt}.}

Similar to \eqref{def:delta*}, given any $(l,d,\varepsilon)\in\bbR_+\times\bbN\times[0,1)$, we can define the fundamental resolution limit for adaptive querying as follows: 
\begin{align}
\delta_\rma^*(l,d,\varepsilon)
&:=\inf\{\delta\in\bbR_+:~\exists~\mathrm{an}~(l,d,\delta,\varepsilon)\mathrm{-adaptive}~\mathrm{query~procedure}\}\label{def:delta*:adaptive},
\end{align}
with analogous definition of mean sample complexity (cf. \eqref{def:sc_non})
\begin{align}
l^*(d,\delta,\varepsilon):=\inf\{l\in\bbR_+:~\exists~\mathrm{an}~(l,d,\delta,\varepsilon)\mathrm{-adaptive}~\mathrm{query~procedure}\}.
\end{align}

\section{Main Results for Non-Adaptive Query Procedures}

\subsection{Non-Asymptotic Bounds}
We first present an upper bound on the error probability of optimal non-adaptive query procedures. Given any $(p,q)\in[0,1]^2$, let $P_Y^{p,q}$ be the marginal distribution on $\calY$ induced by the Bernoulli distribution $P_X=\mathrm{Bern}(p)$ and the measurement-dependent channel $P_{Y|X}^{q}$. Furthermore, define the following information density
\begin{align}
\imath_{p,q}(x;y)&:=\log\frac{P_{Y|X}^q(y|x)}{P_Y^{p,q}(y)},~\forall~(x,y)\in\calX\times\calY\label{def:ipqxy}.
\end{align}
Correspondingly, for any $(x^n,y^n)\in\calX^n\times\calY^n$, we define
\begin{align}
\imath_p(x^n;y^n)
&:=\sum_{i\in[n]}\imath_{p,p}(x_i;y_i)\label{def:ixnyn}
\end{align}
as the mutual information density between $x^n$ and $y^n$.

\begin{algorithm}[bt]
\caption{\color{blue} Non-adaptive query procedure for searching for a multidimensional target over the unit cube}
\label{procedure:nonadapt}
\begin{algorithmic}
\REQUIRE The number of queries $n\in\bbN$, the dimension $d\in\bbN$ and two parameters $(M,p)\in\bbN\times(0,1)$
\ENSURE An estimate $(\hats_1,\ldots,\hats_d)\in[0,1]^d$ of a $d$-dimensional target variable $(s_1,\ldots,s_d)\in[0,1]^d$
\STATE Partition the unit cube of dimension $d$ (i.e., $[0,1]^d$) into $M^d$ equal-sized disjoint cubes $\{\calS_{i_1,\ldots,i_d}\}_{(i_1,\ldots,i_d)\in[M]^d}$.
\STATE Generate $M^d$ binary vectors $\{x^n(i_1,\ldots,i_d)\}_{(i_1,\ldots,i_d)\in[M]^d}$ where each binary vector is generated i.i.d. from a Bernoulli distribution with parameter $p$.
\STATE $t \leftarrow 1$.
\WHILE{$t\leq n$}
\STATE Form the $t$-th query as
\begin{align*}
\calA_t:=\bigcup_{(i_1,\ldots,i_d)\in[M]^d:x_t(i_1,\ldots,i_d)=1}\calS_{i_1,\ldots,i_d}.
\end{align*}
\STATE Obtain a noisy response $y_t$ from the oracle to the query $\calA_t$.
\STATE $t \leftarrow t+1$.
\ENDWHILE
\STATE Generate estimates $(\hats_1,\ldots,\hats_d)$ as
\begin{align*}
\hats_i=\frac{2\hatw_i-1}{2M},~i\in[d]
\end{align*}
where $\hat{\bw}=(\hatw_1,\ldots,\hatw_d)$ is obtained via the maximum mutual information density estimator, i.e,
\begin{align*}
\hat{\bw}=\max_{(\tili_1,\ldots,\tili_d)\in[M]^d}\imath_p(x^n(\tili_1,\ldots,\tili_d);y^n).
\end{align*}
\end{algorithmic}
\end{algorithm}

\begin{theorem}
\label{ach:fbl}
Given any $(n,d,M)\in\bbN^3$, for any $p\in[0,1]$ and any $\eta\in\bbR_+$, \blue{the procedure in Algorithm \ref{procedure:nonadapt}} is an $(n,d,\frac{1}{M},\varepsilon)$-non-adaptive query procedure where
\begin{align}
\varepsilon&\leq 
4n\exp(-2M^d\eta^2)+\exp(n\eta c(p))\mathbb{E}[\min\{1,M^d\Pr\{\imath_p(\barX^n;Y^n)\geq \imath_p(X^n;Y^n)|X^n,Y^n\}]\}\label{com:fbl},
\end{align}
where $(X^n,\barX^n,Y^n)$ is distributed as $P_X^n(X^n)P_X^n(\barX^n)(P_{Y|X}^p)^n(Y^n|X^n)$ with $P_X$ defined as the Bernoulli distribution with parameter $p$ (i.e., $P_X(1)=p$).
\end{theorem}
\blue{The proof of Theorem \ref{ach:fbl} uses a modification of the random coding union bound~\cite{polyanskiy2010finite}} and is given in Appendix \ref{proof:ach}. 

{\color{blue}
Consider the measurement-independent channel where $P_{Y|X}^q=P_{Y|X}^{1}=:P_{Y|X}$ for all $q\in[0,1]$. Similarly to the proof of Theorem \ref{ach:fbl}, we can show that for any $p\in[0,1]$, there exists an $(n,d,\frac{1}{M},\varepsilon)$-non-adaptive query procedure such that
\begin{align}
\varepsilon&\leq \mathbb{E}[\min\{1,M^d\Pr\{\jmath_p(\barX^n;Y^n)\geq \jmath_p(X^n;Y^n)|X^n,Y^n\}]\}\label{com:fbl2},
\end{align}
where the tuple of random variables $(X^n,\barX^n,Y^n)$ is distributed as $P_X^n(X^n)P_X^n(\barX^n)P_{Y|X}^n(Y^n|X^n)$, the information density $\jmath_p(x^n;y^n)$ is defined as
\begin{align}
\jmath_p(x^n;y^n):=\log\frac{P_{Y|X}^n(y^n|x^n)}{P_Y^n(y^n)},
\end{align}
and $P_Y$ is the marginal distribution induced by $P_X$ and $P_{Y|X}$.} Comparing the measurement-independent case \eqref{com:fbl2} with the measurement-dependent case \eqref{com:fbl}, the non-asymptotic upper bound \eqref{com:fbl} in Theorem \ref{ach:fbl} differs from \eqref{com:fbl2} in two aspects: there are an additional additive term and an additional multiplicative term in \eqref{com:fbl}. As is made clear in the proof of Theorem \ref{ach:fbl}, the additive term $4n\exp(-2M^d\eta^2)$ results from the atypicality of the measurement-dependent channel and the multiplicative term $\exp(n\eta c(p))$ appears due to the change-of-measure we use to replace the measurement-dependent channel $P_{Y^n|X^n}^{\calA^n}$ with the measurement-independent channel $(P_{Y|X}^p)^n$.
 
We next provide a non-asymptotic converse bound to complement Theorem \ref{ach:fbl}. For simplicity, for any query $\calA\subseteq[0,1]^d$ and any $(x,y)\in\calX\times\calY$, we use $\imath_{\calA}(x,y)$ to denote $\imath_{|\calA|,|\calA|}(x,y)$.

\begin{theorem}
\label{fbl:converse}
Set $(n,\delta,\varepsilon)\in\bbN\times\bbR_+\times[0,1]$. Any $(n,\delta,\varepsilon)$-non-adaptive query procedure satisfies the following. For any $\beta\in(0,\frac{1-\varepsilon}{2})$ and any $\kappa\in(0,1-\varepsilon-2d\beta)$,
\begin{align}
-d\log\delta\leq -d\log\beta-\log\kappa+\sup_{\calA^n\subseteq[0,1]^{nd}}\sup\bigg\{t\Big|\Pr\Big\{\sum_{i\in[n]}\imath_{\calA_i}(X_i;Y_i)\leq t\bigg\}\leq \varepsilon+2d\beta+\kappa\Big\}\label{uppbound}.
\end{align}
\end{theorem}
The proof of Theorem \ref{fbl:converse} is given in Appendix \ref{proof:converse}. The proof of Theorem \ref{fbl:converse} is decomposed into two steps: i) we use the result in \cite{kaspi2018searching} which states that the excess-resolution probability of any non-adaptive query procedure can be lower bounded by the error probability associated with channel coding over the measurement-dependent channel with uniform message distribution, minus a certain term depending on $\beta$; and ii) we apply the non-asymptotic converse bound for channel coding~\cite[Proposition 4.4]{TanBook} by realizing that, given a sequence of queries, the measurement-dependent channel is simply a time varying channel with deterministic states at each time point.

We remark that the non-asymptotic bounds in Theorems \ref{ach:fbl} and \ref{fbl:converse} hold for any number of queries and any measurement-dependent channels satisfying \eqref{assump:continuouschannel}. As we shall see in the next subsection, these non-asymptotic bounds lead to the second-order asymptotic result in Theorem \ref{result:second}, \blue{which provides an approximation to the finite blocklength fundamental limit $\delta^*(n,d,\varepsilon)$. The exact calculation of the upper bound in Theorem \ref{fbl:converse} is challenging. However, for $n$ sufficiently large, as demonstrated in the proof of Theorem \ref{result:second}, the supremum in \eqref{uppbound} can be achieved by queries $\calA^n$ where each query $\calA_i$ has the same size.}

\subsection{Second-Order Asymptotic Approximation}
In this subsection, we present the second-order asymptotic approximation to the achievable resolution $\delta^*(n,d,\varepsilon)$ of optimal non-adaptive query procedures after $n$ queries subject to a worst case excess-resolution probability of $\varepsilon\in[0,1)$.

Given measurement-dependent channels $\{P_{Y|X}^q\}_{q\in[0,1]}$, the channel ``capacity" is defined as
\begin{align}
C&:=\max_{q\in[0,1]} \mathbb{E}[\imath_{q,q}(X;Y)]\label{def:capacity},
\end{align}
where $(X,Y)\sim \mathrm{Bern}(q)\times P_{Y|X}^q$.

Let the capacity-achieving set $\calP_{\rm{ca}}$ be the set of optimizers achieving \eqref{def:capacity}. Then, for any $\varepsilon\in[0,1)$, define the following ``dispersion'' of the measurement-dependent channel
\begin{align}
V_\varepsilon
&:=\left\{
\begin{array}{cc}
\min_{q\in\calP_{\rm{ca}}}\mathrm{Var}[\imath_{q,q}(X;Y)]&\mathrm{if~}\varepsilon<0.5,\\
\max_{q\in\calP_{\rm{ca}}}\mathrm{Var}[\imath_{q,q}(X;Y)]&\mathrm{if~}\varepsilon\geq 0.5.
\end{array}
\right.
\label{def:veps}
\end{align}
The case of $\varepsilon<0.5$ will be the focus of the sequel of this paper. 
\begin{theorem}
\label{result:second}
Assume for any $q\in\calP_{\rm{ca}}$, the third absolute moment of $\imath_{q,q}(X;Y)$ is finite. For any $\varepsilon\in(0,1)$, the achievable resolution $\delta^*(n,d,\varepsilon)$ of optimal non-adaptive query procedures satisfies
\begin{align}
-\log\delta^*(n,d,\varepsilon)
&=\frac{1}{d}\Big(nC+\sqrt{nV_\varepsilon}\Phi^{-1}(\varepsilon)+O(\log n)\Big)\label{joint:search},
\end{align}
where the remainder term satisfies $-\frac{1}{2}\log n+O(1)\leq O(\log n)\leq\log n+O(1)$.
\end{theorem}
The proof of Theorem \ref{result:second} is provided in Appendix \ref{proof:second}. \blue{In the achievability proof, we make use of the non-adaptive query procedure in Algorithm \ref{procedure:nonadapt} and thus prove its second-order optimality.}

We make the following remarks. Firstly, Theorem \ref{result:second} implies a phase transition analogous to those found in group testing in machine learning~\cite{scarlett2017nips,Scarlett:2016:PTG:2884435.2884439}, which we interpret in Figure \ref{pt}.
\begin{figure}[tb]
\centering
\includegraphics[width=.5\columnwidth]{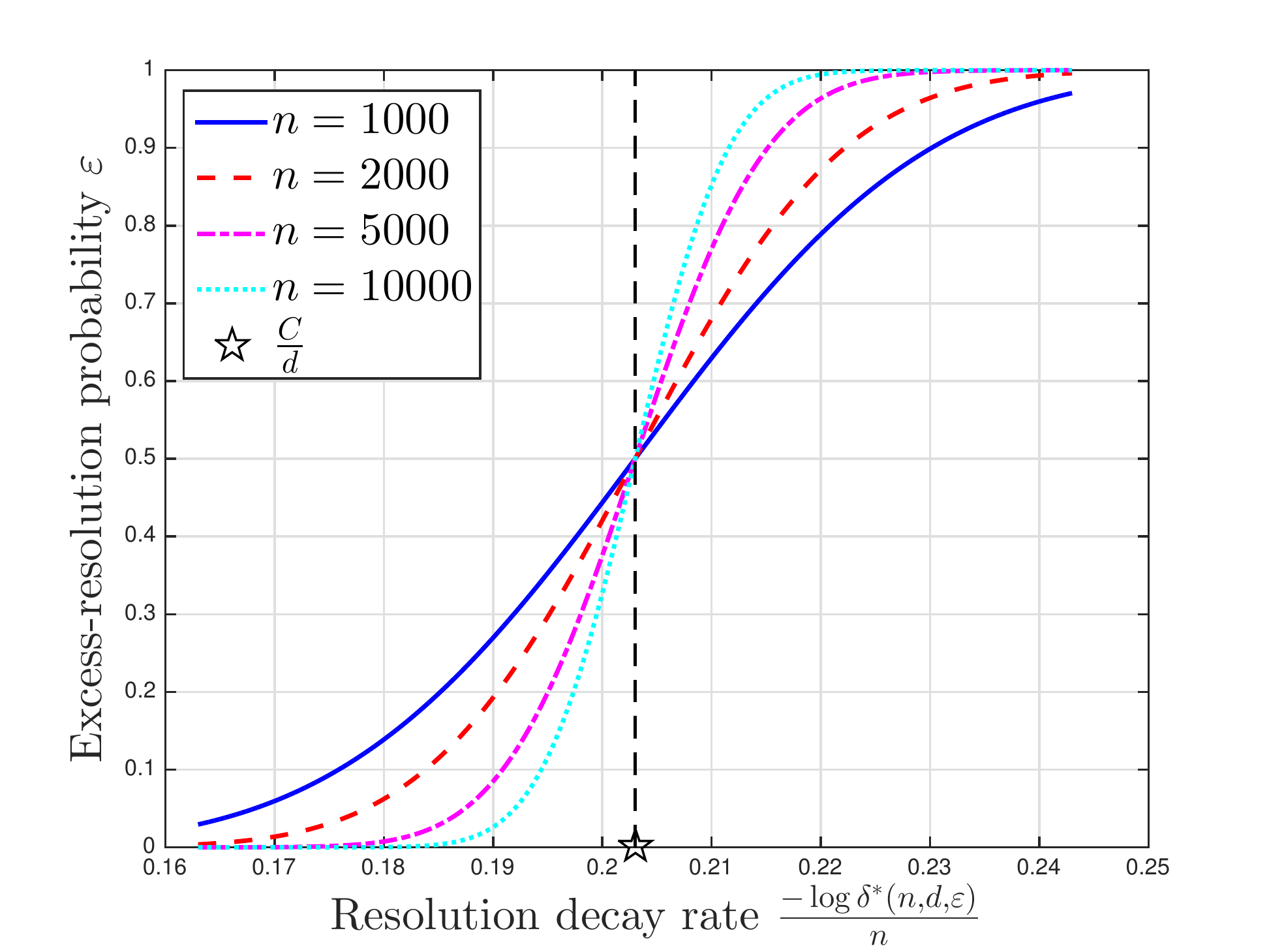}
\caption{Illustration of the phase transition of non-adaptive query procedures for the case of $d=2$ when the noisy channel is a measurement-dependent BSC with parameter $\nu=0.2$. On the one hand, when the resolution decay rate is strictly greater than the capacity $\frac{C}{d}$, then as the number of the queries $n\to\infty$, the excess-resolution probability tends to one. On the other hand, when the resolution decay rate is strictly less than the capacity $\frac{C}{d}$, then the excess-resolution probability vanishes as the number of the queries increases.}
\label{pt}
\end{figure}
\blue{We remark that this phase transition is a direct result of the second-order asymptotic analysis and does not follow from a first-order asymptotic analysis, e.g., that developed in \cite[Theorem 1]{kaspi2018searching}.} As a corollary of Theorem \ref{result:second}, for any $\varepsilon\in(0,1)$,
\begin{align}
\lim_{n\to\infty}-\frac{1}{n}\log \delta^*(n,d,\varepsilon)=\frac{C}{d}\label{phasetransition}.
\end{align}
which takes the form of a strong converse~\cite{zhou2016cilossy,wei2009strong,liu2016brascamp} to the channel coding theorem. The result in \eqref{phasetransition} indicates that tolerating a smaller, or even vanishing excess-resolution probability, does not improve the asymptotic achievable resolution decay rate of an optimal non-adaptive query procedure.

Secondly, Theorem \ref{result:second} refines \cite[Theorem 1]{kaspi2018searching} in several directions. First, Theorem \ref{result:second} is a  second-order asymptotic result which provides an approximation for the finite blocklength performance while \cite[Theorem 1]{kaspi2018searching} only characterizes the asymptotic resolution decay rate with vanishing worst-case excess-resolution probability for a one-dimensional target, i.e., \blue{$\lim_{\varepsilon\to 0}\lim_{n\to\infty}(-\log\delta^*(n,1,\varepsilon))$}. Second, our results hold for any measurement-dependent channel satisfying \eqref{assump:continuouschannel} while \cite[Theorem 1]{kaspi2018searching} only addresses the measurement-dependent BSC.

\blue{Thirdly, the dominant event which leads to an excess-resolution in noisy 20 questions estimation is the atypicality of the information density $\imath_p(X^n;Y^n)$ (cf. \eqref{def:ixnyn}). To characterize the probability of this event, we make use of the Berry--Esseen theorem and show that the mean $C$ and the variance $\rmV_\varepsilon$ of the information density $\imath_{q,q}(X;Y)$ play critical roles.}

Fourthly, we remark that any number $s\in[0,1]$ has the binary expansion \blue{$(b_0.b_1b_2\ldots)$.} We can thus interpret Theorem \ref{result:second} as follows: using an optimal non-adaptive query procedure, after $n$ queries, with probability of at least $1-\varepsilon$, one can extract the first $\lfloor-\log_2\delta^*(n,d,\varepsilon)\rfloor$ bits of the binary expansion of each dimension of the target variable $\bS=(S_1,\ldots,S_d)$.

{\color{blue}
A final remark is that separate estimation over each dimension of a multi-dimensional target variable using the special $d=1$ case in Algorithm \ref{procedure:nonadapt} is not second-order optimal although it is optimal asymptotically. Explanations are provided as follows.

From the asymptotic result in \eqref{phasetransition}, we observe that \emph{asymptotically,} it is in fact optimal to allocate roughly $\frac{n}{d}$ queries to each dimension using the special $d=1$ case of Algorithm \ref{procedure:nonadapt} when searching for a $d$-dimensional target variable. Such a decoupled searching algorithm achieves the asymptotic optimal resolution decay rate for non-adaptive query procedures. However, from a second-order asymptotic point of view, allocating equal number of queries to search over each dimension is \emph{not} optimal. Supposed that over each dimension $i\in[d]$, we allocate $\frac{n}{d}$ queries to search for the value of $S_i$ and tolerate excess-resolution probabilities $\varepsilon_i$, similarly to the achievability part of Theorem \ref{result:second}, we find that the achievable resolution $\delta_{\rm{sep}}(n,d,\varepsilon)$ satisfies
\begin{align}
-\log\delta_{\rm{sep}}(n,d,\varepsilon)
&=\max_{(\varepsilon_1,\ldots,\varepsilon_d):\sum_{i\in[d]}\varepsilon_i\leq \varepsilon}\min_{i\in[d]}\left\{\frac{nC}{d}+\sqrt{\frac{n\rmV_{\varepsilon_i}}{d}}\Phi^{-1}(\varepsilon_i)\right\}+O(\log n)\label{performance:seperate1}\\
&=\frac{nC}{d}+\sqrt{\frac{n\rmV_{\frac{\varepsilon}{d}}}{d}}\Phi^{-1}\left(\frac{\varepsilon}{d}\right)+O(\log n)\label{performance:seperate},
\end{align}
where \eqref{performance:seperate} follows since i) for any $\varepsilon\in(0,1)$ and any
$d\geq 2$, the minimization in \eqref{performance:seperate1} is achieved by some $i\in[d]$ such that $\varepsilon_i<0.5$ because $\Phi^{-1}(a)$ is decreasing in $a\in[0,1]$ and $\Phi^{-1}(a)< 0$ for any $a< 0.5$, and ii) for any $\varepsilon\in[0,1]$, the maximization is achieved by a vector $(\varepsilon_1,\ldots,\varepsilon_d)$ where $\varepsilon_i=\frac{\varepsilon}{d}$ for all $i\in[d]$.

Note that for any $\varepsilon\geq 0.5$ and any $d\geq 2$, the right hand side of \eqref{performance:seperate} is no greater than $\frac{nC}{d}+O(\log n)$ since $\Phi^{-1}(\frac{\varepsilon}{d})\leq \Phi^{-1}(0.5)=0$. However, the right hand side of \eqref{joint:search} is greater than $\frac{nC}{d}+O(\log n)$ since both $\Phi^{-1}(\varepsilon)$ and $V_\varepsilon$ are positive when $\varepsilon>0.5$. Furthermore, when $\varepsilon<0.5$, we have
\begin{align}
\frac{nC}{d}+\sqrt{\frac{n\rmV_{\frac{\varepsilon}{d}}}{d}}\Phi^{-1}\left(\frac{\varepsilon}{d}\right)
&=\frac{nC}{d}+\sqrt{\frac{n\rmV_{\varepsilon}}{d}}\Phi^{-1}\left(\frac{\varepsilon}{d}\right)\label{decreasePhi-1}\\
&<\frac{nC}{d}+\sqrt{\frac{n\rmV_{\varepsilon}}{d}}\Phi^{-1}\left(\varepsilon\right)\label{phi-1non}\\
&<\frac{nC}{d}+\sqrt{\frac{n\rmV_{\varepsilon}}{d^2}}\Phi^{-1}(\varepsilon)\label{usepositivePhiinv}\\
&=\frac{1}{d}\Big(nC+\sqrt{nV_\varepsilon}\Phi^{-1}(\varepsilon)\Big)
\end{align}
where \eqref{decreasePhi-1} follows since $V_\varepsilon$ (cf. \eqref{def:veps}) takes the same value for any $\varepsilon\in[0,0.5)$, \eqref{phi-1non} follows since $\Phi^{-1}(a)$
 in non-decreasing in $a\in[0,1]$ and $\varepsilon>\frac{\varepsilon}{2}\geq \frac{\varepsilon}{d}$, \eqref{usepositivePhiinv} follows since $d\geq 2$ and $\Phi^{-1}(\varepsilon)<0$ for any $\varepsilon<0.5$. Therefore, for any $d\geq 2$, the result in \eqref{performance:seperate} is always smaller than the result in \eqref{joint:search}. This implies that separate searching over each dimension of a multidimensional target variable is in fact \emph{not} optimal. This is verified by a numerical simulation in Section \ref{sec:numerical} (Figure \ref{sim_non_adap_sep}).}

In the following, we specialize Theorem \ref{result:second} to different measurement-dependent channels.
\subsection{Case of Measurement-Dependent BSC}
 We first consider a measurement-dependent BSC. Given any $\nu\in(0,1]$ and any $q\in[0,1]$, let $\beta(\nu,q):=q(1-\nu q)+(1-q)\nu q$. For any $(x,y)\in\{0,1\}^2$, the information density of a measurement-dependent BSC with parameter $\nu$ is 
\begin{align}
\imath_{q,\nu q}(x;y)
\nn&=\bbo(x\neq y)\log(\nu q)+\bbo(x=y)\log(1-\nu q)-\bbo(y=1)\log(\beta(\nu,q))\\*
&\qquad-\bbo(y=0)\log(1-\beta(\nu,q)).
\end{align}
The mean and variance of the information density are respectively
\begin{align}
C(\nu,q)&:=\mathbb{E}[\imath_{q,\nu q}(X;Y)]=h_\rmb(\beta(\nu,q))-h_\rmb(\nu q)\label{def:Cnuq},\\
V(\nu,q)&:=\mathrm{Var}[\imath_{q,\nu q}(X;Y)],
\end{align}
where $h_\rmb(p)=-p\log(p)-(1-p)\log(1-p)$ is the binary entropy function. The capacity of the measurement-dependent BSC with parameter $\nu$ is thus
\begin{align}
C(\nu)=\max_{q\in[0,1]}C(\nu,q)\label{def:cnu},
\end{align}
\begin{figure}[tb]
\centering
\includegraphics[width=.5\columnwidth]{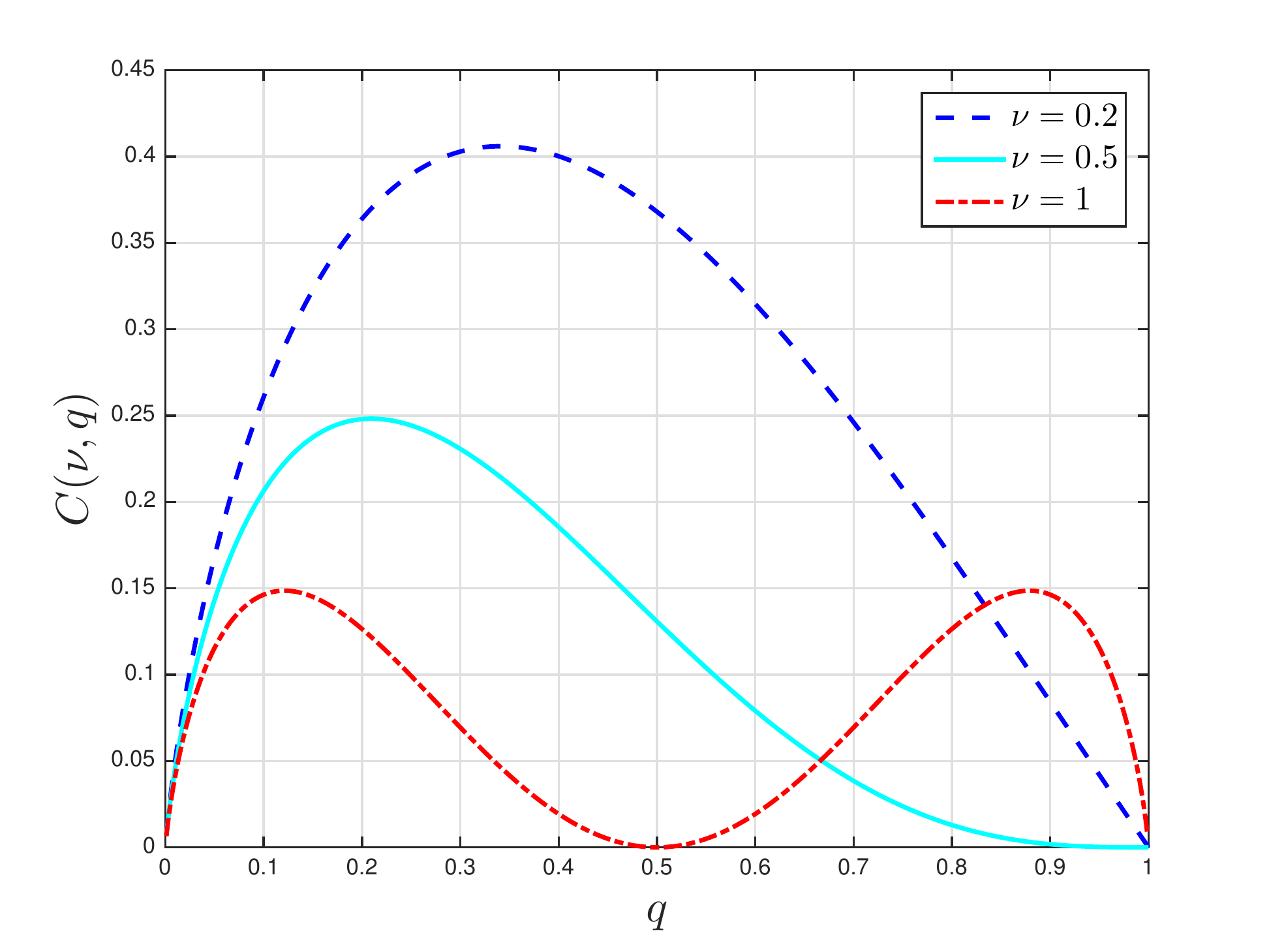}
\caption{Plot of $C(\nu,q)$, the mean of the mutual information density, of a measurement-dependent BSC for various values of $\nu$ and $q\in[0,1]$. For a given $\nu$, the maximum value of $C(\nu,q)$ over $q\in[0,1]$ is the capacity of the measurement-dependent BSC with parameter $\nu$ and the values of $q$ achieving this maximum consists of the set of capacity-achieving parameters $\calP_{\rm{ca}}$. For $\nu=0.2$ and $\nu=0.5$, $\calP_{\rm{ca}}$ is singleton and for $\nu=1$, $\calP_{\rm{ca}}$ contains two elements.}
\label{bsc_capacity}
\end{figure}

Depending on the value of $\nu\in(0,1]$, the set of capacity-achieving parameters $\calP_{\rm{ca}}$ may or may not be a singleton (cf. Figure \ref{bsc_capacity}). In particular, for any $\nu\in(0,1)$, the capacity-achieving parameter $q^*$ is unique. When $\nu=1$, there are two capacity-achieving parameters $q_1^*$ and $q_2^*$ where $q_1^*+q_2^*=1$. It can be verified easily that $V(1,q_1^*)=V(1,1-q_1^*)$. As a result, for any capacity-achieving parameter $q^*$ of the measurement-dependent BSC with parameter $\nu\in(0,1]$, the dispersion of the channel is
\begin{align}
V(\nu)&=V(\nu,q^*).
\end{align}

\begin{corollary}
\label{bsc:non-adaptive}
Let $\nu\in(0,1)$. If the channel from the oracle to the player is a measurement-dependent BSC with parameter $\nu$, then Theorem \ref{result:second} holds with $C=C(\nu)$ and $V_\varepsilon=V(\nu)$ for any $\varepsilon\in(0,1)$.
\end{corollary}
We make the following observations. Firstly, if we let $\nu=1$ and take $n\to\infty$, then for any $\varepsilon\in(0,1)$,
\begin{align}
\lim_{n\to\infty}\frac{-\log\delta^*(n,d,\varepsilon)}{n}=\frac{\max_{q\in[0,1]} \big(h_\rmb(\beta(1,q))-h_\rmb(q)\big)}{d}.
\end{align}
This is a strengthened version of \cite[Theorem 1]{kaspi2018searching} with strong converse.

Secondly, when one considers the measurement-independent BSC with parameter $\nu\in(0,1)$, then it can be shown that the achievable resolution $\delta^*_{\rm{mi}}(n,d,\varepsilon)$ of optimal non-adaptive query procedures satisfies
\begin{align}
-d\log_2 \delta^*_{\rm{mi}}(n,d,\varepsilon)=n(1-h_\rmb(\nu))+\sqrt{n\nu(1-\nu)}\log_2\frac{1-\nu}{\nu}\Phi^{-1}(\varepsilon)+O(\log n).
\end{align}
\blue{To compare the performances of optimal non-adaptive query procedures under measurement-dependent and measurement-independent channels respectively, we plot in Figure \ref{bsc_fbl} the second-order approximation to the per query average number of bits, specifically, we compare $\frac{-\log_2 \delta^*(n,2,\varepsilon)}{n}$ and $\frac{-\log_2 \delta^*_{\rm{mi}}(n,2,\varepsilon)}{n}$ for $\varepsilon=0.001$ and different values of $\nu$ and $n$.}

\blue{When $\nu<0.5$, the optimal query procedure under a measurement-dependent channel achieves a higher resolution than its counterpart in the measurement-independent case.} The intuition is that the probability of receiving wrong answers in the measurement-dependent channel is always smaller compared with the measurement-independent channel with the same parameter. However, when $\nu>0.5$, we find that the relative performances can be reversed. The reasons for this phenomenon are two fold: i) BSC is a symmetric channel, thus under the measurement-independent setting, having a BSC with crossover probability $\nu>0.5$ is equivalent to having a BSC with parameter $1-\nu<0.5$ since one can easily flip all bits; ii) under the measurement-dependent setting, since the probability of receiving wrong answers depends on the size of the query, thus the symmetric nature of BSC is lost.

\begin{figure}[tb]
\centering
\includegraphics[width=.5\columnwidth]{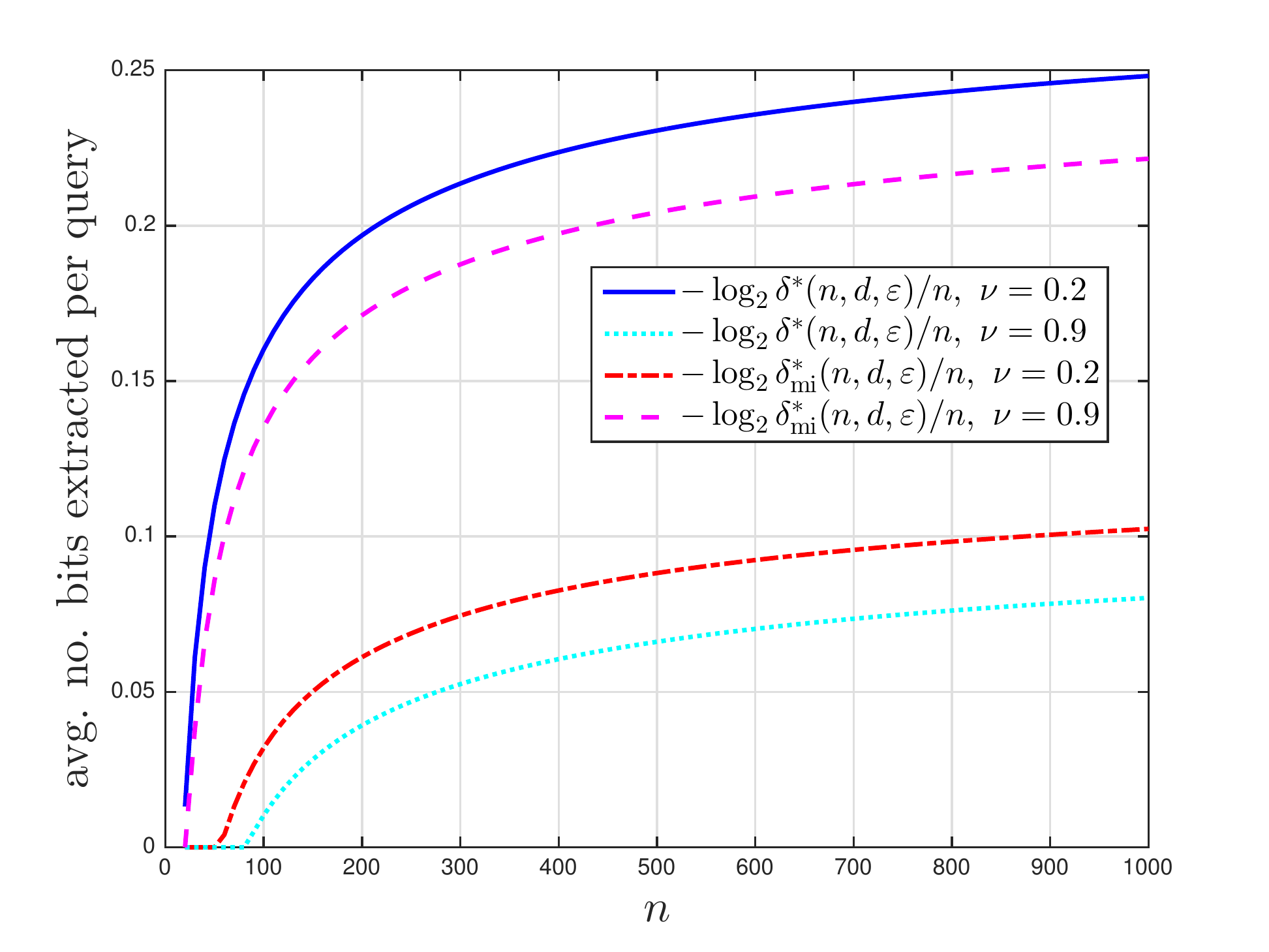}
\caption{\blue{The second-order asymptotic approximation} to the average number of bits extracted per query using an optimal non-adaptive query procedure for both measurement-dependent ($-\log_2\delta^*(n,d,\varepsilon)/n$) and measurement-independent ($-\log_2\delta_{\rm{mi}}^*(n,d,\varepsilon)/n$) versions of the BSC for different parameters $\nu$. Here we consider the case of $d=2$ and $\varepsilon=0.001$.}
\label{bsc_fbl}
\end{figure}

\subsection{Case of Measurement-Dependent BEC}

We next consider a measurement-dependent BEC. Given any $\tau\in[0,1]$ and any $q\in(0,1)$, for any $(x,y)\in\{0,1\}\times\{0,1,\rme\}$, the information density for a measurement-dependent BEC with parameter $\tau$ is
\begin{align}
\imath_{q,q \tau}(x;y)
&=\bbo(y=x)\log(1-q\tau)-\bbo(y=1)\log(q(1-q\tau))-\bbo(y=0)\log((1-q)(1-q\tau)).
\end{align}
The mean and variance of the information density are respectively
\begin{align}
C(\tau,q)
&:=\mathbb{E}[\imath_{q,q \tau}(X;Y)]=(1-q\tau)h_\rmb(q),\\
V(\tau,q)
\nn&:=\mathrm{Var}[\imath_{q,q\tau}(X;Y)]=(1-q\tau)\Big(h_\rmb(q)\log(1-q\tau)+q\log q\log(q(1-q\tau))\\*
&\qquad+(1-q)\log(1-q)\log((1-q)(1-q\tau))\Big)-(1-q\tau)^2h_\rmb(q)^2.
\end{align}

The capacity of the measurement-dependent BEC with parameter $\tau\in[0,1]$ is given by
\begin{align}
C(\tau)=\max_{q\in[0,1]}C(\tau,q)=\max_{q\in[0,0.5]}(1-q\tau)h_\rmb(q),
\end{align}
where the second equality follows since for any $\tau$, $C(\tau,q)$ is decreasing in $q\in[0.5,1]$. We plot $C(\tau,q)$ for different values of $\tau$ in Figure \ref{bec_capacity}. It can be verified that the capacity-achieving parameter for the measurement-dependent BEC is unique and we denote it by $q^*$. Thus, the dispersion of the channel is $V(\tau,q^*)$.

\begin{figure}[tb]
\centering
\includegraphics[width=.5\columnwidth]{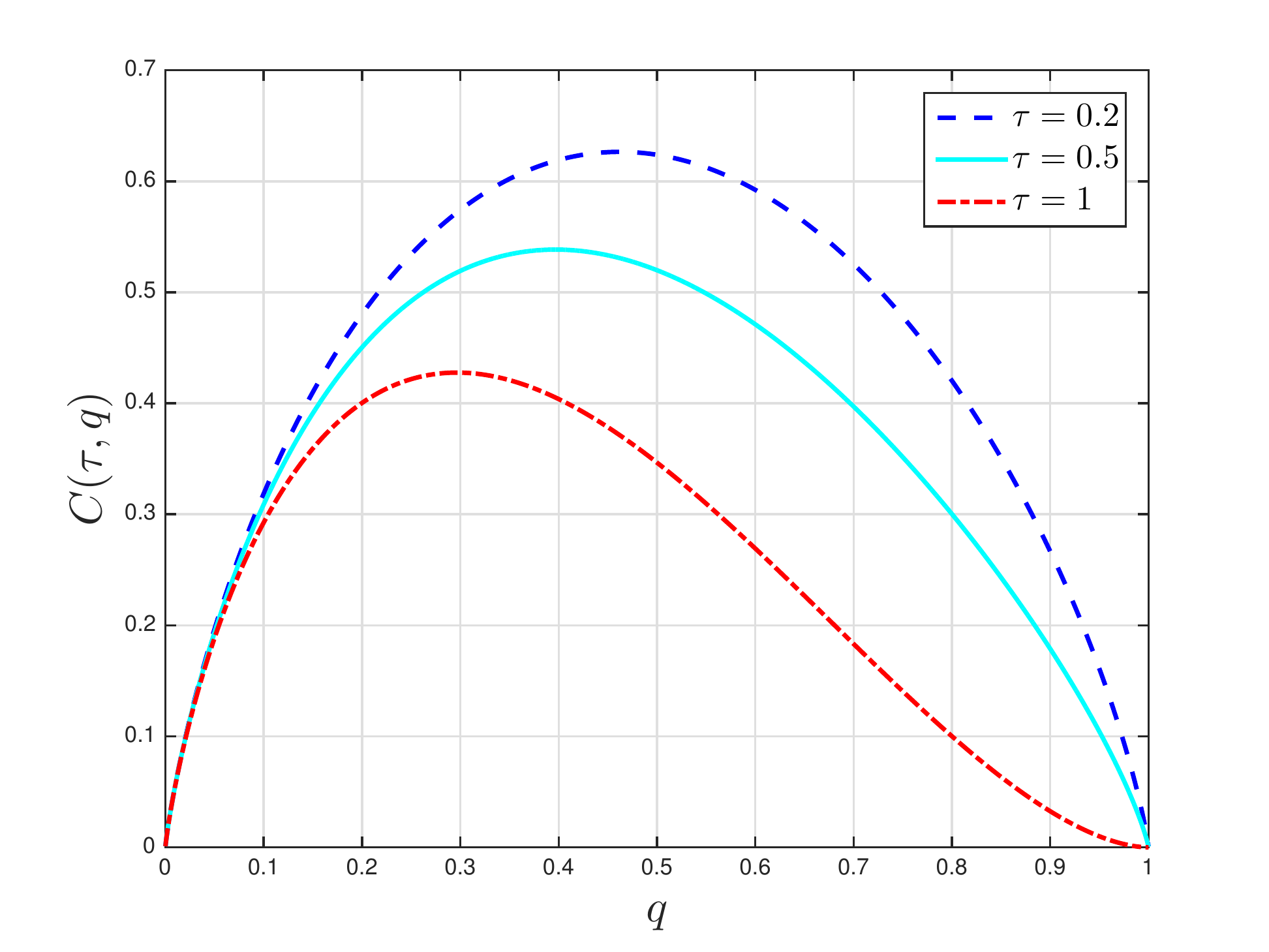}
\caption{Plot of $C(\tau,q)$ for the measurement-dependent BEC with parameter $\tau$ for $q\in[0,1]$. For each given $\tau$, the maximum value of $C(\tau,q)$ over $q\in[0,1]$ is the capacity of the measurement-dependent BEC. Note that the capacity-achieving parameter $q^*$ for measurement-dependent BEC is unique for any $\tau\in(0,1]$.}
\label{bec_capacity}
\end{figure}

We obtain the following:
\begin{corollary}
\label{bec:non-adaptive}
Let $\tau\in[0,1]$. If the channel from the oracle to the player is a measurement-dependent BEC with parameter $\tau$, then Theorem \ref{result:second} holds with $C=C(\tau)$ and $V_\varepsilon=V(\tau,q^*)$ for any $\varepsilon\in[0,1)$.
\end{corollary}
\blue{The remarks we made for Corollary \ref{bsc:non-adaptive} apply equally to Corollary \ref{bec:non-adaptive}, but additional properties are worthwhile to mention.} Firstly, if one considers a measurement-independent BEC with parameter $\tau\in[0,1]$, then the achievable resolution $\delta^*_{\rm{mi}}(n,d,\varepsilon)$ of optimal non-adaptive query procedures satisfies
\begin{align}
-d\log_2 \delta^*_{\rm{mi}}(n,d,\varepsilon)=n(1-\tau)+\sqrt{n\tau(1-\tau)}\Phi^{-1}(\varepsilon)+O(\log n).
\end{align}
We can then compare the performances of optimal query procedures for measurement-dependent and measurement-independent channels.  \blue{Accordingly, in Figure \ref{bec_fbl}, we plot the second-order approximation to both $-\log_2 \delta^*(n,d,\varepsilon)/n$ and $-\log_2 \delta^*_{\rm{mi}}(n,d,\varepsilon)/n$ for $d=2$, $\varepsilon=0.001$ and various values of $\nu$, where again we ignore the remainder $O(\log n)$. As can be observed in Figure \ref{bec_fbl}, for the considered cases, optimal querying with a measurement-dependent channel achieves a higher resolution than its counterpart with a measurement-independent channel.} The intuition is that in the measurement-dependent channel, the probability of erasure is usually smaller than the probability of erasure in a measurement-independent channel. In particular, when $\tau=1$, the measurement-independent channel is degraded by noise and thus no useful information can be obtained from the measurement-independent channel output. In contrast, in the measurement-dependent case, optimal non-adaptive querying can still accurately extract a significant number of bits in the binary expansion of the target variable.

Secondly, if the channel is a noiseless (i.e., $\tau=0$), the achievable resolution of optimal non-adaptive query procedures satisfies
\begin{align}
-d\log_2 \delta^*(n,d,\varepsilon)=n+O(\log n).
\end{align}
Note that, interestingly, for the noiseless 20 questions problem, the achievable resolution of optimal non-adaptive querying does not depend on the target excess-resolution probability $\varepsilon\in[0,1)$ for any number of queries $n$. This is in contrast to the noisy 20 questions problem where a similar phenomenon occurs only when $n\to\infty$, c.f. \eqref{phasetransition}. The implication is that that in the noiseless 20 questions problem, for any number of the queries $n\in\bbN$, the achievable resolution of optimal non-adaptive query procedures cannot be improved even if one tolerates a larger excess-resolution probability $\varepsilon$.

\begin{figure}[tb]
\centering
\includegraphics[width=.5\columnwidth]{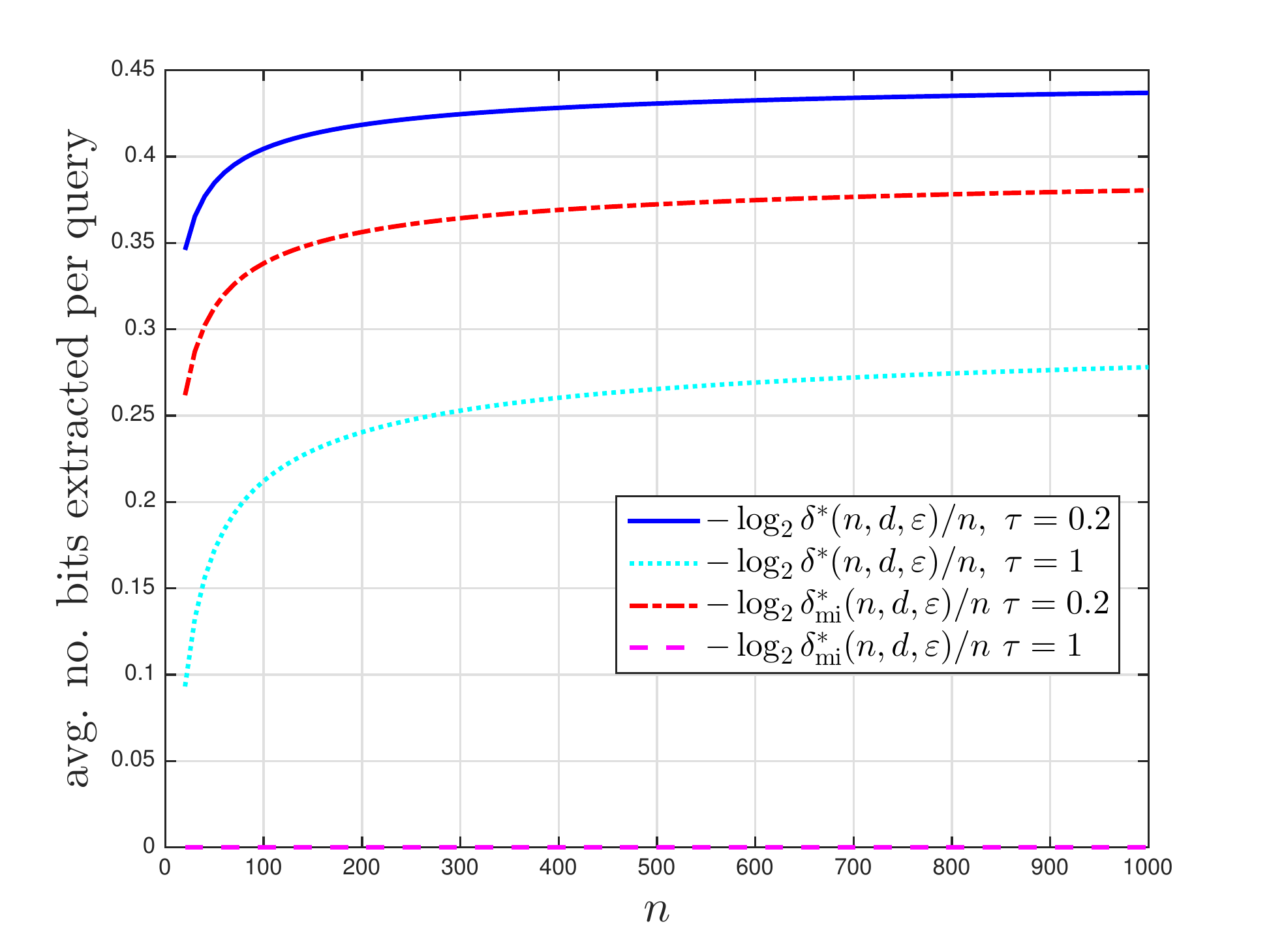}
\caption{\blue{The second-order approximation to the per query average number of bits extracted by optimal non-adaptive query procedures for both measurement-dependent ($-\log_2\delta^*(n,d,\varepsilon)/n$) and measurement-independent ($-\log_2\delta_{\rm{mi}}^*(n,d,\varepsilon)/n$) versions of the BEC with parameter $\tau$. We consider the case of $d=2$ and $\varepsilon=0.001$.}}
\label{bec_fbl}
\end{figure}

\subsection{Case of Measurement-Dependent Z-Channel}
We next consider a measurement-dependent Z-channel. Given any $\zeta\in(0,1]$ and $q\in(0,1]$, for any $(x,y)\in\{0,1\}^2$, the information density of a measurement-dependent Z-channel with parameter $\zeta$ is
\begin{align}
\imath_{q,\zeta q}(x;y)
&=\bbo(y=x=0)\log\frac{1}{1-q+\zeta q^2}+\bbo(y=0,x=1)\log\frac{\zeta q}{1-q+\zeta q^2}+\bbo(y=x=1)\log\frac{1-\zeta q}{q-\zeta q^2}.
\end{align}
The mean and the variance of the information density are respectively
\begin{align}
C(\zeta,q)
&:=\mathbb{E}[\imath_{q,\zeta q}(X;Y)]
=h_\rmb(q(1-\zeta q))-q h_\rmb(\zeta q),\\
V(\zeta,q)
\nn&:=\mathrm{V}[\imath_{q,\zeta q}(X;Y)].
\end{align}
The capacity of the measurement-dependent Z-channel with parameter $\zeta$ is
\begin{align}
C(\zeta)
&=\max_{q\in[0,1]}C(\zeta,q).
\end{align}
We plot $C(\zeta,q)$ for different values of $\zeta$ and $q\in[0,1]$ in Figure \ref{z_capacity}.
\begin{figure}[tb]
\centering
\includegraphics[width=.5\columnwidth]{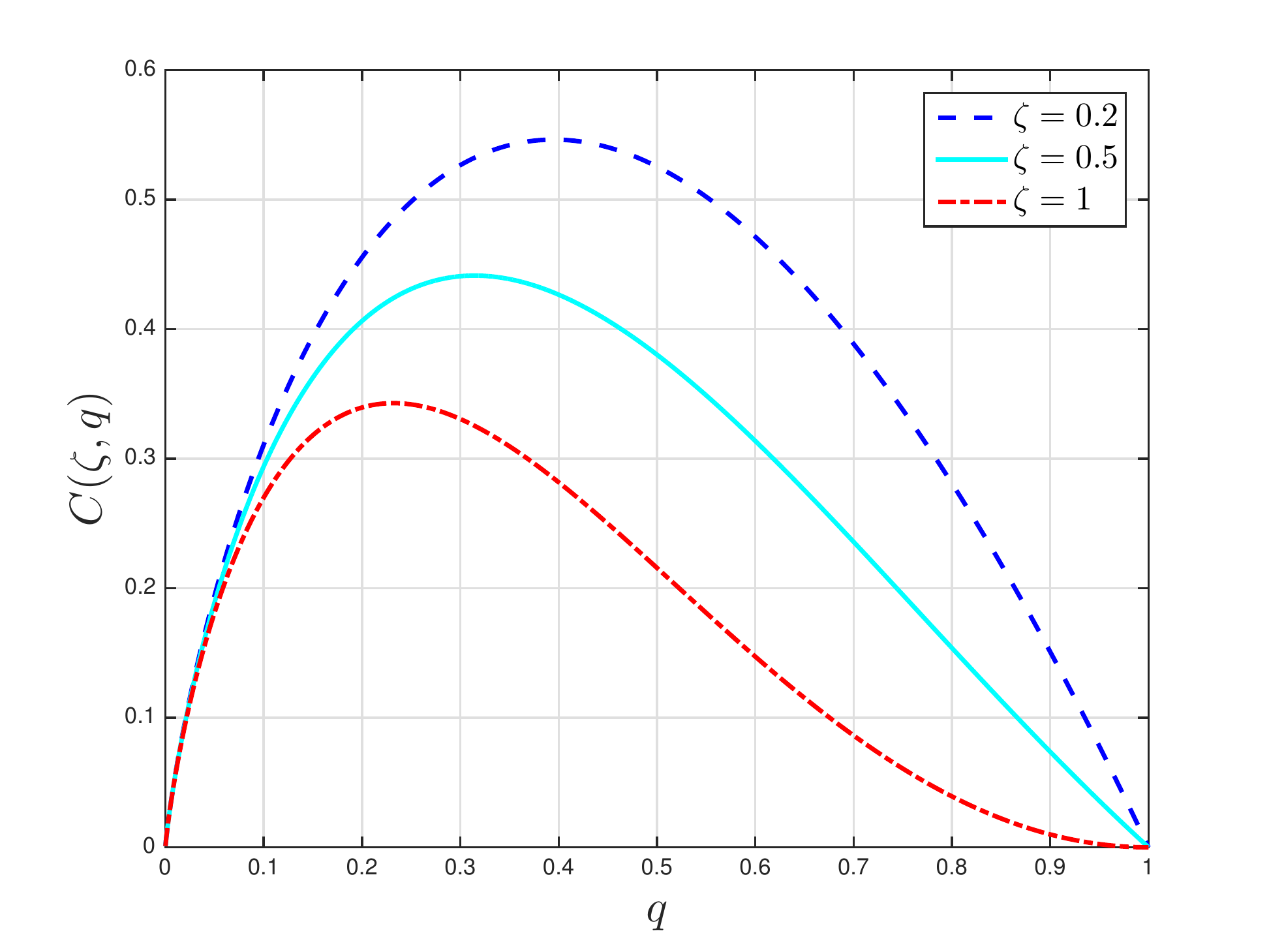}
\caption{Plot of $C(\zeta,q)$, the mean of the information density, of the measurement-dependent binary Z-channel with parameter $\zeta$ for $q\in[0,1]$. For any $\zeta\in(0,1]$, there exists a unique capacity-achieving value $q^*\in[0,1]$.}
\label{z_capacity}
\end{figure}
It can be verified that the capacity achievable parameter for the measurement-dependent Z-channel is unique and we denote the optimizer as $q^*$. Therefore, the dispersion of the Z-channel is $V(\zeta,q^*)$. Our second-order asymptotic result in Theorem \ref{result:second} specializes to the Z-channel as follows.
\begin{corollary}
\label{bz:non-adaptive}
Let $\zeta\in(0,1]$. If the channel from the oracle to the player is a measurement-dependent Z-channel with parameter $\zeta$, then Theorem \ref{result:second} holds with $C=C(\zeta,q^*)$ and $V_\varepsilon=V(\zeta,q^*)$ for any $\varepsilon\in(0,1)$.
\end{corollary}

When one considers a measurement-independent Z-channel with parameter $\zeta$, it can be easily shown that the achievable resolution $\delta^*_{\rm{mi}}(n,d,\varepsilon)$ of optimal non-adaptive query procedures satisfies
\begin{align}
-d\log \delta_{\rm{mi}}^*(n,d,\varepsilon)=nC_{\rm{mi}}(\zeta)+\sqrt{nV_{\rm{mi}}(\zeta)}\Phi^{-1}(\varepsilon)+O(\log n),
\end{align}
where $C_{\rm{mi}}(\zeta)$ and $V_{\rm{mi}}(\zeta)$ are the capacity and dispersion of the Z-channel:
\begin{align}
C_{\rm{mi}}(\zeta)&=\sup_{q\in[0,1]}h_\rmb(q(1-\zeta))-qh_\rmb(\zeta),\\
V_{\rm{mi}}(\zeta)\nn&=q^*_{\rm{mi}}(1-q^*_{\rm{mi}})\Big(\log(1-q^*_{\rm{mi}}+\zeta q^*_{\rm{mi}})\Big)^2+\zeta q^*_{\rm{mi}}(1-\zeta q^*_{\rm{mi}})\Big(\log\frac{\zeta}{1-q^*_{\rm{mi}}+\zeta q^*_{\rm{mi}}}\Big)^2\\
&\qquad+(q^*_{\rm{mi}}-\zeta q^*_{\rm{mi}})(1-q^*_{\rm{mi}}+\zeta q^*_{\rm{mi}})\Big(\log\frac{1-\zeta}{q-\zeta q^*_{\rm{mi}}}\Big),
\end{align}
with $q^*_{\rm{mi}}\in[0,1]$ being the unique optimizer of $C_{\rm{mi}}(\zeta)$.

We compare the performances of optimal non-adaptive querying for the measurement-dependent and measurement-independent Z-channels with different parameters. \blue{In Figure \ref{z_fbl}, the second-order approximation to both  $-\log \delta_{\rm{mi}}^*(n,d,\varepsilon)/n$ and $-\log \delta_{\rm{mi}}^*(n,d,\varepsilon)/n$ is plotted for $d=2$, $\varepsilon=0.001$ and different values of $\zeta$.} \blue{Once again, we find that the optimal non-adaptive query procedure for the measurement-dependent channel achieves a higher resolution than its counterpart for the measurement-independent channel.}
\begin{figure}[tb]
\centering
\includegraphics[width=.5\columnwidth]{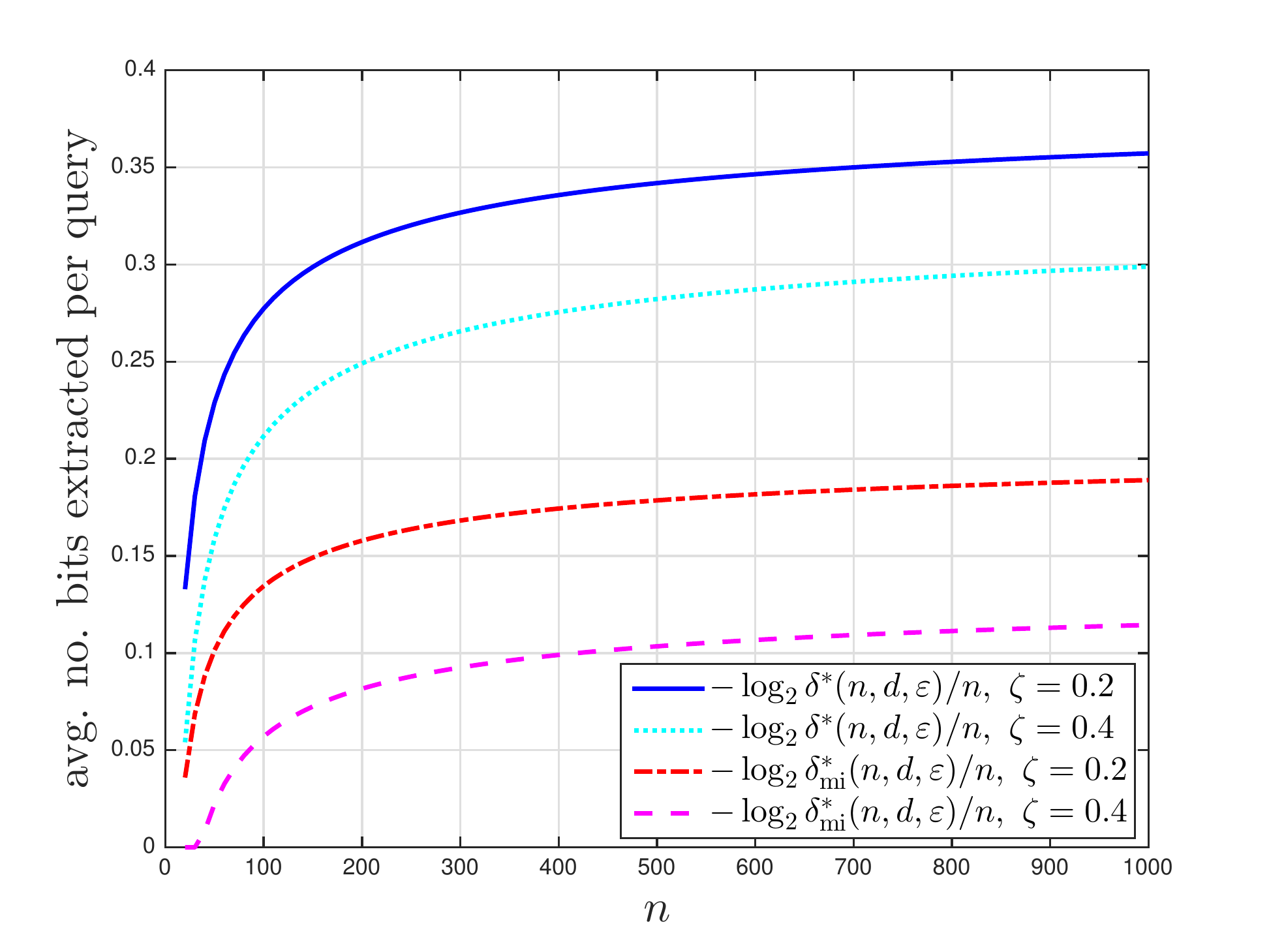}
\caption{\blue{The second-order approximation to the per query average number of bits extracted} by optimal non-adaptive query procedures for both measurement-dependent ($-\log_2\delta^*(n,d,\varepsilon)/n$) and measurement-independent ($-\log_2\delta_{\rm{mi}}^*(n,d,\varepsilon)/n$) versions of the binary Z-channel with parameter $\zeta$. Here we consider the case of $d=2$ and $\varepsilon=0.001$. Note that the per query average number of bits extracted is larger for the measurement-dependent channel.}
\label{z_fbl}
\end{figure}

{\color{blue}
\subsection{Generalization to Simultaneous Searching for Multiple Targets}
\label{sec:search:ktarget}
In this subsection, we consider the case of a simultaneous 20 questions search for multiple targets~\cite{kaspi2015searching}. Given finite integers $(k,d)\in\bbN^2$, let $(\bS_1,\ldots,\bS_k)$ be a sequence of $k$ target random vectors, where for each $i\in[k]$, the random vector $\bS_i=(S_{i,1},\ldots,S_{i,d})$ is generated independently from an arbitrary pdf $f_{\bS}$ defined $[0,1]^d$. The task is to design a query procedure to locate the $k$ targets simultaneously with as few queries as possible. 

Similarly to Definition \ref{def:procedure}, a non-adaptive query procedure is defined as follows.
\begin{definition}
\label{def:procedure:ktarget}
Given any $(n,k,d)\in\bbN^3$, $\delta\in\bbR_+$ and $\varepsilon\in[0,1]$, an $(n,k,d,\delta,\varepsilon)$-non-adaptive query procedure for noisy 20 questions consists of 
\begin{itemize}
\item $n$ queries $(\calA_1,\ldots,\calA_n)$ where each $\calA_i\subseteq[0,1]^d$,
\item and a decoder $g:\calY^n\to[0,1]^{kd}$
\end{itemize}
such that the excess-resolution probability satisfies
\begin{align}
\rmP_\rme(n,k,d,\delta)&:=\sup_{f_{\bS}\in\calF([0,1]^d)}\Pr\{\exists~(i,j)\in[k]\times[d]:~|\hatS_{i,j}-S_{i,j}|>\delta\}\leq \varepsilon\label{def:excessresolution3},
\end{align}
where $\hatS_{i,j}$ is the $j$-th coordinate of the estimate for the $i$-th target vector $\bS_i$ using the decoder $g$, i.e., $g(Y^n)=(\hat{\bS}_1,\ldots,\hat{\bS}_k)$.
\end{definition}
The minimal achievable resolution is then defined as
\begin{align}
\delta^*(n,k,d,\varepsilon)
&:=\inf\{\delta:\exists~\mathrm{an~}(n,k,d,\delta,\varepsilon)\mathrm{-non}\mathrm{-adaptive}\mathrm{~query}\mathrm{~procedure}\}.
\end{align}

The following definitions will be needed. For any $(k,M)\in\bbN^2$, let $\calL(k,M)$ be the set of all length-k vectors whose elements are ordered in increasing order and each element takes values in $[M]$, i.e.,
\begin{align}
\calL(k,M)
&:=\{(i_1,\ldots,i_k)\in[M]^k:\forall~(j,l)\in[k]^2~\mathrm{s.t.~}j<l,~i_j<i_l\}
\label{def:calL}.
\end{align}
Furthermore, given any $(d,M)\in\bbN^2$, define a function $\Gamma:[M]^d\to [M^d]$ as follows: for any $(i_1,\ldots,i_d)\in[M]^d$,
\begin{align}
\Gamma(i_1,\ldots,i_d)=1+\sum_{j\in[d]}(i_j-1)M^{d-j}\label{def:Gamma}.
\end{align}
Note that the function $\Gamma(\cdot)$ is invertible. We denote $\Gamma^{-1}:[M^d]\to [M]^d$ the inverse function.

Given any $(x_1,\ldots,x_k,y)\in[0,1]^k\times\calY$, define the joint distribution
\begin{align}
P_{X^kY}^{p,k}(x_1,\ldots,x_k,y)
&:=\Big(\prod_{j\in[k]}P_X(x_k)\Big)P_{Y|Z}^p(y|1\{\exists j\in[k]:~x_j=1\})\label{def:pjointk},
\end{align}
where $P_{Y|Z}^p$ denotes a measurement-dependent channel. Other distributions, such as $P_Y^{p,k}(\cdot)$ used below, are defined similarly to \eqref{def:pjointk}. For any $k\in\bbN$ and $\calJ\subseteq[k]$, we use $X_\calJ$ to denote the collection of random variables $X_j$ with $j\in\calJ$, use $X^n(\calJ)$ to denote the collection of sequences $X^n(j)$ with $j\in\calJ$ and use $x_\calJ$ and $x^n(\calJ)$ similarly.  For any $\gamma>0$, define the following sets of sequences
\begin{align}
\calD_\calJ^{n,k}(\gamma)
&:=\bigg\{(x^n([k]),y^n)\in\{0,1\}^{kn}\times\calY^n:~\log\frac{(P_{Y|X_{[k]}}^{p,k})^n(y^n|x^n([k]))}{(P_{Y|X_\calJ}^{p,k})^n(y^n|x^n(\calJ))}>d|\calJ|\log M+\gamma\bigg\},~\calJ\subseteq[k],\label{def:calDj}\\
\calD^{n,k}(\gamma)&:=\bigcap_{\calJ\subseteq[k]}\calD_\calJ^{n,k}(\gamma)\label{def:calD}.
\end{align}
Furthermore, given any $t\in[k]$, for any $\calJ\subseteq[t]$, define the (conditional) mutual information density
\begin{align}
\imath_\calJ^{p,t}(x_{[t]};y)
&:=
\left\{
\begin{array}{cc}
\log\frac{P_{Y|X_{[t]}}^{p,t}(y|x_{[t]})}{P_{Y|X_\calJ}^{p,t}(y|x_\calJ)}&|\calJ|<t\\
\log\frac{P_{Y|X_{[t]}}^{p,t}(y|x_{[t]})}{P_Y^{p,t}(y)}&|\calJ|=t
\end{array}
\right.
\label{def:cdmi},
\end{align}
and define the following moments of $\imath_\calJ^{p,t}(\cdot)$:
\begin{align}
C_\calJ(p,t)&:=\mathbb{E}_{P_{X^tY}^{p,t}}[\imath_\calJ^{p,t}(X_{[t]};Y)]\label{def:cjpt},\\
V_\calJ(p,t)&:=\mathrm{Var}_{P_{X^tY}^{p,t}}[\imath_\calJ^{p,t}(X_{[t]};Y)],\\
T_\calJ(p,t)&:=\mathbb{E}_{P_{X^tY}^{p,t}}[|\imath_\calJ^{p,t}(X_{[t]};Y)-C_\calJ(p,t)|^3]\label{def:tjpt}.
\end{align}
Finally, let
\begin{align}
(p^*,t^*)&:=\argmax_{p\in[0,1]}\argmin_{t\in[k]}\frac{C_{[t]}(p,t)}{t}.
\end{align}

\begin{algorithm}
\caption{\color{blue} Non-adaptive query procedure for simultaneous searching for multiple targets}
\label{procedure:nonadapt:ktaget}
\begin{algorithmic}
\REQUIRE The number of queries $n\in\bbN$ and three parameters $(M,p,\gamma)\in\bbN\times(0,1)\times\bbR_+$
\ENSURE A vectors of  estimates of the locations of $k$ targets $(\bs_1,\ldots,\bs_k)$ where each $\bs_i\in[0,1]^d$
\STATE Partition the unit cube of dimension $d$ into $M^d$ equal-sized disjoint cubes $\{\calS_{i_1,\ldots,i_d}\}_{(i_1,\ldots,i_d)\in[M]^d}$.
\STATE Generate $M^d$ binary vectors $\{x^n(i_1,\ldots,i_d)\}_{(i_1,\ldots,i_d)\in[M]^d}$ where each binary vector is generated i.i.d. from a Bernoulli distribution with parameter $p$.
\STATE $l \leftarrow 1$.
\WHILE{$l\leq n$}
\STATE Form the $l$-th query as
\begin{align*}
\calA_l:=\bigcup_{(i_1,\ldots,i_d)\in[M]^d:x_l(i_1,\ldots,i_d)=1}\calS_{i_1,\ldots,i_d}
\end{align*}
\STATE Obtain a noisy response $y_l$ from the oracle to the query $\calA_l$.
\STATE $l \leftarrow l+1$.
\ENDWHILE
\STATE $t \leftarrow k$.
\WHILE{$t>0$}
\IF{$\exists$ a tuple $(j_1,\ldots,j_t)\in\calL(t,M^d)$ such that $(X^n(\Gamma^{-1}(j_1)),\ldots,X^n(\Gamma^{-1}(j_t)))\in\calD^{n,t}(\gamma)$}
\STATE Return the locations of $k$ targets as being in subintervals $\calS_{\Gamma^{-1}(j_l)}$ with $l\in[t]$.
\STATE $t\leftarrow 0$
\ELSE
\STATE $t \leftarrow t-1$.
\ENDIF
\ENDWHILE
\end{algorithmic}
\end{algorithm}

\begin{theorem}
\label{result:second:ktarget}
Assume that i) the optimizer $(p^*,t^*)$ is unique and ii)
the third absolute moment $T_{[t^*]}(p^*,t^*)$ is finite. For any finite numbers $(k,d)\in\bbN^2$ and any $\varepsilon\in(0,1)$, the achievable resolution $\delta^*(n,k,d,\varepsilon)$ of an optimal non-adaptive query procedure satisfies
\begin{align}
-\log\delta^*(n,k,d,\varepsilon)
&=\frac{nC_{[t^*]}(p^*,t^*)+\sqrt{nV_{[t^*]}(p^*,t^*)}\Phi^{-1}(\varepsilon)+O(\log n)}{dt^*}.
\end{align}
\end{theorem}
The proof of Theorem \ref{result:second:ktarget} is provided in Appendix \ref{proof:ktarget}. \blue{Note that Theorem \ref{result:second:ktarget} refines \cite{kaspi2015searching} by i) addressing \emph{arbitrary} finite output-alphabet measurement-dependent noisy channels instead of BSC only, ii) providing a non-asymptotic (second-order asymptotic) bound and iii) considering multidimensional targets instead of one-dimensional targets. Several other remarks are in order.}

Firstly, the proof of Theorem \ref{result:second:ktarget} uses the information spectrum method~\cite{han2003information,han2006information}. In the achievability part, we analyze the performance of the non-adaptive query procedure given in Algorithm \ref{procedure:nonadapt:ktaget}, which is based on random coding and is similar to Algorithm \ref{procedure:nonadapt} except that the decoder is different.

\blue{
Secondly, for any $\varepsilon\in(0,1)$,
\begin{align}
\lim_{n\to\infty}\frac{-\log\delta^*(n,k,d,\varepsilon)}{n}=\frac{C_{[t^*]}(p^*,t^*)}{dt^*}\label{pt:k}.
\end{align}
The result in \eqref{pt:k} implies that a phase transition (strong converse) exists for non-adaptive querying for multiple multidimensional targets, with the critical resolution decay rate given by $\frac{C_{[t^*]}(p^*,t^*)}{dt^*}$.
}

Thirdly, note that
\begin{align}
C_{[t^*]}(p^*,t^*)
&=\max_{p\in[0,1]}\min_{t\in[k]}C_{[t]}(p,t)\label{ctptmin}.
\end{align}
The minimization over $t\in[k]$ in \eqref{ctptmin} follows from the fact that given a certain resolution, the number of cubes containing targets might be fewer than the total number of targets. This is because in any query procedure, we need to do a partition of the unit cube into equal-sized disjoint regions and quantize the targets $(\bs_1,\ldots,\bs_k)$. Two targets $(\bs_i,\bs_j)$ might be quantized into the same cube if they are too close with respect to a given resolution. To ensure that our result holds for all possible cases, a minimization over the number of quantized targets accounts for the worst case. Furthermore, to maximize the performance of the non-adaptive query procedure, we choose the best possible codebook by maximizing over the parameter $p\in[0,1]$. For a measurement-dependent BSC when $d=1$, it was demonstrated in \cite{kaspi2015searching} that the worst case is achieved when all $k$ targets are quantized into distinct regions, i.e., $t^*=k$.

Fourthly, when the number of targets $k$ is unknown but an upper bound $K$ is known, our results in Theorem \ref{result:second:ktarget} still hold by replacing $k$ with the upper bound $K$. Furthermore, the query procedure to achieve the theoretical performance is similar to Algorithm \ref{procedure:nonadapt:ktaget} except that $k$ should be replaced with $K$.
}

\section{Upper Bound on Resolution of Adaptive Querying}
In this section, we present a second-order asymptotic upper bound on the achievable resolution of adaptive query procedures and use this bound to discuss the benefit of adaptivity. Our result provides an approximation to the minimal achievable resolution of the adaptive query procedure in Algorithm \ref{procedure:adapt} in Appendix \ref{proof:second:fbl:adaptive}. 

Recall the definition of the capacity $C$ of measurement-dependent channels in \eqref{def:capacity}.
\begin{theorem}
\label{second:fbl:adaptive}
For any $(l,d,\varepsilon)\in\bbR_+\times\bbN\times[0,1)$,
\begin{align}
-\log\delta^*_\rma(l,d,\varepsilon)\geq \frac{lC}{d(1-\varepsilon)}+O(\log l)\label{md:adaptive}.
\end{align}
\end{theorem}
The proof of Theorem \ref{second:fbl:adaptive} is in Appendix \ref{proof:second:fbl:adaptive}. 

{\color{blue}
We make several remarks. Firstly, it is instructive to compare the performance of Algorithm \ref{procedure:adapt} with existing results~\cite{chiu2016sequential,kaspi2018searching} for the case of $d=1$. Asymptotically, we have that for any $\varepsilon\in(0,1)$, our adaptive query procedure achieves the asymptotic resolution decay rate
\begin{align}
\liminf_{l\to\infty}\frac{-\log\delta^*_\rma(l,1,\varepsilon)}{l}\geq \frac{C}{1-\varepsilon}\label{asym_a2}.
\end{align} 
On the other hand, for a measurement-dependent BSC with any parameter $\nu\in[0,1]$, the adaptive query procedures in \cite{chiu2016sequential,kaspi2018searching} achieve the asymptotic resolution decay rate
\begin{align}
\liminf_{l\to\infty}\frac{-\log\delta^*_\rma(l,1,\varepsilon)}{l}\geq C(0)\label{usec0},
\end{align} 
where $C(\cdot)$ is the channel capacity defined in \eqref{def:cnu}. For small value of excess-resolution probability $\varepsilon$ or large value of channel parameter $\nu$, it is usually true that the adaptive query procedures in \cite{chiu2016sequential,kaspi2018searching} achieve a larger resolution decay rate. However, for large value of $\varepsilon$ and small value of $\nu$, Algorithm \ref{procedure:adapt} can achieve a faster  resolution decay rate. An example is provided in Figure \ref{com_adap_asymp} to compare the asymptotic resolution decay rate of Algorithm \ref{procedure:adapt} and the sorted posterior matching (PM) algorithm in \cite{chiu2016sequential} for a measurement-dependent BSC with different parameters $\nu\in[0,1]$. Furthermore, we provide numerical simulations in Figure \ref{sim_adap} that compare the non-asymptotic simulated performance of Algorithm \ref{procedure:adapt} and the sorted PM algorithm. Note that the sorted PM algorithm was only analyzed for a measurement-dependent BSC in \cite{chiu2016sequential}.
\begin{figure}[tb]
\centering
\includegraphics[width=.5\columnwidth]{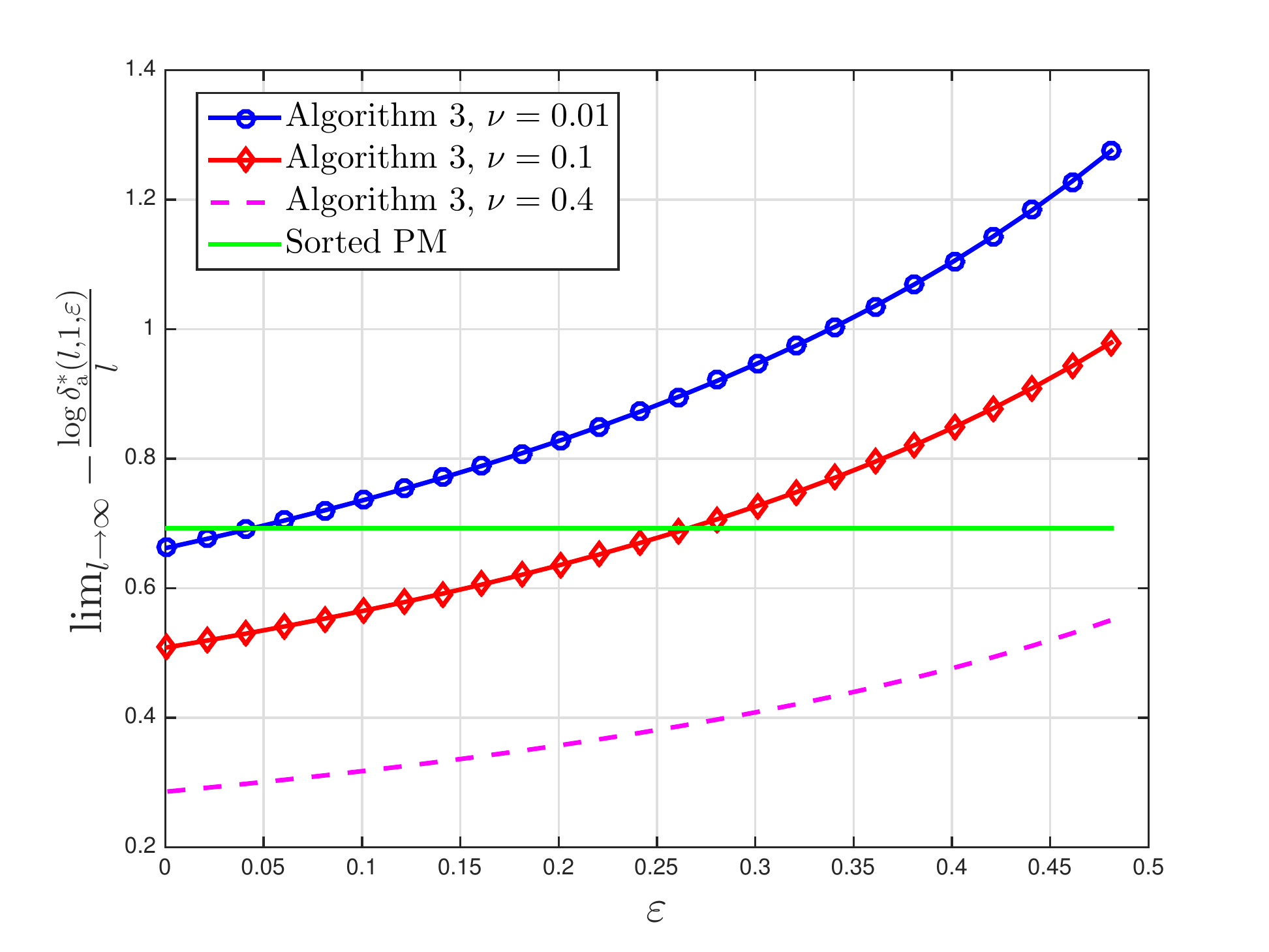}
\caption{Comparison of the asymptotic resolution decay rate of Algorithm \ref{procedure:adapt} and the sorted PM algorithm in \cite{chiu2016sequential} for a measurement-dependent BSC with parameter $\nu$ over a range of tolerable excess-resolution probability $\varepsilon$.}
\label{com_adap_asymp}
\end{figure}

Secondly, a converse bound is necessary to establish the optimality of any adaptive query algorithm under a measurement-dependent channel. However, a converse is elusive, since as pointed out in \cite{kaspi2018searching}, under the measurement-dependent channel, each noisy response $Y_i$ depends not only on the target vector $\bS$, but also the previous queries $\calA^{i-1}$ and noisy responses $Y^{i-1}$. This strong dependency makes it difficult to directly relate the current problem to channel coding with feedback~\cite{horstein1963sequential}. Indeed, under such a setting, the corresponding classical coding analogy is channel coding with feedback and with state where the state has memory. New ideas and techniques are likely required to establish a converse proof for this setting.
}

Thirdly, by comparing Theorem \ref{result:second} to Theorem \ref{second:fbl:adaptive}, we can analyze the benefit of adaptivity for the noisy 20 questions problem with measurement-dependent noise. For any $(n,d,\varepsilon)\in\bbN^2\times[0,1)$, define the benefit of adaptivity, called adaptivity gain, as
\begin{align}
\rmG(n,d,\varepsilon):=\log \delta^*(n,d,\varepsilon)-\log\delta^*_\rma(n,d,\varepsilon).
\end{align}
Using Theorems \ref{result:second} and \ref{second:fbl:adaptive}, we have
\begin{align}
\rmG(n,d,\varepsilon)
\geq\frac{1}{d}\bigg(\frac{nC\varepsilon}{1-\varepsilon}-\sqrt{nV_\varepsilon}\Phi^{-1}(\varepsilon)\bigg)+O(\log n)=:\underline{\rmG}(n,d,\varepsilon).
\end{align}
Note that $\Phi^{-1}(\varepsilon)<0$. To illustrate the adaptivity again, Figure \ref{gain_adaptivity}, we plot $\underline{\rmG}(n,d,\varepsilon)$ for $d=2$, $\varepsilon=0.001$ and three types of measurement-dependent channels with various parameters. Note that adaptive query procedures gain over non-adaptive query procedures since for the former, one can make different number of queries with respect to different realization of the target variable.
\begin{figure}[tb]
\centering
\begin{tabular}{ccc}
\hspace{-.25in} \includegraphics[width=.33\columnwidth]{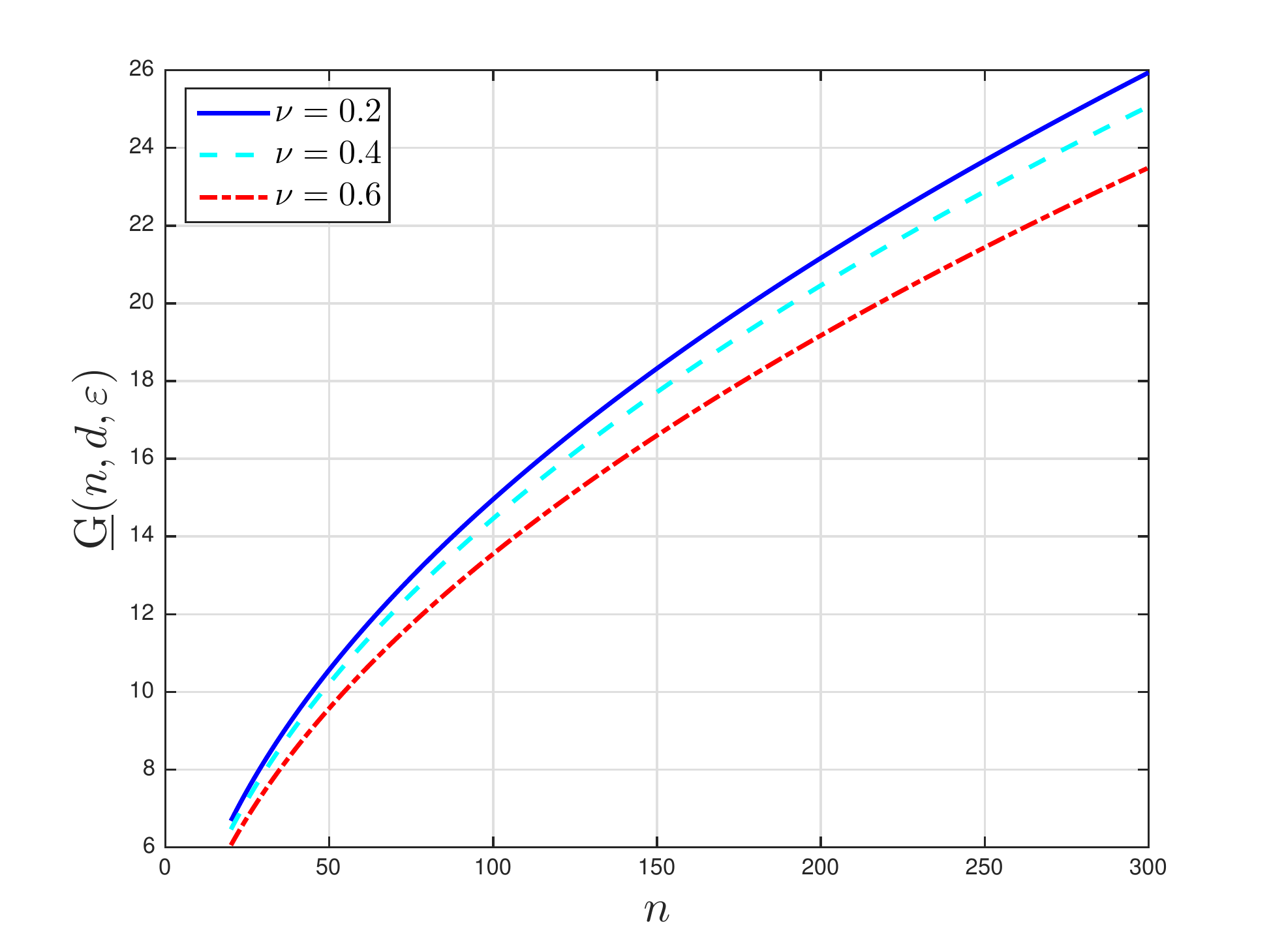}& \hspace{-.4in} \includegraphics[width=.33\columnwidth]{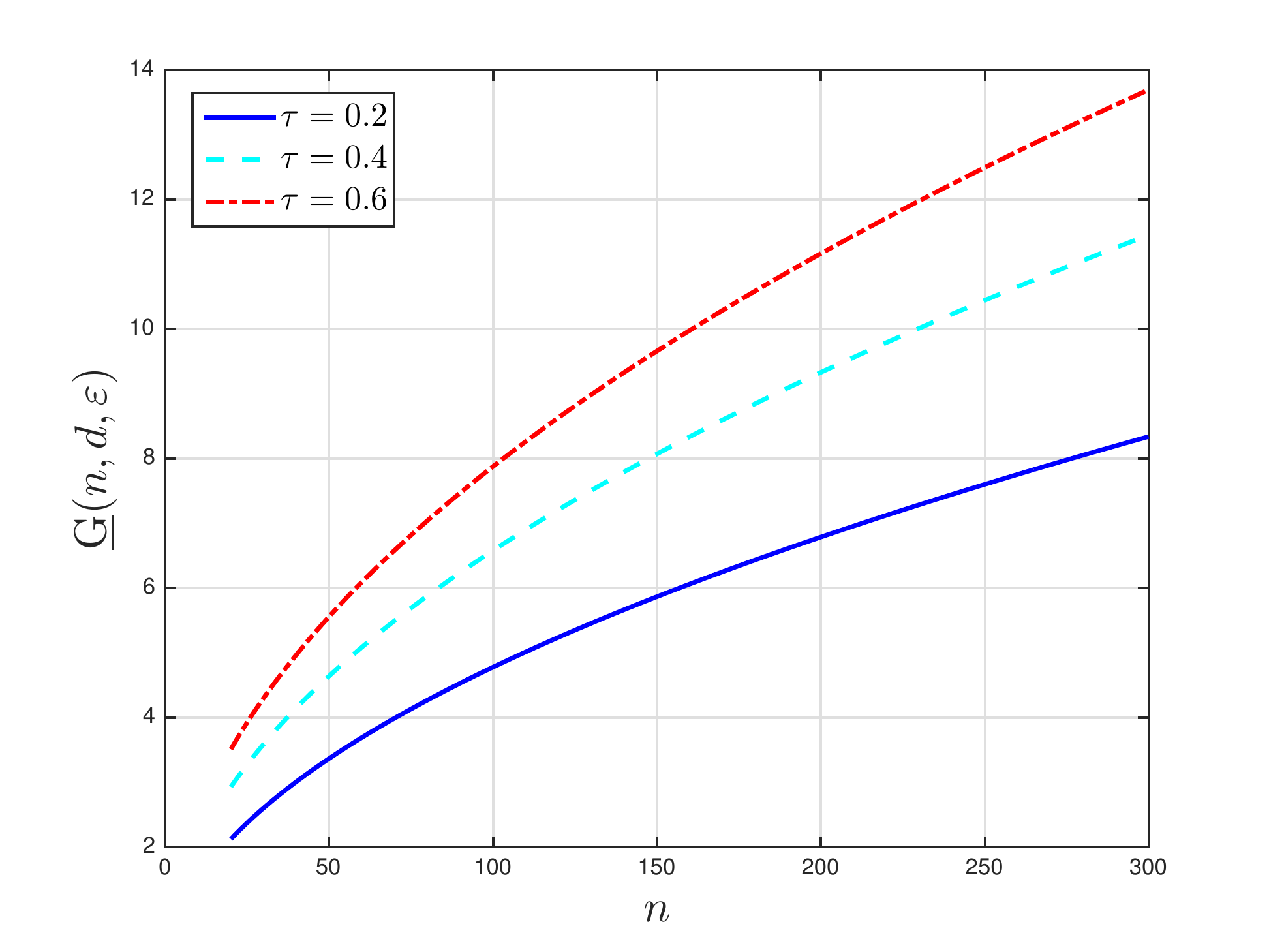}&\hspace{-.4in} \includegraphics[width=.33\columnwidth]{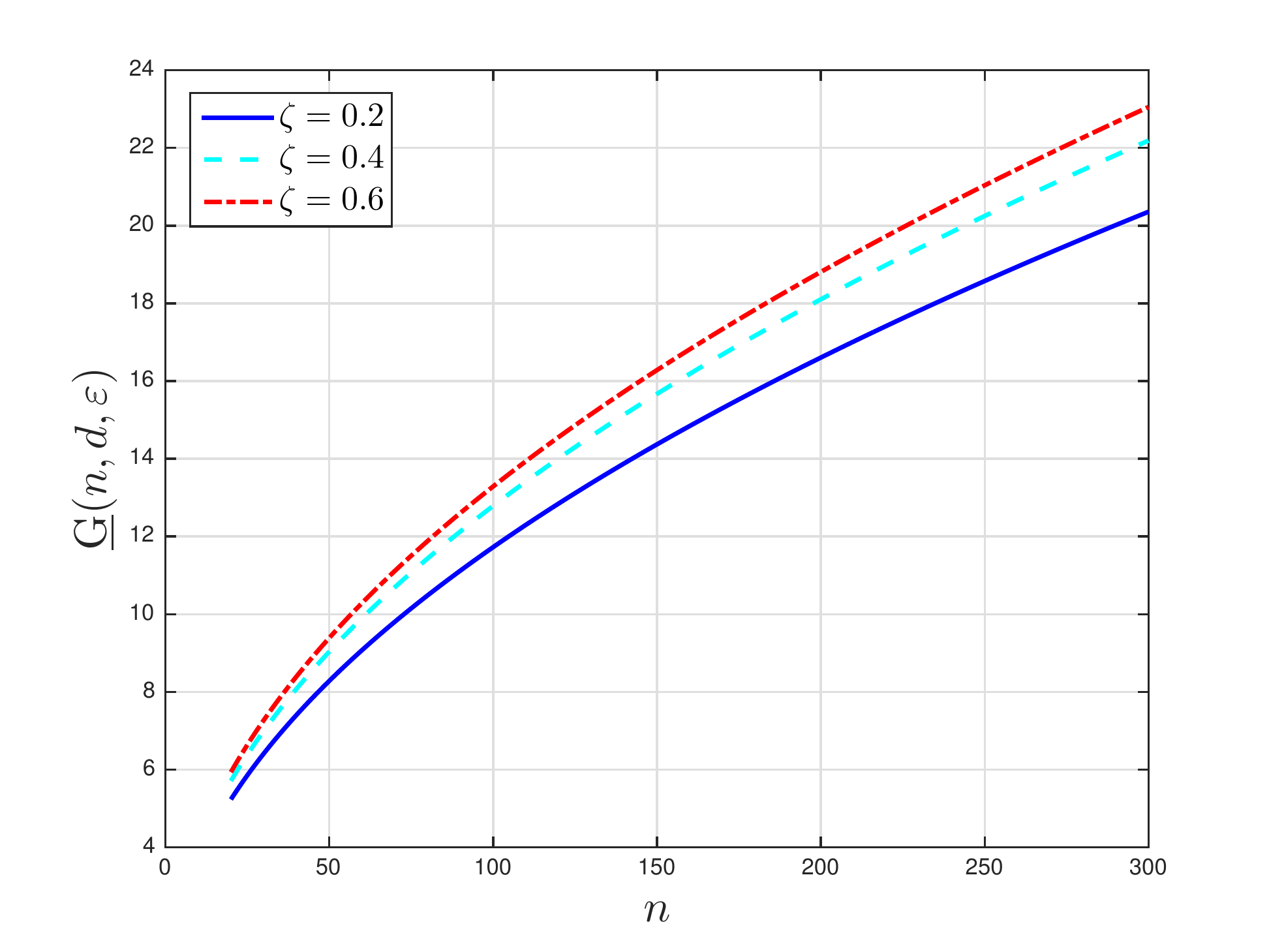}\\
\hspace{-.25in} {(a) measurement-dependent BSC} & \hspace{-.4in}  { (b) measurement-dependent BEC}& \hspace{-.4in}  { (c) measurement-dependent Z-channel}\\
\end{tabular}
\caption{Lower bound $\underline{\rmG}(n,d,\varepsilon)$ on the benefit of adaptivity  where $\underline{\rmG}(n,d,\varepsilon)=\frac{1}{d}\Big(\frac{nC\varepsilon}{1-\varepsilon}-\sqrt{nV_\varepsilon}\Phi^{-1}(\varepsilon)\Big)+O(\log n)$. The $O(\log n)$ term is not included in the plots. We consider the case of $d=2$ and $\varepsilon=0.001$.}
\label{gain_adaptivity}
\end{figure}

\blue{Finally, using the techniques in \cite{polyanskiy2011feedback} and the relationship between adaptive querying in 20 questions and channel coding with feedback}, we have that the achievable resolution $\delta^*_{\rma,\rm{mi}}(l,d,\varepsilon)$ of optimal adaptive query procedures for for measurement-independent channels satisfies
\begin{align}
-\log\delta^*_{\rma,\rm{mi}}(l,d,\varepsilon)=\frac{lC_{\rm{mi}}}{d(1-\varepsilon)}+O(\log l),\label{mi:adaptive}
\end{align}
where $C_{\rm{mi}}$ is the capacity of the measurement-independent channel. Using \eqref{md:adaptive} and \eqref{mi:adaptive}, the performances of adaptive query procedures under measurement-dependent and measurement-independent channels can be compared, analogous to the non-adaptive cases. See Appendix \ref{supp:adap} for numerical results.

\section{Numerical Illustrations}
\label{sec:numerical}

We consider a measurement-dependent BSC with parameter $\nu=0.4$ and set the target excess-resolution probability to be $\varepsilon=0.1$ in call cases.

{\color{blue}
\subsection{Searching for a Multidimensional Target over the Unit Cube}

In this subsection, we present numerical simulations to illustrate Theorem \ref{result:second} on non-adaptive searching for a multidimensional target. We consider the case where the target variable $\bS=(S_1,\ldots,S_d)$ is uniformly distributed over the unit cube of dimension $d$.
\begin{figure}[tb]
\centering
\includegraphics[width=.5\columnwidth]{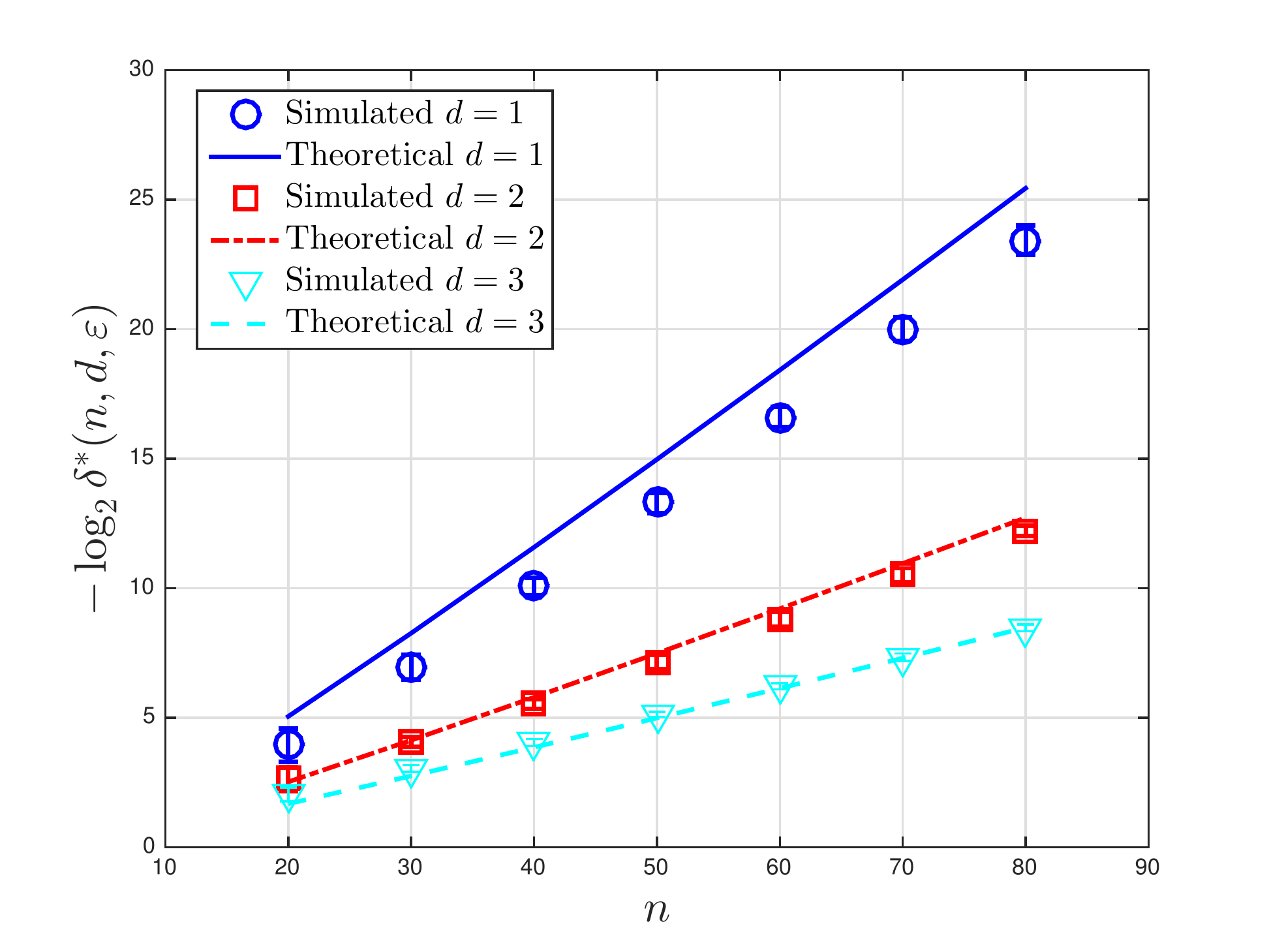}
\blue{
\caption{Minimal achievable resolution of non-adaptive query procedures for estimating a uniformly distributed target variable $\bS=(S_1,\ldots,S_d)$ in the unit cube of dimension $d$. The theoretical results correspond to the second-order asymptotic result in Theorem \ref{result:second} and the simulate results correspond to the Monte Carlo simulation of the non-adaptive query procedure in Algorithm \ref{procedure:nonadapt}. The error bar for simulated results denotes thirty standard deviations below and above the mean.}
\label{sim_non_adap_ddim}
}
\end{figure}

In Figure \ref{sim_non_adap_ddim}, the simulated achievable resolution for the non-adaptive query procedure in Algorithm \ref{procedure:nonadapt} is plotted and compared to the theoretical predictions in Theorem \ref{result:second} for several values of the dimension $d$. Given $d\in\bbN$, for each $n\in\{20,30,\ldots,80\}$, the target resolution in the numerical simulation is chosen to be the reciprocal of $M$ such that
\begin{align}
\log M=\frac{1}{d}\left(nC(\nu)+\sqrt{nV(\nu)}\Phi^{-1}(\varepsilon)\right).
\end{align}
For each number of queries $n\in\{20,30,\ldots,80\}$, the non-adaptive query procedure in Algorithm \ref{procedure:nonadapt} is run independently $10^4$ times and the achievable resolution is calculated. From Figure \ref{sim_non_adap_ddim}, we observe that 
our theoretical result in Theorem \ref{result:second} provides a good approximation to the non-asymptotic performance of the query procedure in Algorithm \ref{procedure:nonadapt}.

\begin{figure}[tb]
\centering
\includegraphics[width=.5\columnwidth]{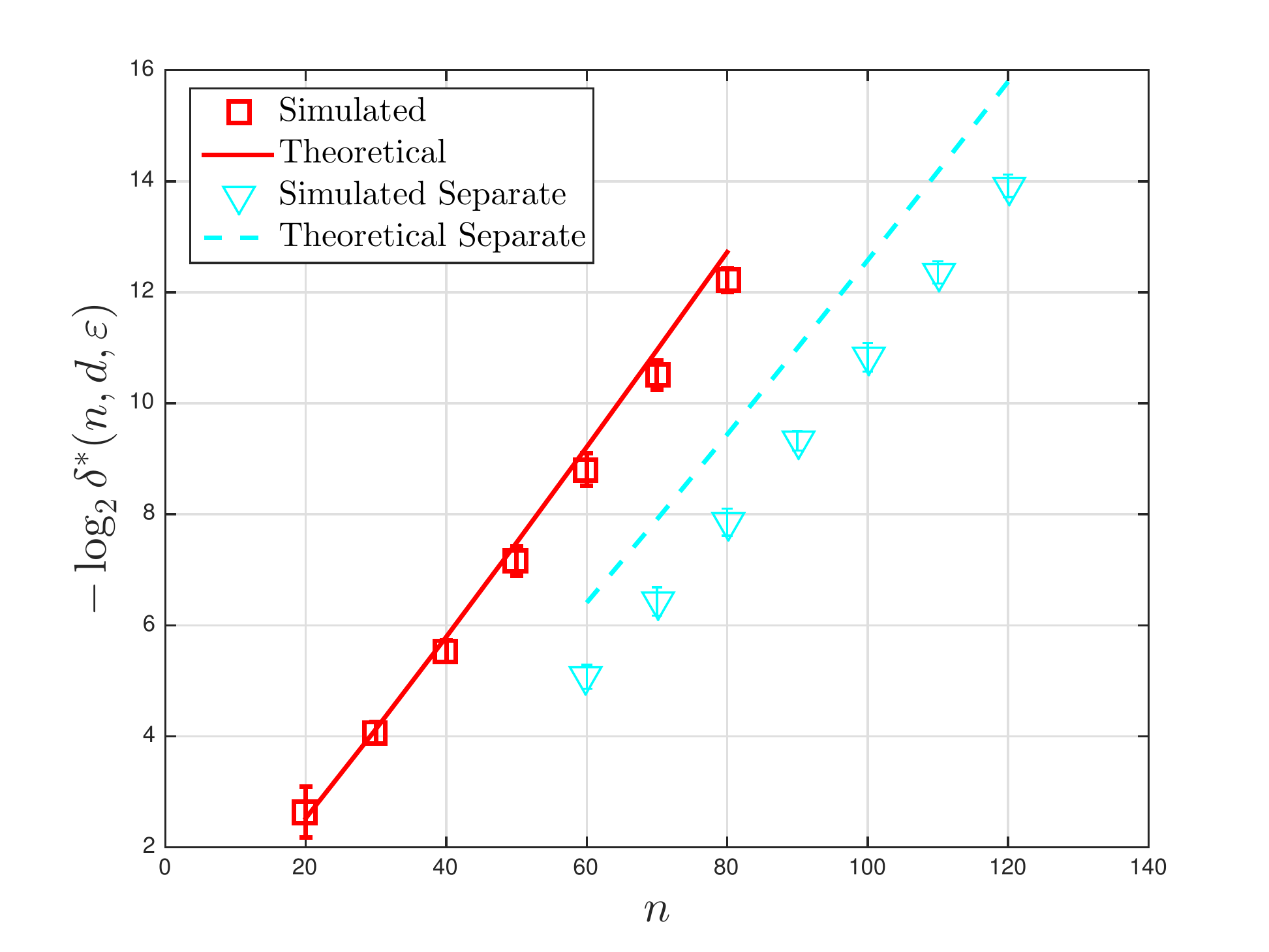}
\caption{
\blue{
Minimal achievable resolution of non-adaptive query procedures of searching for a uniformly distributed target variable $\bS=(S_1,S_2)$ over the unit cube of dimension $d=2$. The red line corresponds to the second-order asymptotic result in Theorem \ref{result:second} and the red square denotes the Monte Carlo simulation of the non-adaptive query procedure in Algorithm \ref{procedure:nonadapt}. The cyan dashed line and triangle correspond to the second-order asymptotic result in \eqref{performance:seperate} and the Monte Carlo simulation of Algorithm \ref{procedure:nonadapt} for separate searching over each dimension of $\bS$ respectively. The error bar for the simulated results denotes thirty standard deviations above and below the mean.}
\label{sim_non_adap_sep}
}
\end{figure}
In Figure \ref{sim_non_adap_sep}, for a $2$-dimensional target variable $\bS=(S_1,S_2)$, the simulated achievable resolutions is plotted for Algorithm \ref{procedure:nonadapt} and
a decoupled dimension-by-dimension search. Also shown are the theoretical predictions in Theorem \ref{result:second} and \eqref{performance:seperate} respectively. The gap between theoretical and simulated results arises since we have not accounted for the third-order term, which scales as $O(\log n)$. From Figure \ref{sim_non_adap_sep}, it could be observed that separate searching over each dimension is strictly suboptimal.

\subsection{Simultaneous Searching for Multiple Targets}
We specialize to the case where $k=2$, $d=1$ and the target random variables $(S_1,S_2)$ are both uniformly distributed over $[0,1]$. For this case, the minimal achievable resolution of the non-adaptive query procedure in Algorithm \ref{procedure:nonadapt:ktaget} is illustrated. Since we consider the measurement-dependent BSC, from \cite{kaspi2015searching}, we know that $t^*=2$.

In Figure \ref{sim_non_adap_ktarget}, the simulated achievable resolution is plotted for the non-adaptive query procedure in Algorithm \ref{procedure:nonadapt:ktaget} and compared to the theoretical predictions in Theorem \ref{result:second:ktarget}. For each $n\in\{40,45,\ldots,60\}$, the target resolution is chosen to be the reciprocal of $M$ satisfying
\begin{align}
\log M=\frac{nC_{[2]}(p^*,2)+\sqrt{nV_{[2]}(p^*,2)}\Phi^{-1}(\varepsilon)-\frac{1}{2}\log n}{2}.
\end{align}
For each $n\in\{40,45,\ldots,60\}$, the non-adaptive procedure in Algorithm \ref{procedure:nonadapt:ktaget} is run independently $10^4$ times and the achievable resolution is calculated.

\begin{figure}[tb]
\centering
\includegraphics[width=.5\columnwidth]{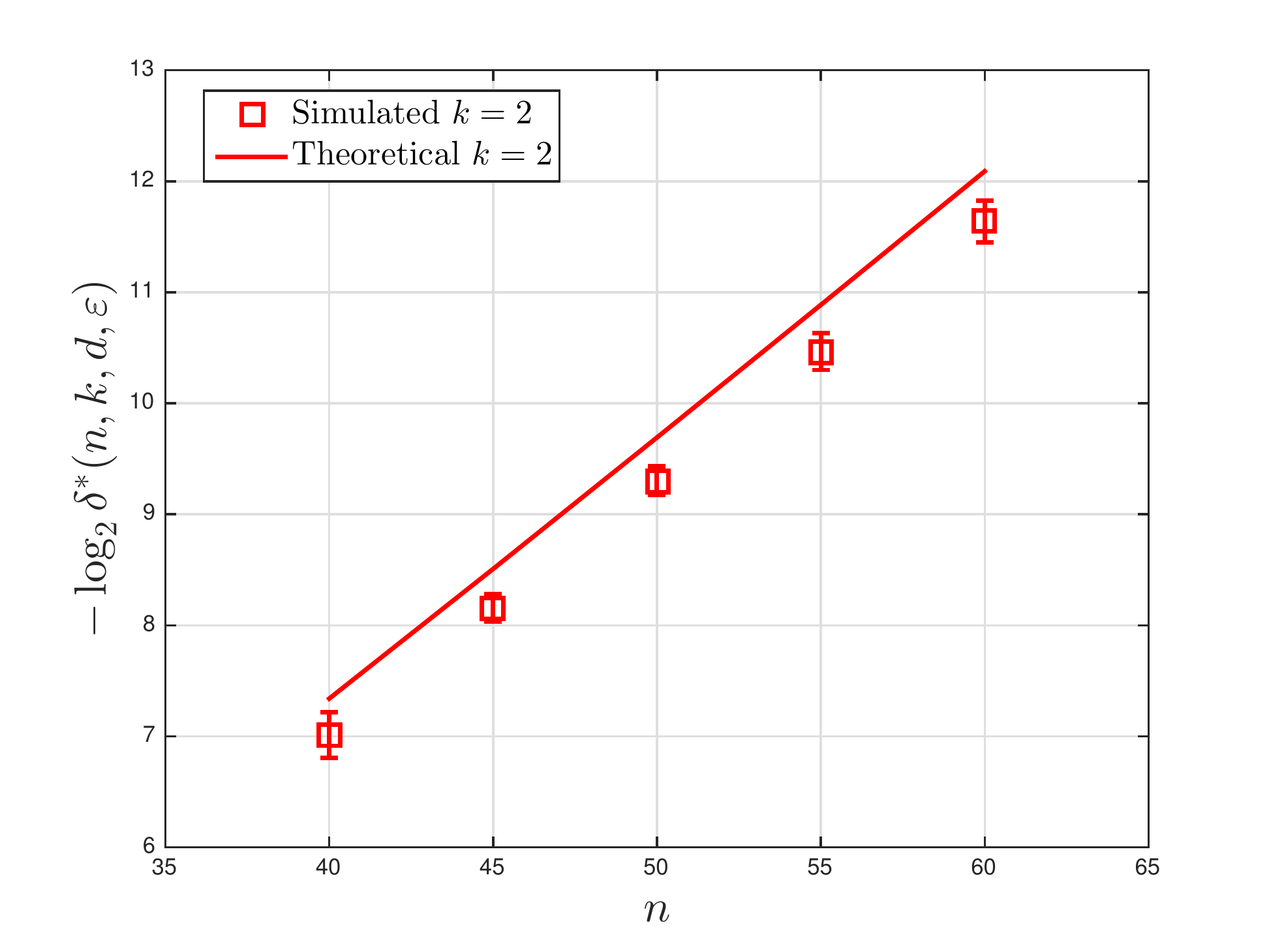}
\blue{
\caption{Minimal achievable resolution of non-adaptive query procedures for estimating $k=2$ independent one-dimensional target variables $(S_1,S_2)$ in the unit interval. The theoretical results correspond to the second-order asymptotic result in Theorem \ref{result:second:ktarget} and the simulate results correspond to the Monte Carlo simulation of the non-adaptive query procedure in Algorithm \ref{procedure:nonadapt:ktaget}. The error bar for simulated results denotes thirty standard deviations below and above the mean.}
\label{sim_non_adap_ktarget}
}
\end{figure}

\subsection{Comparisons of Adaptive Query Procedures}
We consider the case of $d=1$. In Figure \ref{sim_adap}, the \blue{simulated} achievable resolution of the adaptive query procedure in Algorithm \ref{procedure:adapt} is plotted and compared to the theoretical predictions in Theorem \ref{second:fbl:adaptive}. Furthermore, we compare our results with the simulated and asymptotic theoretical performance of the sorted PM algorithm in \cite{chiu2016sequential}. Each point of the simulated result is obtained as follows. We use $\varepsilon=0.1$ as a designed parameter. Given each $n\in\{20,30,\ldots,60\}$, the target resolution in the simulation was selected as the reciprocal of $M$ satisfying
\begin{align}
\log M=\frac{nC(\nu)}{1-\varepsilon}-\log n.
\end{align}
For each $n\in\{20,30,\ldots,60\}$, the adaptive query procedure in Algorithm \ref{procedure:adapt} is run independently $10^4$ times and the average stopping time $l_n$ is determined. 

{\color{blue}
\begin{figure}[tb]
\centering
\includegraphics[width=.5\columnwidth]{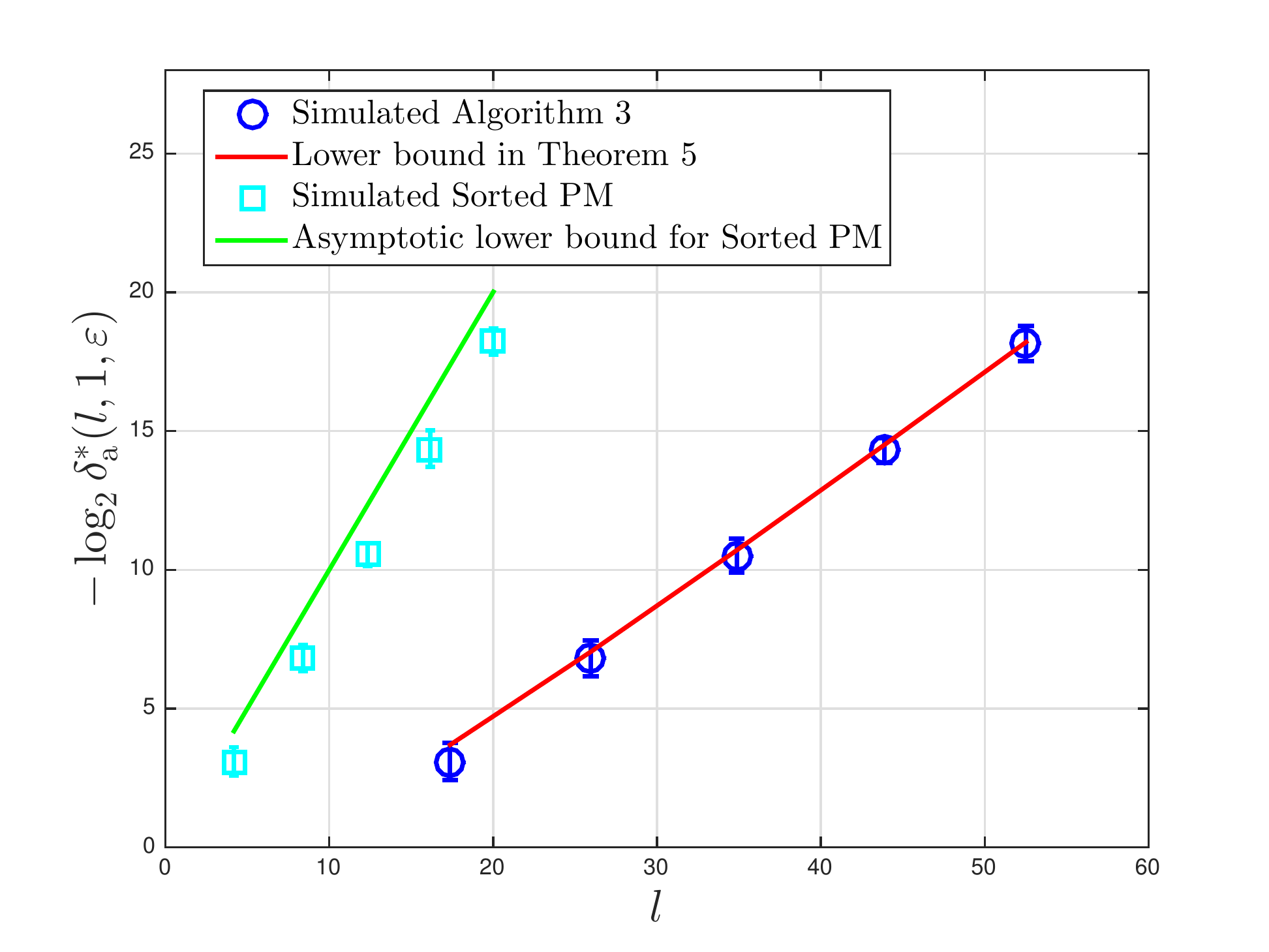}
\caption{Minimal achievable resolution of adaptive query procedures for estimating a uniformly distributed target $S\in[0,1]$ with parameter $\nu=0.4$. The red line corresponds to the theoretical value asserted in Theorem \ref{second:fbl:adaptive} and the blue dots denote the Monte Carlo simulation of the adaptive query procedure in Algorithm \ref{procedure:adapt} in Appendix \ref{proof:second:fbl:adaptive}. The green line denotes the asymptotic lower bound in \cite{chiu2016sequential} and the cyan square denotes the the Monte Carlo simulation of the sorted PM adaptive query procedure in \cite{chiu2016sequential}. The error bar for the simulated result denotes thirty standard deviations above and below the mean.}
\label{sim_adap}
\end{figure}

From Figure \ref{sim_adap}, we observe that our theoretical characterization in Theorem \ref{second:fbl:adaptive} provides a good approximation to the performance of Algorithm \ref{procedure:adapt}. Compared with the sorted PM algorithm in \cite{chiu2016sequential}, our adaptive query procedure in Algorithm \ref{procedure:adapt} is inferior in stop time if constrained to achieve the same resolution and thus the sorted PM algorithm has a better non-asymptotic performance in such a setting. However, if one considers a measurement-dependent BSC with a small parameter $\nu$ (say $\nu=0.01$), Algorithm \ref{procedure:adapt} can actually achieve a better performance when the tolerable excess-resolution probability $\varepsilon$ is relatively large, as demonstrated in the first remark after Theorem \ref{second:fbl:adaptive}. We do not include the small $\nu$ case in the numerical simulation since the time complexity of running numerical examples for such a case is high. Furthermore, for the case of $d=1$, note that the time complexity of the sorted PM algorithm is $O(M\log M)$~\cite[Table 2.1]{chiu2019noisy}, which is larger than $O(M)$ of Algorithm \ref{procedure:adapt}.}
}

\section{Conclusion}
We derived the minimal achievable resolution of non-adaptive query procedures for the noisy 20 questions problem where the channel from the oracle to the player is a measurement-dependent discrete channel. Furthermore, we generalized our results to derive bounds on the achievable resolution of adaptive query procedures and discussed the intrinsic resolution benefit due to adaptivity.

\blue{
There are several avenues for future research. Firstly, for adaptive query procedures, we derived achievability results on searching for a single target over the unit interval. It would be fruitful to apply novel techniques to derive a matching converse bound on the minimal achievable resolution of optimal adaptive query procedures. It would also be interesting to derive results for cases of searching for a multidimensional target variable and simultaneous searching for multiple targets for adaptive query procedures. Secondly, we considered discrete channel (finite output alphabet). It would be interesting to extend our results to continuous channels such as a measurement-dependent additive white Gaussian noise channel~\cite{lalitha2018improved}. Thirdly, in this paper, we were interested in fundamental limits of optimal query procedures. One can propose low-complexity practical query procedures and compare the performances of their proposed query procedures to our derived benchmarks. Finally, we considered stationary targets in the paper. In future, one can extend our results to searching for a moving target with unknown speed~\cite{6970834}.}

\appendix
\subsection{Proof of the Non-Asymptotic Achievability Bound (Theorem \ref{ach:fbl})}
\label{proof:ach}
In this subsection, we analyze the performance of the non-adaptive query procedure in Algorithm \ref{procedure:nonadapt} using ideas from channel coding~\cite{shannon1948mathematical}. To begin with, we first briefly recall the query procedure in Algorithm \ref{procedure:nonadapt}.

Fix any $M\in\bbN$, we partition the unit cube of dimension $d$ into $M^d$ equal-sized disjoint cubes $\{\calS_{i_1,\ldots,i_d}\}_{(i_1,\ldots,i_d)\in[M]^d}$. Let $\bx=\{x^n(i_1,\ldots,i_d)\}_{(i_1,\ldots,i_d)\in[M]^d}$ be a sequence of $M^d$ binary codewords. For each $t\in[n]$, the $t$-th query is designed as
\begin{align}
\calA_t
&:=\bigcup_{(i_1,\ldots,i_d)\in[M]^d:x_t(i_1,\ldots,i_d)=1}\calS_{i_1,\ldots,i_d}\label{def:query:ddim},
\end{align}
where $x_t(i_1,\ldots,i_d)$ denotes the $t$-th element of the codeword $x^n(i_1,\ldots,i_d)$. By the above query design, our $t$-th query to the oracle is whether the target $\bs=(s_1,\ldots,s_d)$ lies in the union of cubes with indices of the codewords whose $t$-th element are one. Hence, for each $t\in[n]$, the $t$-th element of each codeword can be understood as an indicator function for whether a particular cube would be queried in $i$-th question, with one being positive and zero being negative.

For subsequent analysis, given any $s\in[0,1]$, define the following quantization function
\begin{align}
\rmq(s):=\lceil sM\rceil\label{def:qs},
\end{align} 
Given any d-dimensional target variable $\bs$, we use $\bw=(w_1,\ldots,w_d)$ to denote the vector $(\rmq(s_1),\ldots,\rmq(s_d))$, i.e., $w_i=\rmq(s_i)$. Given $\bs$, the noiseless answer of the oracle to the query $\calA_t$ is
\begin{align}
Z_t
&=\bbo(\bs\in\calA_t)=\bbo\bigg(\bs\in\bigcup_{(i_1,\ldots,i_d)\in[M]^d:x_t(i_1,\ldots,i_d)=1}\calS_{i_1,\ldots,i_d}\bigg)\\
&=\bbo(x_t(\bw)=1)=x_t(\bw).
\end{align}
Then the noisy response $Y_t$ is obtained by passing $x_t(\rmq(\bs))$ over the measurement-dependent channel.

Given noisy responses $Y^n=(Y_1,\ldots,Y_n)$, the decoder produces estimates $\hatS=(\hatS_1,\ldots,\hatS_d)$ using the following two-step decoding:
\begin{enumerate}
\item the player first estimates $\bw$ as $\hat{\bW}=(\hatW_1,\ldots,\hatW_d)$ using a maximal mutual information decoder, i.e.,
\begin{align}
\hat{\bW}=(\hatW_1,\ldots,\hatW_d)=\max_{(\tili_1,\ldots,\tili_d)\in[M]^d}\imath_p(x^n(\tili_1,\ldots,\tili_d);Y^n);
\end{align}
\item the player then produces estimates $\hatS=(\hatS_1,\ldots,\hatS_d)$ as follows:
\begin{align}
\hatS_j=\frac{2\hatW_j-1}{2M}
\end{align}
for all $j\in[d]$.
\end{enumerate}

It is easy to verify that using the above query procedure, the estimate $\hatS_i$ is within $\frac{1}{M}$ of the target $s_i$ for all $i\in[d]$ if our estimate $\hat{\bW}=\bw$. Thus the excess-resolution probability of the multidimensional estimation is upper bounded by the error probability of channel coding with $M^d$ messages over a measurement-dependent codebook.

For subsequent analysis, we use $\bW=(W_1,\ldots,W_d)$ to denote the quantized vector of a target vector $\bS=(S_1,\ldots,S_d)\in[0,1]^d$, i.e., $W_i=\rmq(S_i)$ for each $i\in[d]$. We use $\bw$ to denote a particular realization. Note that each pdf $f_{\bS}\in\calF([0,1]^d)$ of the target vector $\bS$ induces a pmf $P_{\bW}\in\calP([M]^d)$. Using our query procedure, we have
\begin{align}
\nn&\sup_{f_{\bS}\in\calF([0,1]^d)}\Pr\left\{\exists~i\in[d],~|\hatS_i-S|>\frac{1}{M}\right\}\\*
&\leq \sup_{f_{\bS}\in\calF([0,1]^d)}\Pr\{\hat{\bW}\neq \bW\}\\
&\leq \sup_{P_{\bW}\in\calP([M]^d)}\Pr\{\hat{\bW}\neq \bW\}\\
&\leq \sup_{P_{\bW}\in\calP([M]^d)}\sum_{\bw}P_{\bW}(\bw)\Pr\{\exists~\bar{\bw}\in[M]^d:~\bar{\bw}\neq \bw,~\imath_p(x^n(\bar{\bw});Y^n)\geq \imath_p(x^n(\bw);Y^n)\}\label{specifydist}\\
&=:\sup_{P_{\bW}\in\calP([M]^d)}\sum_{\bw}P_{\bW}(\bw)\rmP_\rme(\bx,P_{\bW})\label{uppexcessp},
\end{align}
where the probability in \eqref{specifydist} is calculated with respect to the measurement-dependent channel 
\begin{align}
P_{Y^n|X^n}^{\calA^n}(y^n|x^n(\bw))
&=\prod_{t\in[n]}P_{Y|X}^{\calA_t}(y_t|x_t(\bw))\\
&=\prod_{t\in[n]}P_{Y|X}^{q_{t,d}^M(\bx)}(y_t|x_t(\bw))\label{useai},
\end{align}
and in \eqref{useai}, we define
\begin{align}
q_{t,d}^M(\bx):=\frac{1}{M^d}\sum_{\bw\in[M]^d}x_t(\bw).
\end{align}

Note that $\rmP_\rme(\bx,P_{\bW})$ is essentially the error probability of transmitting a message $\bW\in[M]^d$ with pmf $P_{\bW}$ over the measurement-dependent channel $P_{Y^n|X^n}^{\calA^n}$. Thus, to further bound $\rmP_\rme(\bx,P_{\bW})$, we need to analyze the error probability of a channel coding problem over a \emph{codebook dependent channel} where the channel output $Y^n$ depends on all codewords $\{x^n(i_1,\ldots,i_d)\}_{(i_1,\ldots,i_d)\in[M]^d}$. In contrast, in the classical channel coding problem, the channel output depends only on the channel input with respect to the message. However, as we shall see, using the change-of-measure technique and the assumption in \eqref{assump:continuouschannel}, with negligible loss in error probability, we can replace the measurement-dependent channel with a memoryless channel $(P_{Y|X}^p)^n$.

For this purpose, we use random coding ideas~\cite{gallager_ensemble}. Fix a Bernoulli distribution $P_X\in\calP(\{0,1\})$ with parameter $p$, i.e., $P_X(1)=p$. Let $\bX:=\{X^n(i_1,\ldots,i_d)\}_{(i_1,\ldots,i_d)\in[M]^d}$ be $M^d$ independent binary sequences, each generated i.i.d. from $P_X$. Furthermore, for any $(M,d,p,\eta)\in\bbN^2\times(0,1)\times\bbR_+$, define the following typical set of binary codewords $\bx$:
\begin{align}
\calT^n(M,d,p,\eta)
&:=\bigg\{\bx=\{x^n(i_1,\ldots,i_d)\}_{(i_1,\ldots,i_d)\in[M]^d}\in\calX^{Mdn}:\left|q_{t,d}^M(\bx)-p\right|\leq \eta,~\forall~t\in[n]\bigg\}\label{def:typical}.
\end{align}
For any $\bx\in\calT^n(M,d,p,\eta)$, recalling the query design in \eqref{def:query:ddim} and the condition in \eqref{assump:continuouschannel}, we have
\begin{align}
\log \frac{P_{Y^n|X^n}^{\calA^n}(y^n|x^n)}{(P_{Y|X}^p)^n(y^n|x^n)}
&=\sum_{t\in[n]}\log\frac{P_{Y|X}^{q_{t,d}^M(\bx)}(y_i|x_i)}{P_{Y|X}^p(y_i|x_i)}\leq n\eta c(p)\label{fromassumption}.
\end{align}

Note that given any $\bw\in[M]^d$, the joint distribution of $(\bX,Y^n)$ under the current query procedure is
\begin{align}
P_{\bX Y^n}^{\rm{md},\bw}(\bx,y^n)
&=\Big(\prod_{\bar{\bw}\in[M]^d}P_X^n(x^n(\bar{\bw}))\Big)\Big(\prod_{t\in[n]}P_{Y|X}^{q_{t,d}^M(\bx)}(y_t|x_t(\bw))\Big)\label{truedis}.
\end{align}
and furthermore, we need the following alternative joint distribution of $(\bX,Y^n)$ to apply the change-of-measure idea
\begin{align}
P_{\bX Y^n}^{p,\bw}(\bx,y^n)
&=\Big(\prod_{\bar{\bw}\in[M]^d}P_X^n(x^n(\bar{\bw}))\Big)\Big(\prod_{t\in[n]}P_{Y|X}^p(y_t|x_t(\bw))\Big)\label{altdis}.
\end{align}

For any message distribution $P_{\bW}\in\calP([M]^d)$,
\begin{align}
\mathbb{E}_{\bX}[\rmP_\rme(\bX,P_{\bW})]
&\leq \Pr\{\bX\notin\calT^n(M,d,p,\eta)\}+\mathbb{E}_{\bX}[\rmP_\rme(\bX,P_{\bW})\bbo(\bX\in\calT^n(M,d,p,\eta))]\\
&\leq 4n\exp(-2M^d\eta^2)+\mathbb{E}_{\bX}[\rmP_\rme(\bX,P_{\bW})\bbo(\bX\in\calT^n(M,d,p,\eta))]\label{useatypical},
\end{align}
where \eqref{useatypical} follows from \cite[Lemma 22]{tan2014state}, \blue{which provides an upper bound on the probability of the atypicality of i.i.d. random variables and implies that
\begin{align}
\Pr\{\bX\notin\calT^n(M,d,p,\eta)\}\leq 4n\exp(-2M^d\eta^2).
\end{align}
}
The second term in \eqref{useatypical} can be further upper bounded as follows:
\begin{align}
\nn&\mathbb{E}_{\bX}[\rmP_\rme(\bX,P_{\bW})\bbo(\bX\in\calT^n(M,d,p,\eta))]\\*
&=\sum_{\bw}P_{\bW}(\bw)\mathbb{E}_{P_{\bX Y^n}^{\rm{md},\bw}}[\bbo(\bX\in\calT^n(M,d,p,\eta))\bbo(\exists~\bar{\bw}\in[M]^d:~\bar{\bw}\neq\bw,~\imath_p(X^n(\bar{\bw});Y^n)\geq \imath_p(X^n(\bw);Y^n))]\\
&\leq \exp(n\eta c(p))\sum_{\bw}P_{\bW}(\bw)\Pr_{P_{\bX Y^n}^{p,\bw}}\{\exists~\bar{\bw}\in[M]^d:~\bar{\bw}\neq \bw,~\imath_p(X^n(\bar{\bw});Y^n)\geq \imath_p(X^n(\bw);Y^n)\}\label{cofmeasure}\\
&\leq \exp(n\eta c(p))\sum_{\bw}P_{\bW}(\bw)\sum_{\bar{\bw}\in[M]^d:\bar{\bw}\neq \bw}\Pr_{P_{\bX Y^n}^{p,\bw}}\{\imath_p(X^n(\bar{\bw});Y^n)\geq \imath_p(X^n(\bw);Y^n)\}\label{useunbound}\\
&=\exp(n\eta c(p))\mathbb{E}_{P_{X^nY^n}}[\min\{1,M^d\Pr_{P_X^n}\{\imath_p(\barX^n;Y^n)\geq \imath_p(X^n;Y^n)|X^n,Y^n\}]\}\label{rcu},
\end{align}
where \eqref{cofmeasure} follows from \eqref{fromassumption} and the change of measure technique, \eqref{useunbound} follows from the union bound, \eqref{rcu} follows by noting that the codewords $\{X^n(i_1,\ldots,i_d)\}_{(i_1,\ldots,i_d)\in[M]^d}$ are independent under $P_{\bX Y^n}^{\rm{alt}}$, the total number of codewords is no greater than $M^d$ and by applying ideas leading to the random coding union bound~\cite{polyanskiy2010finite}. In \eqref{rcu}, the joint distribution of $(X^n,Y^n)$ is
\begin{align}
P_{X^nY^n}(x^n,y^n)=\prod_{t\in[n]}P_X(x_t)P_{Y|X}^{p}(y_t|x_t).
\end{align}

Combining \eqref{useatypical} and \eqref{rcu}, we conclude that there exists a sequence of binary codewords $\bx$ such that $\rmP_\rme(\bx,P_{\bW})$ is upper bounded by the desired quantity for all message distributions $P_{\bW}\in\calP([M]^d)$ and thus the proof of Theorem \ref{ach:fbl} is completed.

\subsection{Proof of the Non-Asymptotic Converse Bound (Theorem \ref{fbl:converse})}
\label{proof:converse}

\subsubsection{Converse Proof}
Consider any sequence of non-adaptive queries $\calA^n\subseteq([0,1]^d)^n$ and any decoding function $g:\calY^n\to [0,1]^d$ such that the worst case excess-resolution probability with respect to a resolution $\delta$ is upper bounded by $\varepsilon$, i.e.,
\begin{align}
\sup_{f_{\bS}\in\calF([0,1]^d)}\Pr\big\{\exists~i\in[d]:~|\hatS_i-S_i|>\delta\big\}\leq\varepsilon\label{error4converse}.
\end{align}
As a result, for uniformly distributed target vector $\bS=(S_1,\ldots,S_d)$, the excess-resolution probability with respect to $\delta$ is also upper bounded by $\varepsilon$. In the rest of the proof, we consider a \emph{uniformly} distributed $d$-dimensional target $\bS$.

Let $\beta$ be any real number such that $\beta\leq\frac{1-\varepsilon}{2}\leq 0.5$ and let $\tilM:=\lfloor\frac{\beta}{\delta}\rfloor$. Define the following quantization function
\begin{align}
\rmq_\beta(s):=\lceil s\tilM\rceil,~\forall~s\in\calS\label{def:qbeta}.
\end{align}

Given any queries $\calA^n\in([0,1]^d)^n$, the noiseless responses from the oracle are $X^n=(X_1,\ldots,X_n)$ where for each $t\in[n]$, $X_t=\bbo(\bS\in\calA_t)$ is a Bernoulli random variable with parameter being the volume of $\calA_t$, which this follows from the definition of the measurement-dependent channel and the fact that the target variable $\bS$ is uniformly distributed. The noisy responses $Y^n$ is the output of passing $X^n$ over the measurement-dependent channel $P_{Y^n|X^n}^{\calA^n}$. Finally, an estimate $\hat{\bS}=(\hatS_1,\ldots,\hatS_d)$ is produced using the decoding function $g$.

For simplicity, let $\bW:=(W_1,\ldots,W_d)=(\rmq_\beta(S_1),\ldots,\rmq_\beta(S_d))$ and let $\hat{\bW}:=(\rmq_\beta(\hatS_1),\ldots,\rmq_\beta(\hatS_d))$. \blue{Similarly to} \cite{kaspi2018searching}, we have that
\begin{align}
\Pr\{\hat{\bW}\neq \bW\}
&=\Pr\{\hat{\bW}\neq \bW,~\exists~i\in[d]:~|\hatS_i-S_i|>\delta\}+\Pr\{\hat{\bW}\neq \bW,~\forall~i\in[d]:~|\hatS_i-S_i|\leq\delta\}\\
&\leq \Pr\{\exists~i\in[d]:~|\hatS_i-S_i|>\delta\}+\Pr\{\hat{\bW}\neq \bW,~\forall~i\in[d]:~|\hatS_i-S_i|\leq\delta\}\\
&\leq \varepsilon+\Pr\{\hat{\bW}\neq \bW,~\forall~i\in[d]:~|\hatS_i-S_i|\leq\delta\}\label{useerror4converse}\\
&\leq \varepsilon+\Pr\{\exists~i\in[d]:~\hatW_i\neq W_i~\mathrm{and}~|\hatS_i-S_i|\leq\delta\}\\
&\leq \varepsilon+\sum_{i\in[d]}\Pr\{\hatW_i\neq W_i~\mathrm{and}~|\hatS_i-S_i|\leq\delta\}\\
&\leq \varepsilon+2d\delta\tilM\label{boundaryerror}\\
&\leq\varepsilon+2d\beta\label{use:tilM},
\end{align}
where \eqref{useerror4converse} follows from \eqref{error4converse}, \eqref{boundaryerror} follows since i) only when $S_i$ is within $\delta$ to the boundaries (left and right) of the sub-interval with indices $W_i=\rmq_{\beta}(S_i)$ can the events $\hatW_i\neq W_i$ and $|\hatS_i-S_i|\leq \delta$ occur simultaneously, ii) $S_i$ is uniformly distributed over $\calS$ and thus iii) the probability of the event $\{\hatW_i\neq W_i,~|\hatS_i-S_i|\leq \delta\}$ is upper bounded by $2\delta\tilM$, and \eqref{use:tilM} follows from the definition of $\tilM$. To ease understanding of the critical step \eqref{boundaryerror}, we have provided a figure illustration in Figure \ref{figureillus4converse}.

\begin{figure}[tb]
\centering
\setlength{\unitlength}{0.5cm}
\scalebox{1}{
\begin{picture}(20,3)
\linethickness{1pt}
\put(0,0.5){\makebox(0,0){...}}
\put(1,0){\line(1,0){18}}
\put(20,0.5){\makebox(0,0){...}}
\put(1,0){\line(0,1){1}}
\put(7,0){\line(0,1){1}}
\put(13,0){\line(0,1){1}}
\put(19,0){\line(0,1){1}}
\put(8,0){\line(0,1){0.5}}
\put(9,0){\line(0,1){0.5}}
\put(10,0){\line(0,1){0.5}}
\put(11,0){\line(0,1){0.5}}
\put(12,0){\line(0,1){0.5}}
\put(6,0){\line(0,1){0.5}}
\put(14,0){\line(0,1){0.5}}
\put(1.1,1.5){\makebox(0,0){$(k-2)\delta/\beta$}}
\put(7.1,1.5){\makebox(0,0){$(k-1)\delta/\beta$}}
\put(13,1.5){\makebox(0,0){$ k\delta/\beta$}}
\put(19.1,1.5){\makebox(0,0){$(k+1)\delta/\beta$}}
\put(7,0){\transparent{0.3}\color{green}{\rule{\unitlength}{1\unitlength}}}
\put(12,0){\transparent{0.3}\color{blue}{\rule{\unitlength}{1\unitlength}}}
\end{picture}
}
\caption{Figure illustration of \eqref{boundaryerror} for $S_1$. Let $\delta=\frac{1}{600}$ and $\beta=\frac{1}{6}$. Thus, we partition the unit interval $[0,1]$ into $\tilM=100$ sub-intervals each with length $\frac{1}{100}$. In the figure, we plot three consecutive sub-intervals with indices $(k-1,k,k+1)$ for some $k\in[2:\tilM-1]$. Note that the $k$-th interval starts from $\frac{(k-1)\delta}{\beta}$ and end at $\frac{k\delta}{\beta}$ and contains $\frac{1}{\beta}$ small intervals, each of length $\delta$. Suppose $S_1$ lies in $k$-th sub-interval, then only if $S_1$ is with $\delta=\frac{1}{600}$ of the boundaries in $k$-th sub-interval, denoted with shaded color, can we find $\hatS_1$ in adjacent sub-interval such that $|\hatS_1-S_1|\leq \delta$ and $\hatW_1=\rmq_{\beta}(\hatS_1)\neq\rmq_{\beta}(S_1)=W_1$.}
\label{figureillus4converse}
\end{figure}
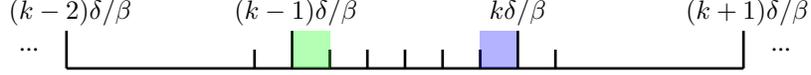

Using \eqref{use:tilM}, we have that the excess-resolution probability of searching for a multidimensional target variable is lower bounded by 
\begin{align}
\varepsilon
&\geq \Pr\{\hat{\bW}\neq \bW\}-2d\beta\label{conversechannel}.
\end{align}

Similarly to the definition of $\Gamma(\cdot)$ in \eqref{def:Gamma}, we define the function $\tilde{\Gamma}:[\tilM]^d\to[\tilM^d]$ as
\begin{align}
\tilde{\Gamma}(i_1,\ldots,i_d)
&=1+\sum_{j\in[d]}i_j\tilM^{d-j}.
\end{align}
Using \eqref{conversechannel}, we have
\begin{align}
\varepsilon
&\geq \Pr\big\{\tilde{\Gamma}(\hat{\bW})\neq \tilde{\Gamma}(\bW)\big\}-2d\beta\label{usegammatocon}\\
&=\Pr\{\hatW\neq W\}-2d\beta\label{finalequation}
\end{align}
where \eqref{usegammatocon} follows since $\tilde{\Gamma}(\cdot)$ is a one-to-one mapping from $[\tilM]^d$ to $[\tilM^d]$ and in \eqref{finalequation}, we define $\hatW=\Gamma(\hat{\bW})\in[M^d]$ and define $W=\Gamma(\bW)\in[M^d]$ similarly. Note that from the problem formulation, since $\bS$ is uniformly distributed over $[0,1]^d$, we have that $\bW$ is uniformly distributed over $[\tilM]^d$ and thus $W$ is uniformly distributed over $[M^d]$.

Note that given queries $\calA^n$, the probability $\Pr\{\hatW\neq W\}$ is the average error probability of channel coding with deterministic states when the distribution of the channel inputs is $P_{X^n}^{\calA^n}$ and the message $W$ is uniformly distributed over $[\tilM^d]$. Therefore, we can use converse bounds for channel coding to bound achievable resolution $\delta$ (via $\tilM$).

Similar as \cite[Proposition 4.4]{TanBook} \blue{which provides a finite blocklength converse bound for the channel coding problem}, we have that for any $\kappa\in(0,1-\varepsilon-2d\beta)$, 
\begin{align}
\log \tilM^d
&\leq\inf_{Q_{Y^n}\in\calP(\calY^n)}
\sup\bigg\{t\Big|\Pr\bigg\{\log\frac{P_{Y^n|X^n}^{\calA^n}(Y^n|X^n)}{Q_{Y^n}(Y^n)}\leq t\bigg\}\leq \varepsilon+2d\beta+\kappa\bigg\}-\log\kappa\label{fbl:converse:final}\\
&=\sup\bigg\{t\Big|\Pr\bigg\{\sum_{i\in[n]}\log\frac{P_{Y|X}^{\calA_i}(Y_i|X_i)}{P_Y^{|\calA_i|,|\calA_i|}(Y_i)}\leq t\bigg\}\leq \varepsilon+2d\beta+\kappa\bigg\}-\log\kappa\label{nconverse},
\end{align} 
where \eqref{nconverse} follows by choose $Q_Y^n$ being the marginal distribution of $Y^n$ induced distribution of $P_{X^n}^{\calA^n}$ and the measurement-dependent channel $P_{Y^n|X^n}^{\calA^n}$.  Note that \eqref{fbl:converse:final} is slightly different from \cite[Proposition 4.4]{TanBook}. In fact, we follow the proof of \cite[Proposition 4.4]{TanBook} with $M$ replaced by $\tilM^d$ and $\varepsilon$ replaced by $\varepsilon+2d\beta$ till the left hand side of  
\cite[Eq. (4.18)]{TanBook}. Then, we use the definition of the so called $\varepsilon$-hypothesis testing divergence~\cite[Eq. (2.9)]{TanBook}.

Since \eqref{nconverse} holds for any sequence of queries $\calA^n\in[0,1]^{nd}$ and any decoder $g:\calY^n\to[0,1]^d$ satisfying \eqref{error4converse}, recalling the definition of $\tilM$ and the definition of $\imath_{\calA_i}(\cdot)$, we have
\begin{align}
-d\log\delta\leq -d\log\beta-\log\kappa+\sup_{\calA^n\in[0,1]^{nd}}\sup\bigg\{t\in\bbR_+\Big|\Pr\Big\{\sum_{i\in[n]}\imath_{\calA_i}(X_i;Y_i)\leq t\bigg\}\leq \varepsilon+2d\beta+\kappa\Big\}.
\end{align}

\subsection{Proof of Second-Order Asymptotics (Theorem \ref{result:second})}
\label{proof:second}
\subsubsection{Achievability Proof}
Invoking Theorem \ref{ach:fbl} with the capacity-achieving parameter $q\in\calP_{\rm{ca}}$, we have that for any $\eta\in\bbR_+$, there exists a non-adaptive query procedure with $n$ queries such that
\begin{align}
\rmP_\rme^n\left(\frac{1}{M}\right)
&\leq 4n\exp(-2M^d\eta^2)+\exp(n\eta c(q))\mathbb{E}[\min\{1,M^d\Pr\{\imath_q(\barX^n;Y^n)\geq \imath_q(X^n;Y^n)\}\}]\label{step1}.
\end{align}
We first bound the expectation term in \eqref{step1} as follows:
\begin{align}
\nn&\mathbb{E}[\min\{1,M^d\Pr\{\imath_q(\barX^n;Y^n)\geq \imath_q(X^n;Y^n)\}\}]\\*
&\leq \Pr\left\{M^d\exp(-\imath_q(X^n;Y^n))\geq \frac{1}{\sqrt{n}}\right\}+\frac{1}{\sqrt{n}}\label{usemany}\\
&=\Pr\left\{d\log M-\imath_q(X^n;Y^n)\geq -\log \sqrt{n}\right\}+\frac{1}{\sqrt{n}}\\
&=\Pr\Big\{\sum_{i\in[n]}\imath_{q,q}(X_i;Y_i)\leq d\log M+\log(\sqrt{n})\Big\}+\frac{1}{\sqrt{n}}\label{step2},
\end{align}
where \eqref{usemany} follows from i) the change of measure technique which states that
\begin{align}
\Pr\{\imath_q(\barX^n;y^n)\geq t\}
&=\sum_{\barx^n}P_X^n(\barx^n)\bbo(\imath_q(\barx^n;y^n)\geq t)\leq \sum_{\barx^n} P_{X|Y}^q(\barx^n|y^n)\exp(-t)=\exp(-t),
\end{align}
and ii) the result in \cite[Eq. (37)]{scarlett2017mismatch} saying that $\mathbb{E}[\min\{1,J\}]\leq \Pr\{J>\frac{1}{\sqrt{n}}\}+\frac{1}{\sqrt{n}}$ for any $n\in\bbN$.

Now choose $M$ such that
\begin{align}
d\log M=nC+\sqrt{nV_\varepsilon}\Phi^{-1}(\varepsilon)-\frac{1}{2}\log n,
\end{align}
and let 
\begin{align}
\eta=\sqrt{\frac{d\log M}{2M^d}}=O\left(\frac{\sqrt{n}}{\exp(nC/2)}\right).
\end{align}
Thus, we have
\begin{align}
4n\exp(-2M^d\eta^2)
&=\frac{4n}{M^d}=
4\exp\bigg(-nC-\sqrt{nV_\varepsilon}\Phi^{-1}(\varepsilon)+\frac{3}{2}\log n\bigg)=O(\exp(-nC))\label{result1},
\end{align}
and
\begin{align}
\exp(n\eta c(q))
&=1+n\eta c(q)+o(n\eta c(q))
=1+O\left(\frac{n^{3/2}}{\exp(nC/2)}\right)\label{result2}.
\end{align}

Finally, applying the Berry-Esseen theorem to \eqref{step2}, we have that for any $q\in\calP_{\rm{ca}}$ and any $\varepsilon\in[0,1)$,
\begin{align}
\mathbb{E}[\min\{1,M^d\Pr\{\imath_{q}(\barX^n;Y^n)\geq \imath_{q}(X^n;Y^n)\}\}]
&\leq \varepsilon+O\left(\frac{1}{\sqrt{n}}\right)\label{result3}.
\end{align}

Combining \eqref{step1} and the results in \eqref{result1} to \eqref{result3}, we have that for $n$ sufficiently large,
\begin{align}
-d\log\delta^*(n,d,\varepsilon)
&\geq nC+\sqrt{nV}\Phi^{-1}(\varepsilon)-\frac{1}{2}\log n.
\end{align}

\subsubsection{Converse Proof}
We now proceed with the converse proof. Given any $\varepsilon\in[0,1)$, for any $\beta\in(0,\frac{1-\varepsilon}{2})$ and any $\kappa\in(0,1-\varepsilon-2d\beta)$, from Theorem \ref{fbl:converse}, we have
\begin{align}
-d\log \delta^*(n,d,\varepsilon)\leq -d\log \beta-\log\kappa+\sup_{\calA^n\in[0,1]^{nd}}\sup\bigg\{t\Big|\Pr\Big\{\sum_{i\in[n]}\imath_{\calA_i}(X_i;Y_i)\leq t\bigg\}\leq \varepsilon+2d\beta+\kappa\Big\}\label{step1:c}.
\end{align}
We first analyze the probability term in \eqref{step1:c}. Given any sequence of queries $\calA^n$, let
\begin{align}
C_{\calA^n}&:=\frac{1}{n}\sum_{i\in[n]}\mathbb{E}[\imath_{\calA_i}(X_i;Y_i)]\label{def:CA^n},\\
V_{\calA^n}&:=\frac{1}{n}\sum_{i\in[n]}\mathrm{Var}[\imath_{\calA_i}(X_i;Y_i)],\\
T_{\calA^n}&:=\frac{1}{n}\sum_{i\in[n]}\mathbb{E}[|\imath_{\calA_i}(X_i;Y_i)-\mathbb{E}[\imath_{\calA_i}(X_i;Y_i)]|^3],
\end{align}
Assume that there exists $V_->0$ such that $V_-\leq V_{\calA^n}$. Applying the Berry-Esseen theorem~\cite{berry1941accuracy,esseen1942liapounoff}, we have that
\begin{align}
\sup\bigg\{t \Big|\Pr\Big\{\sum_{i\in[n]}\imath_{\calA_i}(X_i;Y_i)\leq t\Big\}\leq \varepsilon+2d\beta+\kappa\bigg\}\leq nC_{\calA^n}+\sqrt{nV_{\calA^n}}\Phi^{-1}\bigg(\varepsilon+2d\beta+\kappa +\frac{6T_{\calA^n}}{\sqrt{nV_-^3}}\bigg)\label{step2:c}.
\end{align}
Let $\beta$ and $\kappa$ be chosen so that
\begin{align}
d\beta=\kappa=\frac{1}{\sqrt{n}}.
\end{align}

Using \eqref{step1:c} and \eqref{step2:c}, we have
\begin{align}
-d\log\delta^*(n,d,\varepsilon)
&\leq \log n+\sup_{\calA^n\in[0,1]^{nd}}\bigg(nC_{\calA^n}+\sqrt{nV_{\calA^n}}\Phi^{-1}\bigg(\varepsilon+\frac{2}{\sqrt{n}} +\frac{6T_{\calA^n}}{\sqrt{nV_-^3}}\bigg)\bigg)\label{touseinl2}.
\end{align}

For any sequence of queries $\calA^n$, we have
\begin{align}
C_{\calA^n}
&\leq \sup_{\calA\subseteq[0,1]^d}\mathbb{E}[\imath_{\calA}(X;Y)]=\sup_{p\in[0,1]}\mathbb{E}[\imath_{p}(X;Y)]=C\label{max:C1}.
\end{align}
\blue{Combining \eqref{touseinl2} and \eqref{max:C1}}, when $n$ is sufficiently large, for any $\varepsilon\in[0,1)$,
\begin{align}
-d\log\delta^*(n,d,\varepsilon)
&\leq \log n+\sup_{\calA^n:|\calA_i|=q^*,~\forall i\in[n]}\bigg(nC_{\calA^n}+\sqrt{nV_{\calA^n}}\Phi^{-1}\bigg(\varepsilon+\frac{2}{\sqrt{n}} +\frac{6T_{\calA^n}}{\sqrt{nV_-^3}}\bigg)\bigg)\label{touseinl3}\\
&=\log n+nC+\sqrt{nV_\varepsilon}\Phi^{-1}\bigg(\varepsilon+\frac{2}{\sqrt{n}}+\frac{6T_{\calA^n}}{\sqrt{nV_-^3}}\bigg)\label{usedefveps}\\
&=nC+\sqrt{nV_\varepsilon}\Phi^{-1}(\varepsilon)+\log n+O(1)\label{taylor},
\end{align}
where \blue{\eqref{touseinl3} follows since i) for any $i\in[n]$, the maximum value of $\mathbb{E}[\imath_{\calA_i}[X_i;Y_i]]$ is achieved by any query $\calA_i$ with size $q^*$ which achieves the capacity $C$ and ii) when $n$ is sufficiently large, $nC_{\calA^n}=\frac{1}{n}\sum_{i\in[n]}\mathbb{E}[\imath_{\calA_i}[X_i;Y_i]]$ is the dominant term in the supremum, \eqref{usedefveps} follows from the definition of $V_\varepsilon$ in \eqref{def:veps} and \eqref{taylor} follows from the Taylor's expansion of $\Phi^{-1}(\cdot)$ (cf. \cite[Eq. (2.38)]{TanBook}) and the fact that $T_{\calA^n}$ is finite for discrete random variables $X^n$ and $Y^n$.}

{\color{blue}

\subsection{Proof of Theorem \ref{result:second:ktarget}}
\label{proof:ktarget}

\subsubsection{Achievability Proof}
Recall the definitions of the quantization function $\rmq(\cdot)$ in \eqref{def:qs}. Similarly to Appendix \ref{proof:ach}, we the partition of the unit cube $[0,1]^d$ into 
$M^d$ equal-sized disjoint cubes $\{\calS_{i_1,\ldots,i_d}\}_{(i_1,\ldots,i_d)\in[M]^d}$. Furthermore, let $\bX=\{X^n(i_1,\ldots,i_d)\}_{(i_1,\ldots,i_d)\in[M]^d}$ be a sequence of $M^d$ binary random vectors where each vector is generated i.i.d. from a Bernoulli distribution $P_X$ with parameter $p\in[0,1]$ to be specified. The design of each query is exactly the same as in \eqref{def:query:ddim}. Given $k$ $d$-dimensional target vectors $\bs^k:=(\bs_1,\ldots,\bs_k)$ where for each $i\in[k]$, $\bs_i=(s_{i,1},\ldots,s_{i,d})$, let $w_{i,l}=\rmq(s_{i,l})$ be the quantized version for each $i\in[k]$ and $l\in[d]$ and let $\bw_i=(w_{i,1},\ldots,w_{i,d})$ denote the collection of quantized locations of each dimension of the target $\bs_i$.

It is possible that two targets can be quantized into the same partition, i.e., there exists a pair of distinct indices $(i,j)\in[k]^2$ such that $\bw_i=\bw_j$. If such an event occurs, then the detected number of targets would be smaller than $k$. However, this should be not considered as an error since all the targets can be located within desired resolution even under such scenarios. Recall the definition of $\Gamma(\cdot)$ in \eqref{def:Gamma}. In subsequent analysis, we define the set of unique quantized locations as
\begin{align}
\calW_\rmp(\bs^k):=\{w\in[M^d]:~\exists~i\in[k]~\mathrm{s.t.~}\Gamma(\bw_i)=w\}\label{def:wp}.
\end{align}
and define the number of distinct quantized targets as
\begin{align}
k_\rmp(\bs^k):=|\calW_\rmp(\bs^k)|\label{def:kp}.
\end{align}
Furthermore, let $w^{\uparrow}_\rmp(\bs^k)$ be the collection of elements in $\calW_\rmp(\bs^k)$ in an increasing order and let $w^{\uparrow}_\rmp(\bs^k,i)$ be the $i$-th element.

For each $l\in[n]$, the noiseless answer to the query $\calA_l$ (cf. \eqref{def:query:ddim}) is
\begin{align}
Z_l
:&=1\{\exists~i\in[k]:~\bs_i\in\calA_l\}=1\bigg\{\exists~i\in[k]:~\bs_i\in\bigcup_{(i_1,\ldots,i_d)\in[M]^d:X_l(i_1,\ldots,i_d)=1}\calS_{i_1,\ldots,i_d}\bigg\}\\
&=1\{\exists~i\in[k]:~X_l(\bw_i)=1\}.
\end{align}
Then the noisy response $Y_i$ from the oracle to each query $\calA_i$ is obtained by passing $Z_i$ over the measurement-dependent channel $P_{Y|Z}^{\calA_i}$. For ease of analysis, in subsequent proof, we omit the step of the change-of-measure and assume that the measurement-dependent channel is $P_{Y|Z}^p$. The influence of this change-of-measure step is ignorant for relatively large number of queries $n\in\bbN$ as demonstrated in Appendices \ref{proof:ach} and \ref{proof:second}. 

Recall the definitions of $\calL(k,M)$ in \eqref{def:calL}, $\Gamma(\cdot)$ in \eqref{def:Gamma} and $\calD^{n,k}(\gamma)$ in \eqref{def:calD} . Given noisy responses $Y^n=(Y_1,\ldots,Y_n)$, if there exists an unique tuple $(i_1,\ldots,i_k)\in\calL(k,M^d)$ such that $(X^n(\Gamma^{-1}(i_1)),\ldots,X^n(\Gamma^{-1}(i_k)),Y^n)\in\calD^{n,k}(\gamma)$, the decoder produces estimates of quantized values $(\hat{\bW}_1,\ldots,\hat{\bW}_k)$ such that for each $i\in[k]$, $\hat{\bW}_i=(\hatW_{i,1},\ldots,\hatW_{i,d})=\Gamma^{-1}(i_k)$ and then estimates the target vector $\bS_i$ as $\hat{\bS}_i$ where for each $j\in[d]$, $\hat{\bS}_{i,j}=\frac{2\hatW_{i,j}-1}{2M}$; if more than one such tuple exists, then randomly pick one; if no such tuple exists, then the process continues in finding whether a unique tuple $(i_1,\ldots,i_t)$ exists for all $t\in[k-1]$ (with descending $t$) until $t=0$.

Note that given the above non-adaptive query procedure, an error (excess-resolution event for some target) is made if one of the following three events occurs
\begin{itemize}
\item $\calE_1:$ $(X^n(\Gamma^{-1}(w^\uparrow_\rmp(\bs^k,1))),\ldots,X^n(\Gamma^{-1}(w^\uparrow_\rmp(\bs^k,k_\rmp(\bs^k))),Y^n)\notin\calD^{n,k_\rmp(\bs^k)}(\gamma)$;
\item $\calE_2$: there exists an tuple $(i_1,\ldots,i_{k_\rmp(\bs^k)})\in\calL(k_\rmp(\bs^k),M^d)$ such that $(i_1,\ldots,i_{k_\rmp(\bs^k)})\neq w^\uparrow_\rmp(\bs^k)$ and the following condition satisfies $(X^n(\Gamma^{-1}(i_1)),\ldots,X^n(\Gamma^{-1}(i_{k_\rmp(\bs^k)})),Y^n)\in\calD^{n,k_\rmp(\bs^k)}(\gamma)$;
\item $\calE_3:$ for some $t\in[k_\rmp(\bs^k)+1:k]$, there exists an tuple $(j_1,\ldots,j_t)\in\calL(t,M^d)$ such that $(X^n(\Gamma^{-1}(j_1)),\ldots,X^n(\Gamma^{-1}(j_t)),Y^n)\in\calD^{n,t}(\gamma)$.
\end{itemize}

We will analyze the error probability of the second and third events as follows. First, let us focus on the event $\calE_2$. For ease of analysis, for each $j\in[k_\rmp(\bs^k)]$, define
\begin{align}
\calI_j(\bs^k)
\nn&:=\Big\{(i_1,\ldots,i_{k_\rmp(\bs^k)})\in\calL(k_\rmp(\bs^k),M^d):(i_1,\ldots,i_{k_\rmp(\bs^k)})\neq w^{\uparrow}_\rmp(\bs^k)\\
&\qquad\qquad\qquad\mathrm{and~}\big|l\in[k_\rmp(\bs^k)]:~i_l\neq w^{\uparrow}_\rmp(\bs^k,l)\big|=j\Big\}.
\end{align}

Similarly to \cite{han2003information,han2006information}, using the information spectrum method, we have
\begin{align}
\Pr\{\calE_2\}
&:=\sum_{j\in[k_\rmp(\bs^k)]}\Pr\big\{\exists (i_1,\ldots,i_{k_\rmp(\bs^k)})\in\calI_j(\bs^k):~(X^n(\Gamma^{-1}(i_1)),\ldots,X^n(\Gamma^{-1}(i_{k_\rmp(\bs^k)})),Y^n)\in\calD^{n,k_\rmp(\bs^k)}(\gamma)\big\}\\
&\leq\sum_{j\in[k_\rmp(\bs^k)]}\sum_{(i_1,\ldots,i_{k_\rmp(\bs^k)})\in\calI_j(\bs^k)}\Pr\big\{(X^n(\Gamma^{-1}(i_1)),\ldots,X^n(\Gamma^{-1}(i_{k_\rmp(\bs^k)})),Y^n)\in\calD^{n,k_\rmp(\bs^k)}(\gamma)\big\}\\
\nn&=\sum_{j\in[k_\rmp(\bs^k)]}\sum_{(i_1,\ldots,i_{k_\rmp(\bs^k)})\in\calI_j(\bs^k)}
\sum_{(x^n(\Gamma^{-1}(i_1)),\ldots,x^n(\Gamma^{-1}(i_{k_\rmp(\bs^k)})),y^n)\in\calD^{n,k_\rmp(\bs^k)}(\gamma)}
\Big(\prod_{t\in[k_\rmp(\bs^k)]:i_t\neq w^{\uparrow}_\rmp(\bs^k,t)}P_X^n(x^n(\Gamma^{-1}(i_t)))\Big)\\*
&\quad\qquad\times (P_{\{X_t\}_{t\in[k_\rmp(\bs^k)]:i_t= w^{\uparrow}_\rmp(\bs^k,t)}Y}^{p,k})^n(\{x^n(\Gamma^{-1}(t))\}_{t\in[k_\rmp(\bs^k)]:i_t=w^{\uparrow}_\rmp(\bs^k,t)},y^n)\label{indefromps}\\
\nn&\leq \sum_{j\in[k_\rmp(\bs^k)]}\sum_{(i_1,\ldots,i_{k_\rmp(\bs^k)})\in\calI_j(\bs^k)}
\sum_{(x^n(\Gamma^{-1}(i_1)),\ldots,x^n(\Gamma^{-1}(i_{k_\rmp(\bs^k)})),y^n)\in\calD^{n,k_\rmp(\bs^k)}(\gamma)}\Big(\prod_{t\in[k_\rmp(\bs^k)]}P_X^n(x^n(\Gamma^{-1}(i_t)))\Big)\\*
&\quad\qquad\times (P_{Y|X^{k_\rmp(\bs^k)}}^{p,k})^n(y^n|x^n(\Gamma^{-1}(i_1)),\ldots,
x^n(\Gamma^{-1}(i_{k_\rmp(\bs^k)})))\frac{\exp(-\gamma)}{M^{dj}}\label{usedefDj}\\
&=\sum_{j\in[k_\rmp(\bs^k)]}\sum_{(i_1,\ldots,i_{k_\rmp(\bs^k)})\in\calI_j(\bs^k)}\frac{\exp(-\gamma)}{M^{dj}}\\
&\leq |k_\rmp(\bs^k)|\exp(-\gamma),\label{finalhaha}
\end{align}
where \eqref{indefromps} follows since the noisy response $Y^n$ is only dependent on binary vectors $X^n(j)$ where $j\in w^{\uparrow}_\rmp(\bs^k)$, \eqref{usedefDj} follows from the definition of $\calD^{n,k}(\dot)$ in \eqref{def:calDj}, and \eqref{finalhaha} follows since the size of $\calI_j(\bs^k)$ is no greater than $M^{dj}$.

Similarly to steps leading to \eqref{finalhaha}, we can show that
\begin{align}
\Pr\{\calE_3\}\leq (k-k_\rmp(\bs^k))\exp(-\gamma).
\end{align}
Therefore, for any $\bs^k=(\bs_1,\ldots,\bs_k)$ and any $M\in\bbN$, using the above non-adaptive query procedure, we have that
\begin{align}
\nn&\Pr\left\{\exists (i,j)\in[k]\times[d]:~|\hatS_{i,j}-s_{i,j}|>\frac{1}{M}\right\}\\*
&\leq k\exp(-\gamma)+\Pr\{(X^n(\Gamma^{-1}(w_\rmp^\uparrow(\bs^k,1))),\ldots,X^n(\Gamma^{-1}(w_\rmp^\uparrow(\bs^k,k_\rmp(\bs^k)))),Y^n)\notin\calD^{n,k_\rmp(\bs^k)}(\gamma)\}\label{ktarget:step1}.
\end{align}
Our analysis then focuses on the second term in \eqref{ktarget:step1}. Without loss of generality, we assume that $k_\rmp(\bs^k)=t$ for some $t\in[k]$ and let $w^{\uparrow}_\rmp(\bs^k)=[1:t]$.
Recall the definitions of the (conditional) mutual information density in \eqref{def:cdmi}
and its statistics in \eqref{def:cjpt} to \eqref{def:tjpt}. Then we have
\begin{align}
\nn&\Pr\{(X^n(\Gamma^{-1}(1)),\ldots,X^n(\Gamma^{-1}(t)),Y^n)\notin\calD^{n,t}(\gamma)\}\\*
&\leq \sum_{\calJ\subseteq[t]}\Pr\{(X^n(\Gamma^{-1}(1)),\ldots,X^n(\Gamma^{-1}(t)),Y^n)\notin\calD_\calJ^{n,t}(\gamma)\}\\
&=\sum_{\calJ\subseteq[t]}\Pr\Big\{\sum_{i\in[n]}\imath_\calJ^{p,t}(X_i(\Gamma^{-1}(1)),\ldots,X_i(\Gamma^{-1}(t));Y_i)\leq |\calJ|\log M+\gamma\Big\}.
\end{align}
Now choose $M$ and $\gamma$ such that for some $\varepsilon\in(0,1)$,
\begin{align}
\gamma&=\frac{1}{2}\log n\\
t\log M&=nC_{[t]}(p,t)+\sqrt{nV_{[t]}(p,t)}\Phi^{-1}(\varepsilon)-\frac{1}{2}\log n\label{chooseM:ktarget}.
\end{align}
Using the Berry-Esseen theorem, we have that
\begin{align}
\Pr\Big\{\sum_{i\in[n]}\imath_{[t]}^{p,t}(X_i(\Gamma^{-1}(1)),\ldots,X_i(\Gamma^{-1}(t));Y_i)\leq t\log M+\gamma\Big\}
&\leq \varepsilon+\frac{6T_{[t]}(p,t)}{(\sqrt{nV_{[t]}(p,t)})^3}.
\end{align}
Furthermore, note that for any $t\in[2:k]$ and $\calJ\subseteq[t]$ such that $|\calJ|\leq t-1$, from \cite{kaspi2015searching}, we have
\begin{align}
\frac{C_{\calJ}(p,t)}{|\calJ|}>\frac{C_{[t]}(p,t)}{t}\label{imineq}.
\end{align}
Thus, for any $\calJ\subset[t]$ satisfying $|\calJ|\leq t-1$, $\Pr\big\{\sum_{i\in[n]}\imath_\calJ^{p,t}(X_i(\Gamma^{-1}(1)),\ldots,X_i(\Gamma^{-1}(t));Y_i)\leq |\calJ|\log M+\gamma\big\}$ vanishes exponentially fast according to large deviations theorems~\cite{dembo2009large}. Therefore, we have shown that when $n$ is sufficient large, with the $M$ chosen in \eqref{chooseM:ktarget} and when there are $t$ distinct quantized target variables, the excess-resolution probability is no greater than $\varepsilon$.

Noting that there can be at most  $k$ distinct quantized targets, under the assumption of Theorem \ref{result:second:ktarget}, we conclude that our non-adaptive query procedure can achieve excess-resolution probability $\varepsilon$ with respect to a resolution $\delta$ where
\begin{align}
-\log \delta
&\geq  \max_{p\in[0,1]}\min_{t\in[k]}\frac{nC_{[t]}(p,t)+\sqrt{nV_{[t]}(p,t)}\Phi^{-1}(\varepsilon)-\log n}{t}\\
&=\frac{nC_{[t^*]}(p^*,t^*)+\sqrt{nV_{[t^*]}(p^*,t^*)}\Phi^{-1}(\varepsilon)+O(\log n)}{t^*}.
\end{align}
Finally, the existence of a deterministic code results from the simple fact that $\mathbb{E}[X]\leq a$ implies there exists $x\leq a$ for any random variable $X$ and real number $a\in\bbR$.}

{\color{blue}

\subsubsection{Converse Proof}
The converse proof of Theorem \ref{result:second:ktarget} follows similarly as the achievability proof by combining the ideas in \cite{kaspi2015searching} and \cite[Lemma 4]{han2006information} and thus we only emphasize the difference here.

Consider any non-adaptive query procedure with queries $\calA^n$ and decoder $g:\calY^n\to[0,1]^{dk}$ such that the worst-case excess-resolution probability with respect to $\delta\in\bbR_+$ is upper bounded by $\varepsilon\in(0,1)$, i.e.,
\begin{align}
\sup_{f_{\bS}\in\calF([0,1]^d)}\Pr\big\{\exists~(i,j)\in[k]\times[d]:~|\hatS_{i,j}-S_{i,j}|>\delta\big\}\leq \varepsilon\label{ktarget:step00},
\end{align}
where for each $(i,j)$, $S_{i,j}$ is the $j$-th element of the $i$-th target variable $\bS_i$ and $\hatS_{i,j}$ is its estimate produced by the decoder.

In the rest of the proof, we will consider the case where each target variable is generated from a \emph{uniform} distribution $f_d^\rmU$ over the unit cube of dimension $d$, i.e., $[0,1]^d$. For any $\beta\leq\frac{\varepsilon}{2}\leq 0.5$, let $\tilM:=\lfloor \frac{\beta}{\delta}\rfloor$. Recall the definition of the quantization function $\rmq_\beta(\cdot)$ in \eqref{def:qbeta}. 

Given queries $\calA^n\in[0,1]^{dn}$, the noiseless responses $X^n$ is a sequence of independent random variables where for each $t\in[n]$, $X_t=\bbo(\exists~i\in[k]:~\bS_i=(S_{i,1},\ldots,S_{i,d})\in\calA_t)$ is a Bernoulli random variable with parameter $1-(1-(\mathrm{Vol}(\calA_t)))^k$. For simplicity, let for each $(i,j)\in[k]\times[d]$, let $W_{i,j}:=\rmq_\beta(S_{i,j})$, $\hatW_{i,j}:=\rmq_\beta(\hatS_{i,j})$. Furthermore, for each $i\in[k]$, we use $\bW_i$ to denote $(W_{i,1},\ldots,W_{i,d})$ and use $\hat{\bW_i}$ similarly. Finally, we use $\bW$ to denote $(\bW_1,\ldots,\bW_k)$ and use $\hat{\bW}$ similarly. We then have that
\begin{align}
\Pr\{\hat{\bW}\neq \bW\}
\nn&=\Pr\{\hat{\bW}\neq \bW~\mathrm{and~}\exists~(i,j)\in[k]\times[d]:~|\hatS_{i,j}-S_{i,j}|>\delta\}\\*
&\qquad+\Pr\{\hat{\bW}\neq \bW~\mathrm{and~}\forall~(i,j)\in[k]\times[d]:~|\hatS_{i,j}-S_{i,j}|\leq\delta\}\\
&\leq \Pr\{\exists~(i,j)\in[k]\times[d]:~|\hatS_{i,j}-S_{i,j}|>\delta\}+
\Pr\{\hat{\bW}\neq \bW~\mathrm{and~}\forall~(i,j)\in[k]\times[d]:~|\hatS_{i,j}-S_{i,j}|\leq\delta\}\\
&\leq \varepsilon+\Pr\{\hat{\bW}\neq \bW~\mathrm{and~}\forall~(i,j)\in[k]\times[d]:~|\hatS_{i,j}-S_{i,j}|\leq\delta\}\label{use000}\\
&\leq \varepsilon+\sum_{i\in[k]}\Pr\{\hat{\bW}_i\neq \bW_i~\mathrm{and}~\forall~j\in[d]:~|\hatS_{i,j}-S_{i,j}|\leq\delta\}\\
&\leq \varepsilon+\sum_{i\in[k]}2d\delta\tilM\label{simitoonetarget}\\
&\leq\varepsilon+2kd\beta\label{usetilMagain},
\end{align}
where \eqref{use000} follows from \eqref{ktarget:step00}, \eqref{simitoonetarget} follows similarly to \eqref{boundaryerror}, and \eqref{usetilMagain} follows from the definition of $\tilM$.

Note that since for each $i\in[k]$, $\bS_i=(S_{i,1},\ldots,S_{i,d})$ is uniformly distributed over the unit cube of dimension $d$, each quantized version $W_{i,j}$ is uniformly distributed over the message set $[\tilM]$. Furthermore, from the problem structure, we have that $\bW-\bS^k-X^n-Y^n-\hat{\bS}^k-\hat{\bW}$ forms a Markov chain. In particular, $X^n$ can be understood as a combination of the inputs from all $k$ targets, i.e., if we let $\bZ:=(Z^n(1),\ldots,Z^n(k))$ be a sequence of $k$ i.i.d. random variables from a Bernoulli distribution with parameter $\mathrm{Vol}(\calA_t)$, then $X^n$ is the output of the OR result of $\bZ$, i.e.,
\begin{align}
X_t=\bbo(\exists~i\in[k]:~Z_t(i)=1).
\end{align}
Therefore, using \eqref{usetilMagain}, we have shown that the excess-resolution probability of simultaneous searching for $k$ targets is lower bounded by the error probability of decoding $k$ messages over a OR type multiple access channel~\cite{6970834} and a small penalty term, i.e.,
\begin{align}
\varepsilon\geq \Pr\{\hat{\bW}\neq \bW\}-2kd\beta.
\end{align}

Similarly to \cite[Lemma 4]{han2003information}, we conclude that for any $\gamma\in\bbR_+$,
\begin{align}
\Pr\{\hat{\bW}\neq \bW\}
\geq \max_{\calA^n}\Pr\{(Z^n(1)\ldots,Z^n(k),Y^n)\notin\calD^{n,k}(-\gamma)\}-k\exp(-\gamma),
\end{align}
where the maximization over queries $\calA^n$ is implicitly included in the distribution of i.i.d. sequences $\bZ=(Z^n(1),\ldots,Z^n(k))$.

The rest of the proof is omitted since it proceeds similarly to the achievability proof.
}

\subsection{Proof of Second-Order Achievable Asymptotics for Adaptive Querying (Theorem \ref{second:fbl:adaptive})}
\label{proof:second:fbl:adaptive}

\subsubsection{An Non-Asymptotic Achievability Bound}
In this subsection, we present an adaptive query procedure (cf. Algorithm \ref{procedure:adapt}) and analyze its non-asymptotic performance.

Let $\bX^\infty$ be a collection of $M^d$ random binary vectors $\{X^\infty(i_1,\ldots,i_d)\}_{(i_1,\ldots,i_d)\in[M]^d}$, each with infinite length and let $\bx^\infty$ denote a realization of $\bX^\infty$. Furthermore, let $Y^\infty$ be another random vector with infinite length where each element takes values in $\calY$ and let $y^\infty$ be a realization of $Y^\infty$. For any $\bw\in[0,1]^d$ and any $n\in\bbN$, given any sequence of queries $\calA^n=(\calA_1,\ldots,\calA_n)\in[0,1]^d$, define the following joint distribution of $(\bX^n,Y^n)$
\begin{align}
P_{\bX^nY^n}^{\calA^n,\bw}(\bx^n,y^n)
&=\prod_{t\in[n]}\Big(\prod_{(i_1,\ldots,i_d)\in[M]^d}\mathrm{Bern}_p(x_t(i_1,\ldots,i_d))\Big)P_{Y|X}^{\calA_t}(y_t|x_t(\bw))\label{def:pxyan}.
\end{align}
We can define $P_{\bX^\infty,Y^\infty}^{\calA^n,\bw}$ as a generalization of $P_{\bX^n,Y^n}^{\calA^n,\bw}$ with $n$ replaced by $\infty$. Since the channel is memoryless, such a generalization is reasonable.

Recall the definition of $\Gamma(\cdot)$ in \eqref{def:Gamma} and its inverse $\Gamma^{-1}(\cdot)$.  Given any $\lambda\in\bbR_+$ and any $m\in[M^d]$, define the stopping time
\begin{align}
\tau_m(\bx^\infty,y^\infty)&:=\inf\{n\in\bbN:~\imath_q(x^n(\Gamma^{-1}(m));y^n)\geq \lambda\}\label{def:taum}.
\end{align}
Our non-asymptotic bound states as follows.
\begin{theorem}
\label{fbl:ach:adaptive}
Given any $(d,M)\in\bbR_+\times\bbN$, for any $p\in[0,1]$ and $\lambda\in\bbR_+$, there exists an $(l,d,\frac{1}{M},\varepsilon)$-adaptive query procedure such that
\begin{align}
l&\leq \mathbb{E}[\tau_1(\bX^\infty,Y^\infty)],\\
\varepsilon&\leq(M^d-1)\Pr\{\tau_1(\bX^\infty,Y^\infty)\geq \tau_2(\bX^\infty,Y^\infty)\},
\end{align}
where the expectation and probability are calculated with respect to $P_{\bX^\infty,Y^\infty}^{\calA^\infty,\Gamma^{-1}(1)}$ and $\calA^n$ refers to the queries in Algorithm \ref{procedure:adapt}.
\end{theorem}

\begin{algorithm}
\caption{\color{blue}Adaptive query procedure}
\label{procedure:adapt}
\begin{algorithmic}
\REQUIRE Three parameters $(M,p,\lambda)\in\bbN\times(0,1)\times\bbR_+$
\ENSURE An estimate $(\hats_1,\ldots,\hats_d)\in[0,1]^d$ of a $d$-dimensional target variable $(s_1,\ldots,s_d)\in[0,1]^d$
\STATE Partition the unit cube of dimension $d$ (i.e., $[0,1]^d$) into $M^d$ equal-sized disjoint regions $\{\calS_{i_1,\ldots,i_d}\}_{(i_1,\ldots,i_d)\in[M]^d}$.
\STATE $t \leftarrow 1$
\WHILE {$t>0$}
\STATE Generate $M^d$ binary random variables $\{x_t(i_1,\ldots,i_d)\}_{(i_1,\ldots,i_d)\in[M]^d}$ independently from a Bernoulli distribution with parameter $p$.
\STATE Form the $t$-th query as
\begin{align*}
\calA_t:=\bigcup_{(i_1,\ldots,i_d)\in[M]^d:x_t(i_1,\ldots,i_d)=1}\calS_{i_1,\ldots,i_d}.
\end{align*}
\STATE Obtain the noisy response $y_t$ from the oracle to the query $\calA_t$.
\STATE Calculate accumulated mutual information densities $\imath_q(x^t(i_1,\ldots,i_d);y^t)$ for all $(i_1,\ldots,i_d)\in[M]^d$.
\IF {$\max_{(i_1,\ldots,i_d)\in[M]^d}\imath_q(x^t(i_1,\ldots,i_d);y^t)\geq \lambda$}
\STATE $\tau\leftarrow t$.
\STATE $t\leftarrow 0$.
\ELSE
\STATE $t\leftarrow t+1$.
\ENDIF
\ENDWHILE
\STATE Generate estimates $(\hats_1,\ldots,\hats_d)$ as
\begin{align*}
\hats_i=\frac{2\hatw_i-1}{2M},
\end{align*}
where $\hat{\bw}=(\hatw_1,\ldots,\hatw_d)$ is obtained as follows:
\begin{align*}
\hat{\bw}=\Gamma^{-1}(\hatt),~\hatt=\max\{t\in[M^d]:\imath_q(x^\tau(\Gamma^{-1}(t));y^\tau)\geq \lambda\}.
\end{align*}
\end{algorithmic}
\end{algorithm}

\begin{proof}[Proof of Theorem \ref{fbl:ach:adaptive}]
The proof of Theorem \ref{fbl:ach:adaptive} is inspired by \cite[Theorem 3]{polyanskiy2011feedback} and is largely similar to the proof for non-adaptive query procedures in Appendix \ref{proof:ach}. Thus, we only emphasize the differences here.

The adaptive query procedure we analyze is summarized in Algorithm \ref{procedure:adapt} and briefly rephrased as follows. Let $\bx=\{x^\infty(i_1,\ldots,i_d)\}_{(i_1,\ldots,i_d)\in[M]^d}$ be a sequence of $M^d$ binary codewords with infinite length. Then for any $n\in\bbN$ and any $(i_1,\ldots,i_d)\in[M]^d$, let $X^n(i_1,\ldots,i_d)$ be the first $n$ elements of $X^\infty(i_1,\ldots,i_d)$. Similarly to the proof of Theorem \ref{ach:fbl} in Appendix \ref{proof:ach}, we use the query $\calA_t$ as in \eqref{def:query:ddim} and apply the quantization function $\rmq(\cdot)$ in \eqref{def:qs} to generated quantized targets $\bw=(w_1,\ldots,w_d)$, i.e., $w_i=\rmq(s_i)$ for each $i\in[d]$ given any target variable $\bs=(s_1,\ldots,s_d)\in[0,1]^d$. The noiseless response to the query $\calA_t$ is then $X_t(\bw)$ and the noisy response $y_t$ is obtained by passing $x_t$ through the measurement-dependent channel $P_{Y|X}^{\calA_t}$.

The decoding process is summarized as follows, which includes the design of the stopping time and decoding function. Let $\lambda\in\bbR_+$ be a fixed threshold. Recall the definitions of $\Gamma(\cdot)$ in \eqref{def:Gamma} and $\tau_m(\bx^\infty,y^\infty)$ in \eqref{def:taum}. For any $(M,d)\in\bbN^2$, the stopping time is chosen as
\begin{align}
\tau^*(\bx^\infty,y^\infty):=\min_{m\in[M^d]}\tau_m(\bx^\infty,y^\infty).
\end{align}
The decoder outputs estimates $\hat{\bS}=(\hatS_1,\ldots,\hatS_d)$ via the following  two-stage decoding
\begin{enumerate}
\item the decoder first generates estimates $\hat{\bW}=(\hatW_1,\ldots,\hatW_2)$ as follows:
\begin{align}
\hat{\bW}=\Gamma^{-1}(\hatt),~\max\{t\in[M^d]:\tau_j(\bx^\infty,y^\infty)=\tau^*(\bx^\infty,y^\infty)\},
\end{align}
\item the decoder produces estimates $\hat{\bS}=(\hatS_1,\ldots,\hatS_d)$ as
\begin{align}
\hatS_i=\frac{2\hatW_i-1}{2M},~i\in[d].
\end{align}
\end{enumerate}

Using the adaptive query procedure in Algorithm \ref{procedure:adapt}, we have that the average stopping time satisfies
\begin{align}
\sup_{f_{\bS}\in\calF([0,1]^d)}\mathbb{E}[\tau^*(\bx^\infty,Y^\infty)]
&=\sup_{f_{\bS}\in\calF([0,1]^d)}\int_{\bs\in[0,1]^d}f_{\bS}(\bs)\mathbb{E}[\tau^*(\bx^\infty,Y^\infty)|\bS=\bs]\\
&=\sup_{P_{\bW}\in\calP([M]^d)}\sum_{\bw\in[M]^d}P_{\bW}(\bw)\mathbb{E}[\tau^*(\bx^\infty,Y^\infty)|\bW=\bw]\\
&\leq \sup_{P_{\bW}\in\calP([M]^d)}\sum_{\bw\in[M]^d}P_{\bW}(\bw)\mathbb{E}[\tau_{\Gamma(\bw)}(\bx^\infty,Y^\infty)|\bW=\bw]\label{delay},
\end{align}
and the excess-resolution probability with respect to the resolution $\delta=\frac{1}{M}$ satisfies
\begin{align}
\nn&\sup_{f_{\bS}\in\calF([0,1]^d)}\Pr\{\exists~i\in[d],~|\hatS_i-S_i|>\delta\}\\*
&\leq \sup_{P_{\bW}\in\calP([M]^d)}\Pr\{\hat{\bW}\neq\bW\}\\
&\leq \sup_{P_{\bW}\in\calP([M]^d)}\sum_{\bw\in[M]^d}P_\bW(\bw)\Pr\{\tau_{\Gamma(\bw)}(\bx^\infty,Y^\infty)\geq \tau^*(\bx^\infty,Y^\infty)\}\label{errorp}.
\end{align}
In the following, we will show that there exists binary codewords $\bx^{\infty}$ such that the results in \eqref{delay} and \eqref{errorp} are upper bounded by the desired bounds in Theorem \ref{ach:fbl}.

Let $\bX^{\infty}:=\{X^\infty(i_1,\ldots,i_d)\}_{(i_1,\ldots,i_d)\in[M]^d}$ be a sequence of $M^d$ binary codewords with infinite length where each codeword is generated i.i.d. from the Bernoulli distribution $P_X$ with parameter $p\in(0,1)$. For any $\bw\in[0,1]^d$ and any $n\in\bbN$, using our adaptive query procedure, the joint distribution of $(\bX^n,Y^n)$ is $P_{\bX^n,Y^n}^{\calA^n,\bw}(\bx^n,y^n)$ as defined in \eqref{def:pxyan}. 

\blue{For any $P_{\bW}\in\calP([M]^d)$, we have
\begin{align}
\mathbb{E}_{\bX^\infty}[\tau^*(\bX^\infty,Y^\infty)]
&=\sum_{\bw\in[M]^d}P_{\bW}(\bw)\mathbb{E}_{P_{\bX^\infty,\bY^\infty}^{\calA^n,\bw}}[\tau^*(\bX^\infty,Y^\infty)]\\
&\leq \sum_{\bw\in[M]^d}P_{\bW}(\bw)\mathbb{E}_{P_{\bX^\infty,\bY^\infty}^{\calA^n,\bw}}[\tau_{\Gamma(\bw)}(\bX^\infty,Y^\infty)]\\
&=\sum_{\bw\in[M]^d}P_{\bW}(\bw)\mathbb{E}_{P_{\bX^\infty,\bY^\infty}^{\calA^n,\Gamma^{-1}(1)}}[\tau_1(\bX^\infty,Y^\infty)]\label{symmetry}\\
&=\mathbb{E}_{P_{\bX^\infty,\bY^\infty}^{\calA^n,\Gamma^{-1}(1)}}[\tau_1(\bX^\infty,Y^\infty)],
\end{align}
where \eqref{symmetry} follows since for each $\bw\in[M]^d$, from the definition of $\tau_{\cdot}(\cdot)$ in \eqref{def:taum}, 
\begin{align}
\mathbb{E}_{P_{\bX^\infty,\bY^\infty}^{\calA^n,\bw}}[\tau_{\Gamma(\bw)}(\bX^\infty,Y^\infty)]=\mathbb{E}_{P_{\bX^\infty,\bY^\infty}^{\calA^n,\mathrm{ones}(d)}}[\tau_{\Gamma(\mathrm{ones}(d))}(\bX^\infty,Y^\infty)]=\mathbb{E}_{P_{\bX^\infty,\bY^\infty}^{\calA^n,\Gamma^{-1}(1)}}[\tau_1(\bX^\infty,Y^\infty)],
\end{align}
and we use $\mathrm{ones}(d)$ to denote the all one vector with length $d$.
}

Similarly, we have
\begin{align}
\mathbb{E}_{\bX^{\infty}}[\Pr[\hat{\bW}\neq \bW]]
&\leq \sum_{\bw\in[M]^d}P_{\bW}(\bw)\Pr_{P_{\bX^\infty,\bY^\infty}^{\calA^n,\bw}}\{\tau_{\Gamma(\bw)}(\bX^\infty,Y^\infty)\geq \tau^*(\bX^\infty,Y^\infty)\}\\
&=\sum_{\bw\in[M]^d}P_{\bW}(\bw)\Pr_{P_{\bX^\infty,\bY^\infty}^{\calA^n,\Gamma^{-1}(1)}}\{\tau_1(\bX^\infty,Y^\infty)\geq \tau^*(\bX^\infty,Y^\infty)\}\label{usesymmetry11}\\
&=\Pr_{P_{\bX^\infty,\bY^\infty}^{\calA^n,\Gamma^{-1}(1)}}\{\tau_1(\bX^\infty,Y^\infty)\geq \tau^*(\bX^\infty,Y^\infty)\}\\
&\leq (M^d-1)\Pr_{P_{\bX^\infty,\bY^\infty}^{\calA^n,\Gamma^{-1}(1)}}\{\tau_1(\bX^\infty,Y^\infty)\geq \tau_2(\bX^\infty,Y^\infty)\}\label{usesymmetry22},
\end{align}
where \eqref{usesymmetry11} follows from the symmetry which implies that $\Pr\{\tau_{\Gamma(\bw)}(\bX^\infty,Y^\infty)\geq \tau^*(\bX^\infty,Y^\infty)\}=\Pr\{\tau_1(\bX^\infty,Y^\infty)\geq \tau^*(\bX^\infty,Y^\infty)\}$ for any $\bw\in[M]^d$ and \eqref{usesymmetry22} follows from the union bound and the symmetry similar to \eqref{usesymmetry11}.

The proof of Theorem \ref{fbl:ach:adaptive} is completed by using the simple fact that $\mathbb{E}[X]\leq a$ implies that there exists $x\leq a$ for any random variable $X$ and constant $a\in\bbR$.
\end{proof}

\subsubsection{Proof of Achievable Second-Order Asymptotics}
\blue{
The proof of second-order asymptotics for adaptive querying proceeds similarly as \cite{polyanskiy2011feedback} and we only highlight the differences here.} Let $q^*\in\calP_{\rm{ca}}$ be a capacity-achieving parameter for measurement-dependent channels $\{P_{Y|X}^q\}_{q\in[0,1]}$. From Theorem \ref{fbl:ach:adaptive}, we have that there exists an $(l,d,\frac{1}{M},\varepsilon)$-adaptive query procedure such that 
\begin{align}
l&\leq \mathbb{E}[\tau_1(\bX^\infty,Y^\infty)],\\
\varepsilon&\leq(M^d-1)\Pr\{\tau_1(\bX^\infty,Y^\infty)\geq \tau_2(\bX^\infty,Y^\infty)\}\label{touppe}.
\end{align}
Unless otherwise stated, the expectation and probability are calculated with respect to slight generalization of the joint distribution $P_{\bX^n Y^n}^{\calA^n,\Gamma^{-1}(1)}$ in \eqref{def:pxyan}.

For subsequent analyses, let $P_X$ be the Bernoulli distribution with parameter $q^*$ and let $\tilP_{XY}$ be the following joint distribution
\begin{align}
\tilP_{XY}(x,y)
&:=\sum_{\barx_1,\ldots,\barx^{M^d-1}}P_X(x)\bigg(\prod_{j\in[M^d-1]}P_X(\barx_j)\bigg)P_{Y|X}^{\frac{x+\sum_{j\in[M^d-1]}\barx_j}{M^d}}(y|x).
\end{align} 
Note that $\tilP_{XY}$ is the marginal distribution of $(X_i(\Gamma^{-1}(1)),Y_i)$ for each $i\in[n]$ induced from $P_{\bX^n\bY^n}^{\calA^n,\Gamma^{-1}(1)}$ under our query procedure.

Furthermore, define the ``mismatched'' version of the capacity.
\begin{align}
C_1&:=\mathbb{E}_{\tilP_{XY}}[\imath_{q^*}(X;Y)]\label{def:C1}.
\end{align}
Finally, for each $n\in\bbN$, let
\begin{align}
U_n&:=\imath_{q^*}(X^n;Y^n)=\sum_{i\in[n]}\imath_{q^*,q^*}(X_i;Y_i).
\end{align}
It can be easily verified that $\{U_n-nC_1\}_{n\in\bbN}$ is a martingale and for each $n\in\bbN$, $\mathbb{E}[U_n-nC_1]=0$. The optional stopping theorem~\cite[Theorem 10.10]{williams1991probability} implies that
\begin{align}
0&=\mathbb{E}[U_{\tau_1(\bX^\infty,\bY^\infty)}-C_1\tau_1(\bX^\infty,\bY^\infty)]\\
&\leq \lambda+a_0-C_1\mathbb{E}[\tau_1(\bX^\infty,\bY^\infty)]]\label{uppl2},
\end{align}
where $a_0$ is an upper bound on the information density $U_1$.
Thus,
\begin{align}
\mathbb{E}[\tau_1(\bX^\infty,\bY^\infty)]
&\leq \frac{\lambda+a_0}{C_1}\label{uppl3}.
\end{align}
We then focus on upper bounding \eqref{touppe}. From \eqref{uppl3}, we have that
\begin{align}
\Pr\{\tau_1(\bX^\infty,\bY^\infty)<\infty\}=1,
\end{align}
since otherwise the expectation value of $\tau_1(\bX^\infty,\bY^\infty)$ would be infinity.

Recall the definition of the typical set $\calT(\cdot)$ in \eqref{def:typical}. For any $\eta\in\bbR_+$, we have
\begin{align}
\nn&\Pr\{\tau_1(\bX^\infty,Y^\infty)\geq \tau_2(\bX^\infty,Y^\infty)\}\\*
&\leq \Pr\{\tau_2(\bX^\infty,Y^\infty)<\infty\}\\
&=\sum_{t\in\bbN}\bbo(t<\infty)\Pr\{\tau_2(\bX^\infty,Y^\infty)=t\}\\
&=\sum_{t\in\bbN}\bbo(t<\infty)\Big\{\Pr\{\tau_2(\bX^\infty,Y^\infty)=t,\bX^t\in\calT^t(M,d,q^*,\eta)\}+\Pr\{\bX^t\notin\calT^t(M,d,q^*,\eta)\}\Big\}\\
&\leq\sum_{t\in\bbN}\bbo(t<\infty)\Big\{\exp(t\eta c(p))\Pr_{P_{\bX^\infty,Y^\infty}^{q^*,\Gamma^{-1}(1)}}\{\tau_2(\bX^\infty,Y^\infty)=t\}+4t\exp(-2M^d\eta^2)\Big\}\label{usetypicalha},
\end{align} 
where \eqref{usetypicalha} follows from \eqref{fromassumption} and the upper bound on the probability of atypicality similar to \eqref{useatypical} and in \eqref{usetypicalha}, we use the change-of-measure technique and the distribution $P_{\bX^\infty,Y^\infty}^{q^*,\Gamma^{-1}(1)}$ is a generalization of $P_{\bX Y^n}^{p,\bw}(\cdot)$ in \eqref{altdis} to an infinite length.

Given any $l'\in\bbR_+$, let $(\lambda,\eta,M)\in\bbR_+^2\times\bbN$ be chosen so that
\begin{align}
\lambda&=l'C_1-a_0,\\
d\log M&=\lambda-\log l'\label{herechooseMhehe},\\
\eta&:=\sqrt{\frac{d\log M}{2M^d}}=O\left(\frac{\sqrt{l'}}{\exp(l'C_1/2)}\right).
\end{align}
Then from \eqref{uppl3}, we have
\begin{align}
\mathbb{E}[\tau_1(\bX^\infty,Y^\infty)]\leq l'.
\end{align}
Furthermore, similarly to \cite[Section D]{polyanskiy2011feedback}, we have
\begin{align}
\nn&\Pr\{\tau_1(\bX^\infty,Y^\infty)\geq \tau_2(\bX^\infty,Y^\infty)\}\\*
&=\sum_{t\in\bbN}\bbo(t<\infty)\left(1+O\left(l'^{\frac{3}{2}}\exp\left(-\frac{l'C_1}{2}\right)\right)\right)\Pr_{P_{\bX^\infty,Y^\infty}^{q^*,\Gamma^{-1}(1)}}\{\tau_2(\bX^\infty,Y^\infty)=t\}\\
&=\left(1+O\left(l'^{\frac{3}{2}}\exp\left(-\frac{l'C_1}{2}\right)\right)\right)\lim_{t\to\infty}\Pr_{P_{\bX^\infty,Y^\infty}^{q^*,\Gamma^{-1}(1)}}\{\tau_2(\bX^\infty,Y^\infty)<t\}\\
&=\left(1+O\left(l'^{\frac{3}{2}}\exp\left(-\frac{l'C_1}{2}\right)\right)\right)\lim_{t\to\infty}\mathbb{E}_{P_{\bX^\infty,Y^\infty}^{q^*,\Gamma^{-1}(1)}}\big[\exp(-U_t)\bbo(\tau_1(\bX^\infty,Y^\infty)<t)\big]\label{applycof}\\
&=\left(1+O\left(l'^{\frac{3}{2}}\exp\left(-\frac{l'C_1}{2}\right)\right)\right)\mathbb{E}_{P_{\bX^\infty,Y^\infty}^{q^*,\Gamma^{-1}(1)}}\big[\exp(-U_{\tau_1(\bX^\infty,Y^\infty)})\bbo(\tau_1(\bX^\infty,Y^\infty)<\infty)\big]\label{followfblwfb}\\
&\leq \left(1+O\left(l'^{\frac{3}{2}}\exp\left(-\frac{l'C_1}{2}\right)\right)\right)\exp(-\lambda)\label{def:tau1},
\end{align}
where \eqref{applycof} follows from the change-of-measure technique, \eqref{def:tau1} follows from the definition of $\tau_1(\bX^\infty,Y^\infty)$ in \eqref{def:taum} and \eqref{followfblwfb} follows similarly as \cite[Eq. (113) to Eq. (117)]{polyanskiy2011feedback} \blue{and the details are as follows:
\begin{align}
\nn&\lim_{t\to\infty}\mathbb{E}_{P_{\bX^\infty,Y^\infty}^{\mathrm{md},1}}\big[\exp(-U_t)\bbo(\tau_1(\bX^\infty,Y^\infty)<t)\big]\\*
&=\lim_{t\to\infty}\mathbb{E}_{P_{\bX^\infty,Y^\infty}^{\mathrm{md},1}}\big[\exp(-U_t)\bbo(\tau_t(\bX^\infty,Y^\infty)<t)\big]\label{def:taut}\\
&=\lim_{t\to\infty}\mathbb{E}_{P_{\bX^\infty,Y^\infty}^{\mathrm{md},1}}\big[\exp(-U_{\tau_t(\bX^\infty,Y^\infty)})\bbo(\tau_t(\bX^\infty,Y^\infty)<t)\big]\label{useost4mar}\\
&=\mathbb{E}_{P_{\bX^\infty,Y^\infty}^{\mathrm{md},1}}\big[\lim_{t\to\infty}\exp(-U_{\tau_t(\bX^\infty,Y^\infty)})\bbo(\tau_t(\bX^\infty,Y^\infty)<t)\big]\\
&=\mathbb{E}_{P_{\bX^\infty,Y^\infty}^{\mathrm{md},1}}\big[\exp(-U_{\tau_1(\bX^\infty,Y^\infty)})\bbo(\tau_1(\bX^\infty,Y^\infty)<\infty)\big]
\end{align}
where in \eqref{def:taut}, we define $\tau_t(\bX^\infty,Y^\infty):=\min\{\tau_1(\bX^\infty,Y^\infty),t\}$, \eqref{useost4mar} follows by applying the optional stopping theorem~\cite[Theorem 10.10]{williams1991probability} to the martingale $\exp(-U_t)$ and the stopping time $\tau_t(\bX^\infty,Y^\infty)$.
}

Therefore, we have for $l'$ sufficient large,
\begin{align}
\nn&(M^d-1)\Pr_{P_{\bX^\infty,Y^\infty}^1}\{\tau_1(\bX^\infty,Y^\infty)\geq \tau_2(\bX^\infty,Y^\infty)\}\\*
&\leq \left(1+O\left(l'^{\frac{3}{2}}\exp\left(-\frac{l'C_1}{2}\right)\right)\right)M^d\exp(-\lambda)\\
&=\left(1+O\left(l'^{\frac{3}{2}}\exp\left(-\frac{l'C_1}{2}\right)\right)\right)\frac{1}{l'}\label{usechooseMhehe},
\end{align}
where \eqref{usechooseMhehe} follows from the choice of $M$ in \eqref{herechooseMhehe}.

Recall the definition of the ``capacity'' $C$ of measurement-dependent channels in \eqref{def:capacity}. Using the definition of $C_1$ in \eqref{def:C1}, we have that
\begin{align}
C_1
&=\mathbb{E}_{\barP_{XY}}[\imath_{q^*}(X;Y)]\\
&=\mathbb{E}_{P_{XY}}\bigg[\frac{\barP_{XY}(X,Y)}{P_{XY}(X,Y)}\imath_{q^*}(X;Y)\bigg]\\
&\leq \exp(\eta c(p))\mathbb{E}_{P_{XY}}[\imath_{q^*}(X;Y)]+2\exp(-2M^d\eta^2)\label{changeofmeasureagain2}\\
&=\exp(\eta c(p))C+2\exp(-2M^d\eta^2).
\end{align}
where \eqref{changeofmeasureagain2} follows from the change-of-measure technique and the result in \eqref{fromassumption}. Given the choice of $M$ and $\eta$, we have
\begin{align}
C_1=C+O(l'\exp(-l')).
\end{align}

Thus, till now, we have constructed an $(l',d,\frac{1}{l'},\frac{1}{M})$-adaptive query procedure for sufficiently large $l'$. For any $\varepsilon\in[0,1)$, consider the following query procedure: with probability $\frac{l'\varepsilon-1}{l'-1}$, we do not pose any query and with the remaining probability, we use the above-constructed $(l',d,\frac{1}{l'},\frac{1}{M})$-adaptive query procedure. For $l'$ sufficiently large, it is easy to verify that the combined adaptive query procedure is an $(l,d,\varepsilon,\delta)$-adaptive query procedure
where
\begin{align}
l
&=\left(1-\frac{\varepsilon l'-1}{l'-1}\right)l'=\frac{l'^2(1-\varepsilon)}{l'-1}=(1-\varepsilon)l'+(1-\varepsilon)(1-\frac{1}{l'-1})\approx (1-\varepsilon)l',\\
-d\log \delta
&=l'C_1-a_0-\log l'=l' C+O(\log l')=\frac{Cl}{1-\varepsilon}+O(\log l).
\end{align}

\subsection{Comparison the Performances of Adaptive Querying for Measurement-Dependent and Measurement-Independent Channels}
\label{supp:adap}

To compare the performances of optimal adaptive query procedures under measurement-dependent and measurement-independent channels, \blue{we plot in Figure \ref{com_fbl_adap} the second-order asymptotic approximation to the number of bits (in the binary expansion of the target random variable $S$) extracted} after $n$ queries, i.e., $-\log_2 \delta_\rma^*(n,d,\varepsilon)$ and $-\log_2 \delta^*_{\rm{a},\rm{mi}}(n,d,\varepsilon)$ for $\varepsilon=0.001$ (the $O(\log n)$ term is ignored). Note that for measurement-dependent noise channel, we plot only a lower bound to the second-order approximation of $-\log_2\delta_\rma^*(n,d,\varepsilon)$. The remarks for non-adaptive querying are still valid for adaptive querying.

\begin{figure}[tb]
\centering
\begin{tabular}{cc}
\hspace{-.25in} \includegraphics[width=.5\columnwidth]{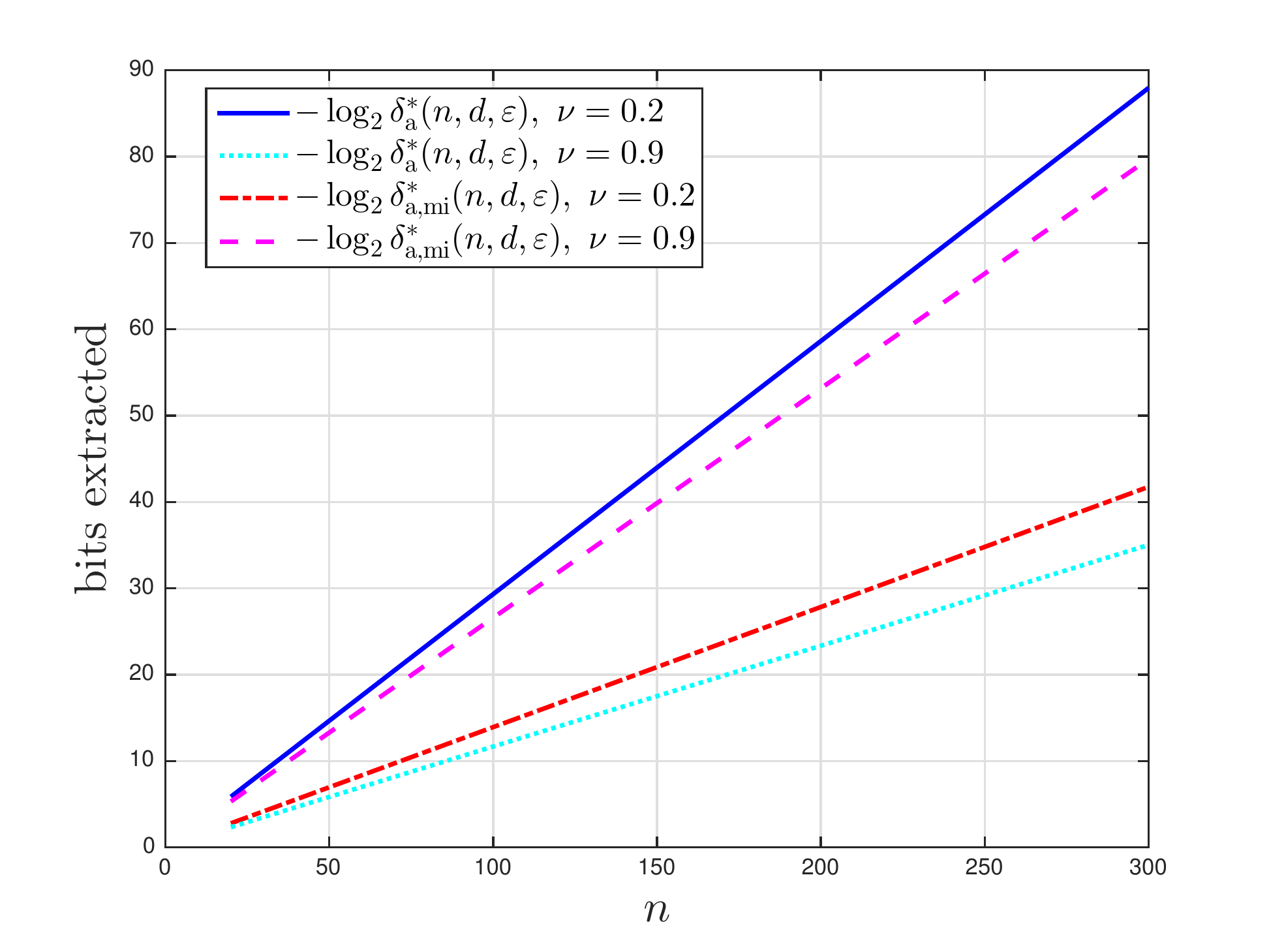}& \hspace{-.4in} \includegraphics[width=.5\columnwidth]{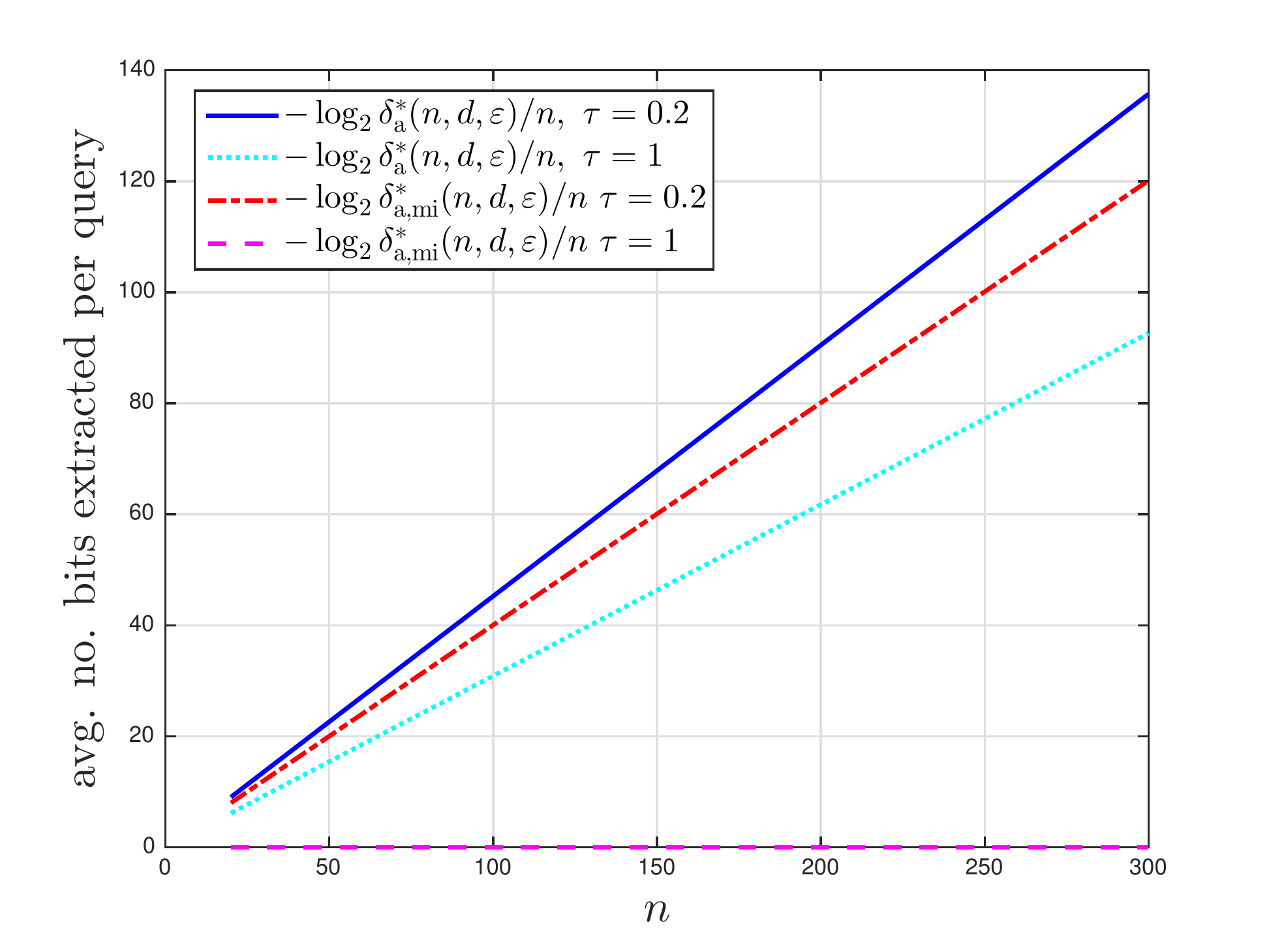}\\
\hspace{-.25in} {(a) BSC with parameter $\nu$} & \hspace{-.4in}  { (b) BEC with parameter  $\tau$}\\
\hspace{-.25in} \includegraphics[width=.5\columnwidth]{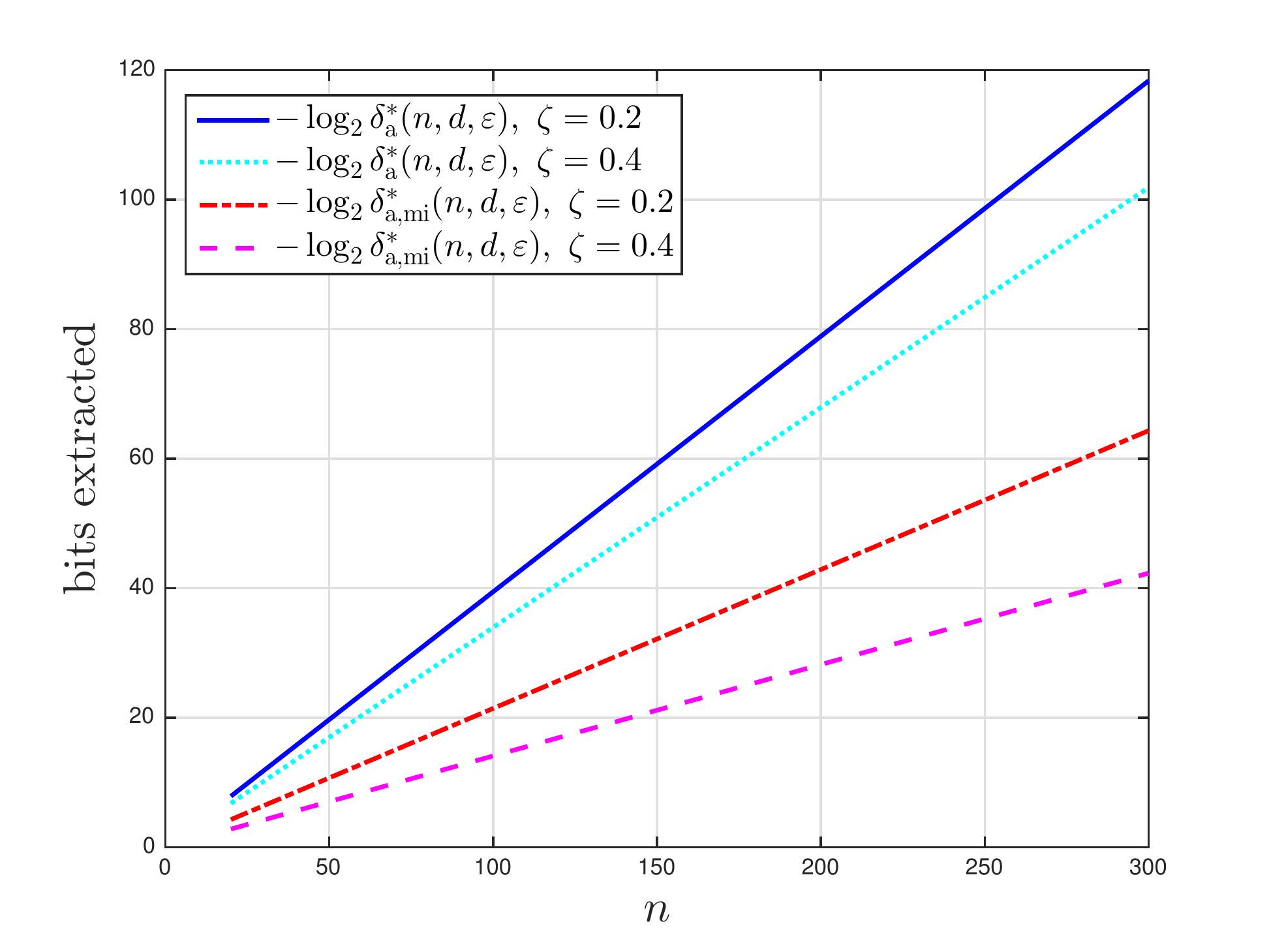}\\
\hspace{-.25in} {(c) Binary Z-channel with parameter $\zeta$}
\end{tabular}
\caption{\blue{The second-order asymptotic approximation to the number of bits extracted} by optimal adaptive query procedures for both measurement-dependent ($-\log_2\delta_\rma^*(n,d,\varepsilon)$) and measurement-independent ($-\log_2\delta_{\rma,\rm{mi}}^*(n,d,\varepsilon)$) versions of three channels where we neglect the $O(\log l)$ term. Note that for measurement-dependent channels, we plot only a lower bound (asserted in Theorem \ref{second:fbl:adaptive}) to \blue{the second-order approximation of $-\log_2\delta_\rma^*(n,d,\varepsilon)$}. In all plots, we consider the case of $d=2$ and $\varepsilon=0.001$.}
\label{com_fbl_adap}
\end{figure}

\section*{Acknowledgments}
The authors acknowledge anonymous reviewers for helpful comments and suggestions.

\bibliographystyle{IEEEtran}
\bibliography{IEEEfull_lin}

\end{document}

%% file: preamble.tex
\usepackage[mathscr]{eucal}
\usepackage[cmex10]{amsmath}
\usepackage{epsfig,epsf,psfrag}
\usepackage{amssymb,amsmath,amsthm,amsfonts,latexsym}
\usepackage{amsmath,graphicx,bm,xcolor,url,overpic}
\usepackage{fixltx2e}
\usepackage{array}
\usepackage{verbatim}
\usepackage{bm}
\usepackage{algorithmic}
\usepackage{algorithm}
\usepackage{verbatim}
\usepackage{textcomp}
\usepackage{mathrsfs}
\usepackage{epstopdf}

\newcommand{\openone}{\leavevmode\hbox{\small1\normalsize\kern-.33em1}}

\catcode`~=11 \def\UrlSpecials{\do\~{\kern -.15em\lower .7ex\hbox{~}\kern .04em}} \catcode`~=13 

\allowdisplaybreaks[4]

\newcommand{\nn}{\nonumber}

\newcommand{\calA}{\mathcal{A}}

\newcommand{\calD}{\mathcal{D}}
\newcommand{\calE}{\mathcal{E}}
\newcommand{\calF}{\mathcal{F}}

\newcommand{\calI}{\mathcal{I}}
\newcommand{\calJ}{\mathcal{J}}

\newcommand{\calL}{\mathcal{L}}

\newcommand{\calP}{\mathcal{P}}

\newcommand{\calS}{\mathcal{S}}
\newcommand{\calT}{\mathcal{T}}

\newcommand{\calW}{\mathcal{W}}
\newcommand{\calX}{\mathcal{X}}
\newcommand{\calY}{\mathcal{Y}}

\newcommand{\ba}{\mathbf{a}}

\newcommand{\bs}{\mathbf{s}}
\newcommand{\bS}{\mathbf{S}}

\newcommand{\bw}{\mathbf{w}}
\newcommand{\bW}{\mathbf{W}}
\newcommand{\bx}{\mathbf{x}}
\newcommand{\bX}{\mathbf{X}}

\newcommand{\bY}{\mathbf{Y}}

\newcommand{\bZ}{\mathbf{Z}}

\newcommand{\rma}{\mathrm{a}}

\newcommand{\rmb}{\mathrm{b}}

\newcommand{\rmd}{\mathrm{d}}

\newcommand{\rme}{\mathrm{e}}

\newcommand{\rmG}{\mathrm{G}}

\newcommand{\rmp}{\mathrm{p}}
\newcommand{\rmP}{\mathrm{P}}
\newcommand{\rmq}{\mathrm{q}}

\newcommand{\rmU}{\mathrm{U}}

\newcommand{\rmV}{\mathrm{V}}

\newcommand{\bbN}{\mathbb{N}}

\newcommand{\bbR}{\mathbb{R}}

\DeclareMathAlphabet{\mathbsf}{OT1}{cmss}{bx}{n}
\DeclareMathAlphabet{\mathssf}{OT1}{cmss}{m}{sl}

\DeclareSymbolFont{bsfletters}{OT1}{cmss}{bx}{n}  
\DeclareSymbolFont{ssfletters}{OT1}{cmss}{m}{n}
\DeclareMathSymbol{\bsfGamma}{0}{bsfletters}{'000}
\DeclareMathSymbol{\ssfGamma}{0}{ssfletters}{'000}
\DeclareMathSymbol{\bsfDelta}{0}{bsfletters}{'001}
\DeclareMathSymbol{\ssfDelta}{0}{ssfletters}{'001}
\DeclareMathSymbol{\bsfTheta}{0}{bsfletters}{'002}
\DeclareMathSymbol{\ssfTheta}{0}{ssfletters}{'002}
\DeclareMathSymbol{\bsfLambda}{0}{bsfletters}{'003}
\DeclareMathSymbol{\ssfLambda}{0}{ssfletters}{'003}
\DeclareMathSymbol{\bsfXi}{0}{bsfletters}{'004}
\DeclareMathSymbol{\ssfXi}{0}{ssfletters}{'004}
\DeclareMathSymbol{\bsfPi}{0}{bsfletters}{'005}
\DeclareMathSymbol{\ssfPi}{0}{ssfletters}{'005}
\DeclareMathSymbol{\bsfSigma}{0}{bsfletters}{'006}
\DeclareMathSymbol{\ssfSigma}{0}{ssfletters}{'006}
\DeclareMathSymbol{\bsfUpsilon}{0}{bsfletters}{'007}
\DeclareMathSymbol{\ssfUpsilon}{0}{ssfletters}{'007}
\DeclareMathSymbol{\bsfPhi}{0}{bsfletters}{'010}
\DeclareMathSymbol{\ssfPhi}{0}{ssfletters}{'010}
\DeclareMathSymbol{\bsfPsi}{0}{bsfletters}{'011}
\DeclareMathSymbol{\ssfPsi}{0}{ssfletters}{'011}
\DeclareMathSymbol{\bsfOmega}{0}{bsfletters}{'012}
\DeclareMathSymbol{\ssfOmega}{0}{ssfletters}{'012}

\newcommand{\tili}{\tilde{i}}

\newcommand{\tilM}{\tilde{M}}

\newcommand{\tilP}{\tilde{P}}

\newcommand{\hats}{\hat{s}}
\newcommand{\hatS}{\hat{S}}

\newcommand{\hatt}{\hat{t}}

\newcommand{\hatw}{\hat{w}}
\newcommand{\hatW}{\hat{W}}

\newcommand{\barx}{\bar{x}}

\newcommand{\barP}{\bar{P}}

\newcommand{\barX}{\bar{X}}

\DeclareMathOperator*{\argmax}{arg\,max}
\DeclareMathOperator*{\argmin}{arg\,min}

\newtheorem{theorem}{Theorem}

\newtheorem{corollary}{Corollary}

\newtheorem{definition}{Definition}

%% file: twenty_question_r5.bbl
\begin{thebibliography}{10}
\providecommand{\url}[1]{#1}
\csname url@samestyle\endcsname
\providecommand{\newblock}{\relax}
\providecommand{\bibinfo}[2]{#2}
\providecommand{\BIBentrySTDinterwordspacing}{\spaceskip=0pt\relax}
\providecommand{\BIBentryALTinterwordstretchfactor}{4}
\providecommand{\BIBentryALTinterwordspacing}{\spaceskip=\fontdimen2\font plus
\BIBentryALTinterwordstretchfactor\fontdimen3\font minus
  \fontdimen4\font\relax}
\providecommand{\BIBforeignlanguage}[2]{{%
\expandafter\ifx\csname l@#1\endcsname\relax
\typeout{** WARNING: IEEEtran.bst: No hyphenation pattern has been}%
\typeout{** loaded for the language `#1'. Using the pattern for}%
\typeout{** the default language instead.}%
\else
\language=\csname l@#1\endcsname
\fi
#2}}
\providecommand{\BIBdecl}{\relax}
\BIBdecl

\bibitem{renyi1961problem}
A.~R{\'e}nyi, ``On a problem of information theory,'' \emph{MTA Mat. Kut. Int.
  Kozl. B}, vol.~6, pp. 505--516, 1961.

\bibitem{burnashev1974interval}
M.~V. Burnashev and K.~Zigangirov, ``An interval estimation problem for
  controlled observations,'' \emph{Problemy Peredachi Informatsii}, vol.~10,
  no.~3, pp. 51--61, 1974.

\bibitem{ulam1991adventures}
S.~M. Ulam, \emph{Adventures of a Mathematician}.\hskip 1em plus 0.5em minus
  0.4em\relax Univ of California Press, 1991.

\bibitem{pelc2002searching}
A.~Pelc, ``Searching games with errors---fifty years of coping with liars,''
  \emph{Theoretical Computer Science}, vol. 270, no. 1-2, pp. 71--109, 2002.

\bibitem{jedynak2012twenty}
B.~Jedynak, P.~I. Frazier, and R.~Sznitman, ``Twenty questions with noise:
  Bayes optimal policies for entropy loss,'' \emph{Journal of Applied
  Probability}, vol.~49, no.~1, pp. 114--136, 2012.

\bibitem{chung2018unequal}
H.~W. Chung, B.~M. Sadler, L.~Zheng, and A.~O. Hero, ``Unequal error protection
  querying policies for the noisy 20 questions problem,'' \emph{IEEE Trans.
  Inf. Theory}, vol.~64, no.~2, pp. 1105--1131, 2018.

\bibitem{lalitha2018improved}
A.~Lalitha, N.~Ronquillo, and T.~Javidi, ``Improved target acquisition rates
  with feedback codes,'' \emph{IEEE Journal of Selected Topics in Signal
  Processing}, vol.~12, no.~5, pp. 871--885, 2018.

\bibitem{kaspi2018searching}
Y.~Kaspi, O.~Shayevitz, and T.~Javidi, ``Searching with measurement dependent
  noise,'' \emph{IEEE Trans. Inf. Theory}, vol.~64, no.~4, pp. 2690--2705,
  2018.

\bibitem{tsiligkaridis2014collaborative}
T.~Tsiligkaridis, B.~M. Sadler, and A.~O. Hero, ``Collaborative 20 questions
  for target localization,'' \emph{IEEE Trans. Inf. Theory}, vol.~60, no.~4,
  pp. 2233--2252, 2014.

\bibitem{tsiligkaridis2015decentralized}
------, ``On decentralized estimation with active queries,'' \emph{IEEE
  Transactions on Signal Processing}, vol.~63, no.~10, pp. 2610--2622, 2015.

\bibitem{chiu2016sequential}
S.-E. Chiu and T.~Javidi, ``Sequential measurement-dependent noisy search,'' in
  \emph{2016 IEEE Information Theory Workshop (ITW)}.\hskip 1em plus 0.5em
  minus 0.4em\relax IEEE, 2016, pp. 221--225.

\bibitem{chung2017bounds}
H.~W. Chung, B.~M. Sadler, and A.~O. Hero, ``Bounds on variance for unimodal
  distributions,'' \emph{IEEE Trans. Inf. Theory}, vol.~63, no.~11, pp.
  6936--6949, 2017.

\bibitem{polyanskiy2010finite}
Y.~Polyanskiy, H.~V. Poor, and S.~Verd\'u, ``Channel coding rate in the finite
  blocklength regime,'' \emph{IEEE Trans. Inf. Theory}, vol.~56, no.~5, pp.
  2307--2359, 2010.

\bibitem{TanBook}
V.~Y.~F. Tan, ``Asymptotic estimates in information theory with non-vanishing
  error probabilities,'' \emph{{Foundations and Trends$\,$\textregistered $ $
  in Communications and Information Theory}}, vol.~11, no. 1--2, pp. 1--184,
  2014.

\bibitem{csiszar2011information}
I.~Csisz\'ar and J.~K{\"o}rner, \emph{Information Theory: Coding Theorems for
  Discrete Memoryless Systems}.\hskip 1em plus 0.5em minus 0.4em\relax
  Cambridge University Press, 2011.

\bibitem{kaspi2015searching}
Y.~Kaspi, O.~Shayevitz, and T.~Javidi, ``Searching for multiple targets with
  measurement dependent noise,'' in \emph{IEEE ISIT}.\hskip 1em plus 0.5em
  minus 0.4em\relax IEEE, 2015, pp. 969--973.

\bibitem{forney1968exponential}
G.~Forney, ``Exponential error bounds for erasure, list, and decision feedback
  schemes,'' \emph{IEEE Trans. Inf. Theory}, vol.~14, no.~2, pp. 206--220,
  1968.

\bibitem{yamamoto1979asymptotic}
H.~Yamamoto and K.~Itoh, ``Asymptotic performance of a modified
  schalkwijk-barron scheme for channels with noiseless feedback (corresp.),''
  \emph{IEEE Trans. Inf. Theory}, vol.~25, no.~6, pp. 729--733, 1979.

\bibitem{polyanskiy2011feedback}
Y.~Polyanskiy, H.~V. Poor, and S.~Verd{\'u}, ``Feedback in the non-asymptotic
  regime,'' \emph{IEEE Trans. Inf. Theory}, vol.~57, no.~8, pp. 4903--4925,
  2011.

\bibitem{el2011network}
A.~El~Gamal and Y.-H. Kim, \emph{Network Information Theory}.\hskip 1em plus
  0.5em minus 0.4em\relax Cambridge University Press, 2011.

\bibitem{berger1971rate}
T.~Berger, \emph{Rate-Distortion Theory}.\hskip 1em plus 0.5em minus
  0.4em\relax Wiley Online Library, 1971.

\bibitem{kostina2013lossy}
V.~Kostina, ``Lossy data compression: Non-asymptotic fundamental limits,''
  Ph.D. dissertation, Department of Electrical Engineering, Princeton
  University, 2013.

\bibitem{scarlett2017nips}
J.~Scarlett and V.~Cevher, ``Phase transitions in the pooled data problem,'' in
  \emph{Advances in Neural Information Processing Systems 30}, I.~Guyon, U.~V.
  Luxburg, S.~Bengio, H.~Wallach, R.~Fergus, S.~Vishwanathan, and R.~Garnett,
  Eds.\hskip 1em plus 0.5em minus 0.4em\relax Curran Associates, Inc., 2017,
  pp. 377--385.

\bibitem{Scarlett:2016:PTG:2884435.2884439}
------, ``Phase transitions in group testing,'' in \emph{Proceedings of the
  Twenty-seventh Annual ACM-SIAM Symposium on Discrete Algorithms}, ser. SODA
  '16.\hskip 1em plus 0.5em minus 0.4em\relax Philadelphia, PA, USA: Society
  for Industrial and Applied Mathematics, 2016, pp. 40--53.

\bibitem{zhou2016cilossy}
L.~Zhou, V.~Y.~F. Tan, and M.~Motani, ``Exponential strong converse for content
  identification with lossy recovery,'' \emph{IEEE Trans. Inf. Theory},
  vol.~64, no.~8, pp. 5879---5897, 2018.

\bibitem{wei2009strong}
W.~Gu and M.~Effros, ``A strong converse for a collection of network source
  coding problems,'' in \emph{IEEE ISIT}, 2009, pp. 2316--2320.

\bibitem{liu2016brascamp}
J.~Liu, T.~A. Courtade, P.~Cuff, and S.~Verd{\'u}, ``Smoothing {Brascamp-Lieb}
  inequalities and strong converses for common randomness generation,'' in
  \emph{IEEE ISIT}, 2016, pp. 1043--1047.

\bibitem{han2003information}
T.~S. Han, \emph{Information-Spectrum Methods in Information Theory}.\hskip 1em
  plus 0.5em minus 0.4em\relax Springer Berlin Heidelberg, 2003.

\bibitem{han2006information}
------, ``An information-spectrum approach to capacity theorems for the general
  multiple-access channel,'' \emph{IEEE Trans. Inf. Theory}, vol.~44, no.~7,
  pp. 2773--2795, 2006.

\bibitem{horstein1963sequential}
M.~Horstein, ``Sequential transmission using noiseless feedback,'' \emph{IEEE
  Trans. Inf. Theory}, vol.~9, no.~3, pp. 136--143, 1963.

\bibitem{chiu2019noisy}
S.-E. Chiu, ``Noisy binary search: Practical algorithms and applications,''
  Ph.D. dissertation, UC San Diego, 2019.

\bibitem{6970834}
Y.~{Kaspi}, O.~{Shayevitz}, and T.~{Javidi}, ``Searching with measurement
  dependent noise,'' in \emph{IEEE ITW}, Nov 2014, pp. 267--271.

\bibitem{shannon1948mathematical}
C.~E. Shannon, ``A mathematical theory of communication,'' \emph{Bell Syst.
  Tech. J.}, vol.~27, no.~1, pp. 379--423, 1948.

\bibitem{gallager_ensemble}
R.~Gallager, ``The random coding bound is tight for the average code
  (corresp.),'' \emph{IEEE Trans. Inf. Theory}, vol.~19, no.~2, pp. 244--246,
  1973.

\bibitem{tan2014state}
M.~Tomamichel and V.~Y.~F. Tan, ``Second-order coding rates for channels with
  state,'' \emph{IEEE Trans. Inf. Theory}, vol.~60, no.~8, pp. 4427--4448,
  2014.

\bibitem{scarlett2017mismatch}
J.~Scarlett, V.~Y.~F. Tan, and G.~Durisi, ``The dispersion of nearest-neighbor
  decoding for additive {Non-Gaussian} channels,'' \emph{IEEE Trans. Inf.
  Theory}, vol.~63, no.~1, pp. 81--92, 2017.

\bibitem{berry1941accuracy}
A.~C. Berry, ``The accuracy of the {Gaussian} approximation to the sum of
  independent variates,'' \emph{Transactions of the {A}merican mathematical
  society}, vol.~49, no.~1, pp. 122--136, 1941.

\bibitem{esseen1942liapounoff}
C.-G. Esseen, \emph{On the {Liapounoff} limit of error in the theory of
  probability}.\hskip 1em plus 0.5em minus 0.4em\relax Almqvist \& Wiksell,
  1942.

\bibitem{dembo2009large}
A.~Dembo and O.~Zeitouni, \emph{Large Deviations Techniques and
  Applications}.\hskip 1em plus 0.5em minus 0.4em\relax Springer, 2009,
  vol.~38.

\bibitem{williams1991probability}
D.~Williams, \emph{Probability with martingales}.\hskip 1em plus 0.5em minus
  0.4em\relax Cambridge university press, 1991.

\end{thebibliography}
